\documentclass[a4paper]{llncs}

\usepackage{lmodern}
\usepackage[utf8]{inputenc}
\usepackage[T1]{fontenc}
\usepackage{amssymb}
\usepackage{amsmath}
\usepackage{graphicx}
\usepackage{multicol}
\usepackage{color}
\usepackage{hyperref}
\usepackage{listings}
\usepackage{fancyvrb}
\usepackage{mathtools}
\usepackage{algpseudocode}
\usepackage{semantic}
\usepackage{chngcntr}
\usepackage{url}
\usepackage{paralist}
\usepackage[absolute]{textpos}

\usepackage[papersize={15.17cm, 23.29cm}, textwidth=12.2cm,
 textheight=19.3cm, centering]{geometry} 

\newcommand{\keywords}[1]{\par\addvspace\baselineskip
\noindent\keywordname\enspace\ignorespaces#1}

\begin{document}
\setlength{\TPHorizModule}{\linewidth}
\setlength{\TPVertModule}{\TPHorizModule}
\textblockorigin{8.5mm}{3mm} 
\begin{textblock}{1}(0,0)
  \fbox{\parbox{\textwidth}{\scriptsize\copyright Springer-Verlag,
      2014.  This is the author's version of the work. It is posted
      here by permission of Springer-Verlag for  
      your personal use. Not for redistribution. The final 
      publication is published in the proceedings of the 3rd
      Conference on Principles of Security and Trust (POST '14),
      and is available at 
      \href{http://www.link.springer.com}{link.springer.com} 
}}
\end{textblock}

\mainmatter  

\title{Information Flow Control in WebKit's JavaScript Bytecode}

\titlerunning{Information Flow Control in WebKit's JavaScript Bytecode}

%
%
%
\author{Abhishek Bichhawat\inst{1} \and Vineet Rajani\inst{2} \and
  Deepak Garg\inst{2} \and Christian Hammer\inst{1}}

\authorrunning{}

\institute{Saarland University, Germany\and MPI-SWS, Germany}

%
%

\maketitle

\newcommand{\CH}[1]{\textcolor{red}{CH: #1}}

\newtheorem{mydef}{Definition}
\newtheorem{myLemma}{Lemma}
\newtheorem{myThm}{Theorem}
\newtheorem{myaxiom}{Axiom}

\newcommand{\TODO}[1]{\textcolor{red}{TODO: #1}}

\newcommand {\mV} {\mathcal{V}}
\newcommand {\mL} {\mathcal{L}}
\newcommand {\mH} {\mathcal{H}}
\newcommand {\mE} {\mathcal{E}}

\newcommand {\Succ} {\mathit{Succ}}
\newcommand {\Left} {\mathit{Left}}
\newcommand {\Right} {\mathit{Right}}
\newcommand {\src} {\mathit{src}}
\newcommand {\dst} {\mathit{dst}}
\newcommand {\isIPD} {\mathit{isIPD}}
\newcommand {\IPD} {\mathit{IPD}}
\newcommand {\ipd} {\mathit{ipd}}
\newcommand {\cf} {\mathit{CF}}
\newcommand {\cond} {\mathit{cond}}
\newcommand {\push} {\mathit{push}}
\newcommand {\false} {\mathit{false}}
\newcommand {\true} {\mathit{true}}
\newcommand {\val} {\mathit{value}}
\newcommand {\base} {\mathit{base}}
\newcommand {\Empty} {\mathit{empty}}
\newcommand {\eV} {\mathit{excValue}}
\newcommand {\SC} {\mathit{sc}}
\newcommand {\LO} {\Lambda}
\newcommand{\CFG}{\mathit{CFG}}

\newcommand {\veq}{\approx^{\beta}_L}
\newcommand {\sigeq}{\simeq^{\beta}_L}
\newcommand {\teq}{\approx^{\beta}_L}
\newcommand {\req}{\cong^{\beta}_L}
\newcommand {\sceq}{\sim^{\beta}_L}
\newcommand {\lsd}{\Gamma(!\sigma(\dst))}

\newcommand{\Figure}[1]{Fig.~\ref{#1}}


\begin{abstract}

Websites today routinely combine JavaScript from multiple sources,
both trusted and untrusted. Hence, JavaScript security is of paramount
importance. A specific interesting problem is information flow control
(IFC) for JavaScript. In this paper, we develop, formalize and
implement a dynamic IFC mechanism for the JavaScript engine of a
production Web browser (specifically, Safari's WebKit engine). Our IFC
mechanism works at the level of JavaScript \emph{bytecode} and hence
leverages years of industrial effort on optimizing both the source to
bytecode compiler and the bytecode interpreter.  We track both
explicit and implicit flows and observe only moderate
overhead. Working with bytecode results in new challenges including
the extensive use of unstructured control flow in bytecode (which
complicates lowering of program context taints), unstructured
exceptions (which complicate the matter further) and the need to make
IFC analysis permissive. We explain how we address these challenges,
formally model the JavaScript bytecode semantics and our
instrumentation, prove the standard property of
termination-insensitive non-interference, and present experimental
results on an optimized prototype.

\end{abstract}


\section{Introduction}
JavaScript (JS) is an indispensable part of the modern Web. More than
95\% of all websites use JS for browser-side computation in
Web applications~\cite{richards11ECOOP}. Aggregator websites
(e.g.,~news portals) integrate content from various mutually untrusted
sources.  Online mailboxes display context-sensitive advertisements.
All these components are glued together with JS. The dynamic
nature of JS permits easy inclusion of external libraries and
third-party code, and encourages a variety of code injection attacks,
which may lead to integrity violations. Confidentiality violations
like information stealing are possible wherever third-party code is
loaded directly into another web page~\cite{jang10CCS}.  Loading
third-party code into separate iframes protects the main frame by the
same-origin policy, but hinders interaction that mashup pages
crucially rely on and does not guarantee absence of
attacks~\cite{oopsla13}.  \emph{Information flow control} (IFC) is an
elegant solution for such problems. It ensures security even in the
presence of untrusted and buggy code.  IFC for JS differs from
traditional IFC as JS is extremely
dynamic~\cite{oopsla13,richards11ECOOP}, which makes sound static
analysis difficult.

Therefore, research on IFC for JS has focused on dynamic
techniques. These techniques may be grouped into four broad
categories. First, one may build an IFC-enabled, custom interpreter
for JS source~\cite{csf12,bello13PhD}. This turns out to be extremely
slow and requires additional code annotations to handle
semi-structured control flow like exceptions, return-in-the-middle,
break and continue.  Second, we could use a black-box technique,
wherein an off-the-shelf JS interpreter is wrapped in a monitor. This
is nontrivial, but doable with only moderate overhead and has been
implemented in secure multi-ex\-e\-cu\-tion
(SME)\cite{SME,flowfox}. However, because SME is a black-box
technique, it is not clear how it can be generalized beyond
\emph{non-interference}~\cite{goguen} 
to handle \emph{declassification}~\cite{declass1,declass2}.
Third, some variant of inline reference monitoring (IRM) might inline
taint tracking with the client code. Existing security systems for JS
with IRM require subsetting the language in order to prevent dynamic
features that can invalidate the monitoring process.  Finally, it is
possible to instrument the runtime system of an existing JS engine,
either an interpreter or a just-in-time compiler (JIT), to monitor the
program on-the-fly. While this requires adapting the respective
runtime, it incurs only moderate overhead because it retains other
optimizations within the runtime and is resilient to subversion
attacks.

In this work, we opt for the last approach. We instrument a production
JS engine to track taints dynamically and enforce 
\emph{term\-in\-ation-in\-sens\-itive
  non-inter\-ference}~\cite{volpano}. Specifically, we instrument the
bytecode interpreter in WebKit, the JS engine used in Safari and other
open-source browsers.  The major benefit of working in the bytecode
interpreter as opposed to source is that we retain the benefits of
these years of engineering efforts in optimizing the production
interpreter and the source to bytecode compiler.

We describe the key challenges that arise in dynamic IFC for JS
bytecode (as opposed to JS source), present our formal model of the
bytecode, the WebKit JS interpreter and our instrumentation, present
our correctness theorem, and list experimental results from a
preliminary evaluation with an optimized prototype running in
Safari. In doing so, our work significantly advances the
state-of-the-art in IFC for JS. Our main contributions are:
\begin{compactitem}
\item We formally model WebKit's bytecode syntax and semantics, our
  instrumentation for IFC analysis and prove non-interference. As far
  as we are aware, this is the first formal model of bytecode of an
  in-production JS engine. This is a nontrivial task because WebKit's
  bytecode language is large (147 bytecodes) and we built the model
  through a careful and thorough understanding of approximately 20,000
  lines of actual interpreter code.\footnote{Unlike some prior work,
    we are not interested in modeling semantics of JS specified by the
    ECMAScript standard. Our goal is to remain faithful to the
    production bytecode interpreter. Our formalization is based on
    WebKit build $\#$r122160, which was the last build when we started
    our work.}
\item Using ideas from prior work~\cite{just11PLASTIC}, we use
  on-the-fly intra-procedural static analysis of immediate
  post-dominators to restrict overtainting, even with bytecode's
  pervasive unstructured conditional jumps.  We extend the prior work
  to deal with exceptions. Our technique covers all unstructured
  control flow in JS (including break and continue), without requiring
  additional code annotations of prior work~\cite{bello13PhD} and
  improves permissiveness.
\item To make IFC execution more permissive, we propose and implement
  a byte\-code-specific variant of the \emph{permissive-upgrade} check~\cite{plas10}.
\item We implement our complete IFC mechanism in WebKit and observe
  moderate overheads.
\end{compactitem}

\paragraph{Limitations}
We list some limitations of our work to clarify its scope. Although
our instrumentation covers all WebKit bytecodes, we have not yet
instrumented or modeled native JS methods, including those that
manipulate the Document Object Model (DOM). This is ongoing work,
beyond the scope of this paper. Like some prior work~\cite{csf12}, our
sequential non-interference theorem covers only single invocations of
the JS interpreter. In reality, JS is reactive. The interpreter is
invoked every time an event (like a mouse click) with a handler occurs
and these invocations share state through the DOM. We expect that
generalizing to \emph{reactive non-interference}~\cite{rni} will not
require any instrumentation beyond what we already plan to do for the
DOM.  Finally, we do not handle JIT-compilation as it is considerably
more engineering effort. JIT can be handled by inlining our IFC
mechanism through a bytecode transformation.

Due to lack of space, several proofs and details of the model have
been omitted from this paper. They can be found in the technical
appendix (Section~\ref{sec:appendix}).

\section{Related Work}
\label{sec:related}

Three classes of research are closely related to our work:
formalization of JS semantics, IFC for
dynamic languages, and formal models of Web browsers. Maffeis et
al.~\cite{taly} present a formal semantics for the entire ECMA-262
specification, the foundation for JS 3.0.  Guha et
al.~\cite{Guha} present the semantics of a core language which models
the essence of JS and argue that all of JS 3.0 can be
translated to that core. S5~\cite{S5} extends~\cite{Guha} to include
accessors and eval. Our work goes one step further and formalizes the
core language of a production JS engine (WebKit), which is
generated by the source-to-bytecode compiler included in
WebKit. Recent work by Bodin et al.~\cite{popl14} presents a Coq
formalization of ECMAScript Edition 5 along with an extracted
executable interpreter for it. This is a formalization of the English
ECMAScript specification whereas we formalize the JS bytecode
implemented in a real Web browser.

Information flow control is an active area of security research. With
the widespread use of JS, research in dynamic techniques for
IFC has regained momentum.  Nonetheless, static analyses are not
completely futile. Guarn\-ieri et al.~\cite{guarnieri11ISSTA} present
a static abstract interpretation for tracking taints in JS. However,
the omnipresent \texttt{eval} construct is not supported and this
approach does not take implicit flows into account.  Chugh et
al. propose a staged information flow approach for
JS~\cite{stagedIFC}. They perform server-side static policy
checks on statically available code and generate residual
policy-checks that must be applied to dynamically loaded code.  This
approach is limited to certain JS constructs excluding dynamic
features like dynamic field access or the \texttt{with} construct.

Austin and Flanagan~\cite{plas09} propose purely dynamic IFC for
dynamically-typed languages like JS. They use the
no-sensitive-upgrade (NSU) check~\cite{zdancewic02PhD} to handle
implicit flows. Their per\-mis\-sive-upgrade strategy~\cite{plas10} is
more permissive than NSU but retains ter\-mi\-na\-tion-insensitive
non-interference. We build on the permissive-upgrade strategy.  Just
et al.~\cite{just11PLASTIC} present dynamic IFC for JS bytecode with
static analysis to determine implicit flows precisely even in the
presence of semi-unstructured control flow like \texttt{break} and
\texttt{continue}. Again, NSU is leveraged to prevent implicit
flows. Our overall ideas for dealing with unstructured control flow
are based on this work. In contrast to this paper, there was no
formalization of the bytecodes, no proof of correctness, and implicit
flow due to exceptions was ignored.

Hedin and Sabelfeld propose a dynamic IFC approach for a language
which models the core features of JS~\cite{csf12}, but they
ignore JS's constructs for semi-struc\-tured control flow
like \texttt{break} and \texttt{continue}. Their approach leverages a dynamic type
system for JS source. To improve permissiveness, their
subsequent work~\cite{esorics12} uses testing. It detects security
violations due to branches that have not been executed and injects
annotations to prevent these in subsequent runs. A further extension
introduces annotations to deal with semi-structured control
flow~\cite{bello13PhD}. Our approach relies on analyzing CFGs and
does not require annotations.


Secure multi-execution (SME)~\cite{SME} is another approach to enforcing
non-interference at runtime. Conceptually, one executes the same code
once for each security level (like low and high) with the following
constraints: high inputs are replaced by default values for the low
execution, and low outputs are permitted only in the low execution.
This modification of the semantics forces even unsafe scripts to
adhere to non-interference. FlowFox~\cite{flowfox} demonstrates SME in
the context of Web browsers. Executing a script
multiple times can be prohibitive for a security lattice with multiple
levels. 
Further, all writes to the DOM are
considered publicly visible output, while tainting allows persisting a
security label on DOM elements. 
It is also unclear how declassification may be integrated into SME.  
Austin and Flanagan~\cite{austin12POPL} introduce a
notion of faceted values to simulate multiple executions in one
run. They keep $n$ values for every variable corresponding to $n$
security levels. All the values are used for computation as the
program proceeds but the mechanism enforces non-interference by
restricting the leak of high values to low observers.

Browsers work reactively; input is fed to an event queue that is processed
over time. Input to one event can produce output that influences the input
to a subsequent event.
Bohannon et al.~\cite{rni} present a formalization of a reactive system
and compare several definitions of reactive non-interference. Bielova
et al.~\cite{rnib} extend reactive non-interference to a browser model
based on SME. This is currently the only approach that supports
reactive non-interference for JS. We will extend our work to the
reactive setting as the next step.

Finally, Featherweight Firefox~\cite{FWF} presents a formal model of a
browser based on a reactive model that resembles that of Bohannon et
al.~\cite{rni}.  It instantiates the consumer and producer states in
the model with actual browser objects like window, page, cookie store,
mode, connection, etc.
Our current work entirely focuses on the formalization of the
JS engine and taint tracking to monitor information leaks. We
believe these two approaches complement each other and plan to
integrate such a model into our future holistic enforcement mechanism
spanning JS, the DOM and other browser components.



\section{Background}
\label{sec:background}

We provide a brief overview of basic concepts in dynamic enforcement
of information flow control (IFC).  In dynamic IFC, a language runtime
is instrumented to carry a security label or taint with every
value. The taint is an element of a pre-determined lattice and is an
upper bound on the security levels of all entities that have
influenced the computation that led to the value. For simplicity of
exposition, we use throughout this paper a three-point lattice
$\{L,H,\star\}$ ($L$ = low or public, $H$ = high or secret, $\star$ =
partially leaked secret), with $L \sqsubseteq H \sqsubseteq
\star$~\cite{plas10}. For now, readers may ignore $\star$. Our
instrumentation works over a more general powerset lattice, whose
individual elements are Web domains. We write $r^\ell$ for a value $r$
tagged with label~$\ell$.

Information flows can be categorized as \emph{explicit} and
\emph{implicit}~\cite{denning}.  Explicit flows arise as a result of
variables being assigned to others, or through primitive
operations. For instance, the statement \texttt{x = y + z} causes an
explicit flow from values in both \texttt{z} and \texttt{y} to
\texttt{x}. Explicit flows are handled in the runtime by updating the
label of the computed value (\texttt{x} in our example) with the least
upper bound of the labels of the operands in the computation
(\texttt{y}, \texttt{z} in our example).

Implicit flows arise from control dependencies. For example, in the
program \texttt{l = 0; if (h) \{l = 1;\}}, there is an implicit flow
from \texttt{h} to the final value of \texttt{l} (that value is
\texttt{1} iff \texttt{h} is \texttt{1}). To handle implicit flows,
dynamic IFC systems maintain the so-called $pc$ label (program-context
label), which is an upper bound on the labels of values that have
influenced the control flow thus far. In our last example, if the
value in \texttt{h} has label $H$, then $pc$ will be $H$ within the
\texttt{if} branch. After \texttt{l = 1} is executed, the final value
of \texttt{l} inherits not only the label of \texttt{1} (which is
$L$), but also of the $pc$; hence, that label is also $H$. This alone
does not prevent information leaks: When \texttt{h = 0}, \texttt{l}
ends with $\texttt{0}^L$; when \texttt{h = 1}, \texttt{l} ends with
$\texttt{1}^H$. Since $\texttt{0}^L$ and $\texttt{1}^H$ can be
distinguished by a public attacker, this program leaks the value of
\texttt{h} despite correct propagation of implicit taints. Formally,
the instrumented semantics so far fail the standard property of
\emph{non-interference}~\cite{goguen}.

This problem can be resolved through the well-known
\emph{no-sensitive-upgrade} (NSU) check~\cite{zdancewic02PhD,plas09},
which prohibits assignment to a
low-labeled variable when $pc$ is high. This recovers non-interference
if the adversary cannot observe program termination
(\emph{termination-insensitive non-interference}). In our
example, when \texttt{h = 0}, the program terminates with $\texttt{l}
= \texttt{0}^L$. When \texttt{h = 1}, the instruction \texttt{l = 1}
gets stuck due to NSU. These two outcomes are deemed observationally
equivalent for the low adversary, who cannot determine whether or
not the program has terminated in the second case. Hence, the program
is deemed secure.

Roughly, a
program is termination-insensitive non-interferent if any two
terminating runs of the program starting from low-equivalent heaps
(i.e., heaps that look equivalent to the adversary) end in
low-equivalent heaps. 
Like all sound dynamic
IFC approaches, our instrumentation renders any JS program
termination-insensitive non-interferent, at the cost of modifying
semantics of programs that leak information.



\section{Design, Challenges, Insights and Solutions}
\label{sec:design}

We implement dynamic IFC for JS in the widely used WebKit
engine by instrumenting WebKit's bytecode interpreter.  In WebKit,
bytecode is generated by a source-code compiler.  Our goal is to not
modify the compiler, but we are forced to make slight changes to it to
make it compliant with our instrumentation. The modification is
explained in Section~\ref{sec:impl}. Nonetheless, almost all our work
is limited to the bytecode interpreter.

WebKit's bytecode interpreter is a rather standard stack machine, with
several additional data structures for JS-specific features
like scope chains, variable environments, prototype chains and
function objects. Local variables are held in registers on the call
stack. Our instrumentation adds a label to all data structures,
including registers, object properties and scope chain pointers, adds
code to propagate explicit and implicit taints and implements a more
permissive variant of the NSU check. Our label is a word size bit-set
(currently 64 bits); each bit in the bit-set represents taint from a
distinct domain (like google.com). Join on labels is simply bitwise
or.

Unlike the ECMAScript specification of JS semantics, the
actual implementation does \emph{not} treat scope chains or variable
environments like ordinary objects. Consequently, we model and
instrument taint propagation on all these data structures
separately. Working at the low-level of the bytecode also leads to
several interesting conceptual and implementation issues in taint
propagation as well as interesting questions about the threat model,
all of which we explain in this section. Some of the issues are quite
general and apply beyond JS. For example, we combine our
dynamic analysis with a bit of static analysis to handle unstructured
control flow and exceptions.

\paragraph{Threat model and compiler assumptions}
We explain our high-level threat model. Following standard practice,
our adversary may observe all low-labeled values in the heap (more
generally, an adversary at level $\ell$ in a lattice can observe all
heap values with labels $\leq \ell$).
However, we do not allow the adversary to directly observe internal
data structures like the call stack or scope chains. This is
consistent with actual interfaces in a browser that third-party
scripts can access.  In our non-interference proofs we must also show
low-equivalence of these internal data structures across two runs to
get the right induction invariants, but assuming that they are
inaccessible to the adversary allows more permissive program
execution, which we explain in Section~\ref{sec:challenges}.

The bytecode interpreter executes in a shared space with other browser
components, so we assume that those components do not leak information
over side channels, e.g., they do not copy heap data from secret to
public locations. This also applies to the compiler, but we do not
assume that the compiler is functionally correct. Trivial errors in
the compiler, e.g., omitting a bytecode could result in a leaky
program even when the source code has no information leaks. Because
our IFC works on the compiler's output, such compiler errors are not a
concern. Formally, we assume that the compiler is an unspecified
deterministic function of the program to compile and of the call
stack, but not of the heap.
This assumption also matches how the compiler works within WebKit: It
needs access to the call stack and scope chain to optimize generated
bytecode.
However, the compiler never needs access to the heap.
We ignore information leaks due to other side channels like timing.




\subsection{Challenges and Solutions}
\label{sec:challenges}

IFC for JS is known to be difficult due to JS's highly dynamic
nature. Working with bytecode instead of source code makes IFC
harder. Nonetheless, solutions to many JS-specific IFC concerns
proposed in earlier work~\cite{csf12} also apply to our
instrumentation, sometimes in slightly modified form. For example, in
JS, every object has a fixed parent, called a prototype, which is
looked up when a property does not exist in the child. This can lead
to implicit flows: If an object is created in a high context (when the
$pc$ is high) and a field missing from it, but present in the
prototype, is accessed later in a low context, then there is an
implicit leak from the high $pc$.  This problem is avoided in both
source- and bytecode-level analysis in the same way: The ``prototype''
pointer from the child to the parent is labeled with the $pc$ where
the child is created, and the label of any value read from the parent
after traversing the pointer is joined with this label. Other
potential information flow problems whose solutions remain unchanged
between source- and bytecode-level analysis include implicit leaks
through function pointers and handling of
\texttt{eval}~\cite{just11PLASTIC,csf12}.

Working with bytecode both leads to
some interesting insights, which are, in some
cases, even applicable to source code analysis and other languages, and poses new
challenges. We
discuss some of these challenges and insights.

\paragraph{Unstructured control flow and CFGs}
To avoid overtainting $pc$ labels, an important goal in implicit flow
tracking is to determine when the influence of a control construct has
ended. For block-structured control flow limited to \texttt{if} and
\texttt{while} commands, this is straightforward: The effect of a
control construct ends with its lexical scope, e.g., in (\texttt{if
  (h) \{l = 1;\}; l = 2}), \texttt{h} influences the control flow at
\texttt{l = 1} but not at \texttt{l = 2}. This leads to a
straightforward $pc$ upgrading and downgrading strategy: One maintains
a \emph{stack} of $pc$ labels~\cite{zdancewic02PhD}; the effective
$pc$ is the top one. When entering a control flow construct like
\texttt{if} or \texttt{while}, a new $pc$ label, equal to the join of
labels of all values on which the construct's guard depends with the
previous effective $pc$, is pushed. When exiting the construct, the
label is popped.

Unfortunately, it is unclear how to extend this simple strategy to
non-block-structured control flow constructs such as exceptions,
\texttt{break}, \texttt{continue} and \texttt{return}-in-the-middle
for functions, all of which occur in JS. For example, consider
the program \texttt{l = 1; while(1) \{... if (h) \{break;\}; l =
  0; break;\}} with \texttt{h} labeled $H$. This program leaks the
value of \texttt{h} into \texttt{l}, but no assignment to \texttt{l}
appears in a block-scope guarded by \texttt{h}. Indeed, the $pc$
upgrading and downgrading strategy just described is ineffective for
this program. Prior work on source code IFC either omits some of these
constructs~\cite{csf12,acsac09}, or introduces additional classes of
labels to address these problems --- a label for
exceptions~\cite{csf12}, a label for each loop containing
\texttt{break} or \texttt{continue} and a label for each
function~\cite{bello13PhD}.  These labels are more restrictive than
needed, e.g., the code indicated by dots in the example above is
executed irrespective of the condition \texttt{h} in the first
iteration, and thus there is no need to raise the $pc$ before checking
that condition. Further, these labels are programmer annotations,
which we cannot support as we do not wish to modify the compiler.

Importantly, unstructured control flow is a \emph{very serious}
concern for us, because WebKit's bytecode has completely unstructured
branches like jump-if-false. In fact, all control flow, except
function calls, is unstructured in bytecode.

To solve this problem, we adopt a solution based on static analysis of
generated bytecode~\cite{denning82,just11PLASTIC}. We maintain a
control flow graph (CFG) of known bytecodes and for each branch node,
compute its immediate post-dominator (IPD). The IPD of a node is the
first instruction that will definitely be executed, no matter which
branch is taken.
Our $pc$ upgrading and downgrading
strategy now extends to arbitrary control flow: When executing a
branch node, we push a new $pc$ label on the stack \emph{along with}
the node's IPD. When we actually reach the IPD, we pop the $pc$
label. In~\cite{Xin,Masri}, the authors prove that the IPD marks the
end of the scope of an operation and hence the security context of the
operation, so our strategy is sound. In our earlier example, the IPD
of \texttt{if(h) \ldots} is the end of the \texttt{while} loop because
of the first \texttt{break} statement, so when \texttt{h == 0}, the
assignment \texttt{l = 1} fails due to the NSU check and the program
is termination-insensitive non-interference secure.

JS requires dynamic code compilation. We are forced to extend the CFG
and to compute IPDs whenever code for either a function or an
\texttt{eval} is compiled. Fortunately, the IPD of a node in the CFG
lies either in the same function as the node or some function earlier
in the call-chain (the latter may happen due to exceptions), so
extending the CFG does not affect computation of IPDs of earlier
nodes. This also relies on the fact that code generated from
\texttt{eval} cannot alter the CFG of earlier functions in the call
stack~\cite{just11PLASTIC}. In the actual implementation, we optimize
the calculation of IPDs further by working only intra-procedurally, as
explained below. At the end, our IPD-based solution works for all
forms of unstructured control flow, including unstructured branches in
the bytecode, and semi-structured \texttt{break}, \texttt{continue},
\texttt{return}-in-the-middle and exceptions in the source code.

\paragraph{Exceptions and synthetic exit nodes}
Maintaining a CFG in the presence of exceptions is expensive. An
exception-throwing node in a function that does not catch that
exception should have an outgoing control flow edge to the next
exception handler in the call-stack. This means that (a) the CFG is,
in general, inter-procedural, and (b) edges going out of a function
depend on its calling context, so IPDs of nodes in the function must
be computed \emph{every time the function is called}. Moreover, in the
case of recursive functions, the nodes must be replicated for every
call. This is rather expensive. Ideally, we would like to build the
function's CFG once when \emph{the function is compiled} and work
intra-procedurally (as we would had there been no exceptions). We
explain how we attain this goal in the sequel.\footnote{This problem
  and our solution are not particular to JS; they apply to
  dynamic IFC analysis in all languages with exceptions and
  functions.}


In our design, every function that may throw an unhandled exception
has a special, \emph{synthetic exit node} (SEN), which is placed after
the regular return node(s) of the function. Every exception-throwing
node, whose exception will not be caught within the function, has an
outgoing edge to the SEN, which is traversed when the exception is
thrown. The semantics of SEN (described below) correctly transfer
control to the appropriate exception handler. By doing this, we
eliminate all cross-function edges and our CFGs become
intra-procedural. The CFG of a function can be computed when the
function is compiled and is never updated. (In our implementation, we
build two variants of the CFG, depending on whether or not there is an
exception handler in the call stack. This improves efficiency, as we
explain later.)

Control flows to the SEN when the function returns normally or when an
exception is thrown but not handled within the function. If no
unhandled exception occurred within the function, then the SEN
transfers control to the caller (we record whether or not an unhandled
exception occurred). If an unhandled exception occurred, then the SEN
triggers a special mechanism that searches the call stack backward for
the first appropriate exception handler and transfers control to
it. (In JS, exceptions are indistinguishable, so we need to
find only the first exception handler.) Importantly, we pop the
call-stack up to the frame that contains the first exception
handler but do \emph{not} pop the $pc$-stack, which ensures
that all code up to the exception handler's IPD executes with the same
$pc$ as the SEN, which is indeed the semantics one would expect if we
had a CFG with cross-function edges for exceptions. This prevents
information leaks.

If a function does not handle a possible exception but there is an
exception handler on the call stack, then all bytecodes that could
potentially throw an exception have the SEN as one successor in the
CFG. Any branching bytecode will thus need to push to the $pc$-stack
according to the security label of its condition. However, we do
\emph{not} push a new $pc$-stack entry if the IPD of the current node
is the same as the IPD on the top of the $pc$-stack (this is just an
optimization) or if the IPD of the current node is the SEN, as in this
case the \emph{real} IPD, which is outside of this method, is already
on the $pc$-stack.  These semantics emulate the effect of having
cross-function exception edges.

For illustration, consider the following two functions \texttt{f} and
\texttt{g}. The $\diamond$ at the end of \texttt{g} denotes its
SEN. Note that there is an edge from \texttt{throw 9} to $\diamond$
because \texttt{throw 9} is not handled within~\texttt{g}. $\Box$
denotes the IPD of the handler \texttt{catch(e) \{ l = 1; \}}.

\noindent 
\begin{tabular}{@{}p{0.55\textwidth}p{0.35\textwidth}@{}}
\begin{tabular}[t]{@{}l@{}}
\texttt{function f() = \{} \\
~~\texttt{l = 0;} \\
~~\texttt{try \{ g(); \} catch(e) \{ l = 1; \}} \\
~~$\Box$ \texttt{return l;}\\
\texttt{\}}
\end{tabular} & 
\begin{tabular}[t]{@{}l@{}}
\texttt{function g() = \{}\\
~~\texttt{if (h) \{throw 9;\}}\\
~~\texttt{return 7;}\\
\texttt{\}} $\diamond$
\end{tabular}
\end{tabular}

It should be clear that in the absence of instrumentation, when
\texttt{f} is invoked with $pc = L$, the two functions together leak
the value of \texttt{h} (which is assumed to have label $H$) into the
return value of \texttt{f}. We show how our SEN mechanism prevents
this leak.  When invoking \texttt{g()} we do not know if there will be
an exception in this function. Depending on the outcome of this method
call, we will either jump to the exception handler or continue at
$\Box$. Based on that branch, we push the current pc and IPD $(L,
\Box)$ on the $pc$-stack.  When executing the condition \texttt{if
  (h)} we do \emph{not} push again, but merely update the top element
to $(H, \Box)$.  If \texttt{h == 0}, control reaches $\diamond$
without an exception but with $pc = H$ because the IPD of \texttt{if
  (h)} is $\diamond$. At this point, $\diamond$ returns control to
\texttt{f}, thus $pc = H$, but at $\Box$, $pc$ is lowered to $L$, so
\texttt{f} ends with the return value $\texttt{0}^L$. If \texttt{h ==
  1}, control reaches $\diamond$ with an unhandled exception. At this
point, following the semantics of SEN, we find the exception handler
\texttt{catch(e) \{ l = 1; \}} and invoke it with the same $pc$ as the
point of exception, i.e., $H$. Consequently, NSU prevents the
assignment \texttt{l = 1}, which makes the program
termination-insensitive non-interferent.

Because we do not wish to replicate the CFG of a function every time
it is called recursively, we need a method to distinguish the same
node corresponding to two different recursive calls on the
$pc$-stack. For this, when pushing an IPD onto the $pc$-stack, we pair
it with a pointer to the current call-frame. Since the call-frame
pointer is unique for each recursive call, the CFG node paired with
the call-frame identifies a unique merge point in the real control
flow graph.

In practice, even the intra-procedural CFG is quite dense because many
JS bytecodes can potentially throw exceptions and, hence, have
edges to the SEN. To avoid overtainting, we perform a crucial
common-case optimization: When there is no exception handler on the
call stack we do not create the SEN and the corresponding edges from
potentially exception-throwing bytecodes at all. This is safe as a
potentially thrown exception can only terminate the program instantly,
which satisfies termination-insensitive non-interference if we ensure
that the exception message is not visible to the attacker. Whether or
not an exception handler exists is easily tracked using a stack of
Booleans that mirrors the call-stack; in our design we overlay this
stack on the $pc$-stack by adding an extra Boolean field to each entry
of the $pc$-stack.
In summary, each entry of our $pc$-stack is a quadruple containing a
security label, a node in the intraprocedural CFG, a call-frame
pointer and a Boolean value. In combination with SENs, this design
allows us to work only with intraprocedural CFGs that are computed
when a function is compiled. This improves efficiency.

\paragraph{Permissive-upgrade check, with changes}
The standard NSU check halts program execution whenever an attempt is
made to assign a variable with a low-labeled value in a high $pc$. In
our earlier example, \texttt{l = 0; if (h) \{l = 1;\}}, assuming that
\texttt{h} stores a $H$-labeled value, program execution is halted at
the command \texttt{l = 1}. As Austin and Flanagan (AF in the sequel)
observe~\cite{plas10}, this may be overly restrictive when \texttt{l}
will not, in fact, have observable effects (e.g., \texttt{l} may be
overwritten by a constant immediately after \texttt{if (h) \{l =
  1;\}}). So, they propose propagating a special taint called $\star$
into \texttt{l} at the instruction \texttt{l = 1} and halting a
program when it tries to \emph{use} a value labeled $\star$ in a way
that will be observable (AF call this special taint $P$ for ``partially leaked'').
This idea, called the
\emph{permissive-upgrade} check, allows more program execution than
NSU would, so we adopt it. In fact, this additional permissiveness is
absolutely essential for us because the WebKit compiler often
generates dead assignments within branches, so execution would
pointlessly halt if standard NSU were used.

We differ from AF in \emph{what} constitutes a use of a value labeled
$\star$. As expected, AF treat occurrence of $\star$ in the guard of a
branch as a use. Thus, the program \texttt{l = 0; if (h) \{l = 1;\};
  if (l) \{l' = 2\}} is halted at the command \texttt{if (l)} when
\texttt{h == 1} because \texttt{l} obtains taint $\star$ at the
assignment \texttt{l = 1} (if the program is not halted, it leaks
\texttt{h} through \texttt{l'}). However, they allow $\star$-tainted
values to flow into the heap. Consider the program \texttt{l = 0; if
  (h) \{l = 1;\}; obj.a = l}. This program is insecure in our model:
The heap location \texttt{obj.a}, which is accessible to the
adversary, ends with $\texttt{0}^L$ when \texttt{h == 0} and with
$\texttt{1}^\star$ when \texttt{h == 1}. AF deem the program secure by
assuming that any value with label $\star$ is low-equivalent to any
other value (in particular, $\texttt{0}^L$ and $\texttt{1}^\star$ are
low-equivalent).  However, this definition of low-equivalence for
dynamic analysis is virtually impossible to enforce if the adversary
has access to the heap outside the language: After writing
$\texttt{0}^L$ to \texttt{obj.a} (for \texttt{h == 0}), a dynamic
analysis cannot determine that the alternate execution of the program
(for \texttt{h == 1}) \emph{would have} written a $\star$-labeled
value and, hence, cannot prevent the adversary from seeing
$\texttt{0}^L$.

Consequently, in our design, we use a modified permissive-upgrade
check, which we call the \emph{deferred NSU check}, wherein a program
is halted at any construct that may potentially flow a $\star$-labeled
value into the heap. This includes all branches whose guard contains a
$\star$-labeled value and any assignments whose target is a heap
location and whose source is $\star$-labeled. However, we do not
constrain flow of $\star$-labeled values in data structures that are
invisible to the adversary in our model, e.g., local registers and
variable environments. This design critically relies on treating
internal data structures differently from ordinary JS objects,
which is not the case, for instance, in the ECMAScript specification.

\newcommand{\mC}{\mathcal{C}}

\section{Formal Model and IFC}\label{sec:lang}

We formally model WebKit's JS bytecode and the semantics of
its bytecode interpreter with our instrumentation of dynamic IFC. We
prove termination-insensitive non-interference for programs executed
through our instrumented interpreter. We do not model the construction
of the CFG or computation of IPDs; these are standard. To keep
presentation accessible, we present our formal model at a somewhat
high-level of abstraction. Details are resolved in our technical
appendix.




\subsection{Bytecode and Data Structures}
\label{sec:syntax}

\begin{figure}[tb]
\scalebox{0.72}{
\setlength{\tabcolsep}{4pt}
\begin{tabular}{ll}
\begin{tabular}{l}
\begin{tabular}{llll}
 ins &:= prim-ins & $|$ obj-ins \\
           & $|$ func-ins & $|$ scope-ins 
           & $|$ exc-ins\\
\end{tabular}\\
\begin{tabular}{ll}
          prim-ins &:= prim dst:r src1:r src2:r \\
                           & $|$ mov dst:r src:r \\
                           & $|$ jfalse cond:r target:offset\\
                           & $|$ loop-if-less src1:r src2:r target:offset \\
                           & $|$ typeof dst:r src:r \\
                           & $|$ instanceof dst:r value:r cProt:r \\
\end{tabular}\\ 
\begin{tabular}{ll}
	  func-ins & := enter\\
			 & $|$ ret result:r\\
			 & $|$ end result:r\\
			 & $|$ call func:r args:n\\
                         & $|$ call-put-result res:r\\
			 & $|$ call-eval func:r args:n\\
                         & $|$ create-arguments dst:r\\
			 & $|$ new-func dst:r func:f\\
                         & $|$ create-activation dst:r\\
                         & $|$ construct func:r args:n\\
                         & $|$ create-this dst:r\\
\end{tabular} 
\end{tabular}
&
\begin{tabular}{l}
\begin{tabular}{ll}
	  obj-ins &:= new-object dst:r \\
			 & $|$ get-by-id dst:r base:r prop:id\\
			 & $|$ put-by-id base:r prop:id value:r direct:b\\
			 & $|$ del-by-id  dst:r base:r prop:id\\
			 & $|$ get-pnames dst:r base:r i:n size:n breaktarget:offset\\
			 & $|$ next-pname dst:r base:r i:n size:n
                         iter:n target:offset\\
                         & $|$ put-getter-setter base:r prop:id getter:r
                         setter:r \\
\end{tabular}\\
\begin{tabular}{ll}
	  scope-ins &:= 
  			    resolve dst:r prop:id\\
                           & $|$ resolve-skip dst:r prop:id
                           skip:n \\
                           & $|$ resolve-global dst:r prop:id
                           \\
                           & $|$ resolve-base dst:r prop:id
                           isStrict:bool\\
                           & $|$ resolve-with-base bDst:r pDst:r
                           prop:id\\
                           & $|$ get-scoped-var dst:r index:n
                           skip:n\\
                           & $|$ put-scoped-var index:n skip:n
                           value:r\\
                           & $|$ push-scope  scope:r\\
			   & $|$ pop-scope \\
			   & $|$ jmp-scope count:n target:offset\\
\end{tabular}\\
\begin{tabular}{ll}
	  exc-ins & := throw ex:r\\
                       & $|$ catch ex:r \\
\end{tabular}
\end{tabular}
\end{tabular}
}

\vspace*{-3mm}
\caption{Instructions}
\vspace*{-4mm}\label{fig:inst}
\end{figure}

The version of WebKit we model uses a total of 147 bytecodes or
instructions, of which we model 69. 
The remaining 78 bytecodes are redundant from the perspective of
formal modeling because they are \emph{specializations} or wrappers on
other bytecodes to improve efficiency. The syntax of the 69 bytecodes
we model is shown in \Figure{fig:inst}. The bytecode \texttt{prim}
abstractly represents 34 primitive binary and unary (with just the
first two arguments) operations, all of which behave similarly. For
convenience, we divide the bytecodes into primitive instructions
(prim-ins), instructions related to objects and prototype chains
(obj-ins), instructions related to functions (func-ins), instructions
related to scope chains (scope-ins) and instructions related to
exceptions (exc-ins). A bytecode has the form $\mathit{\langle
  inst\_name~~ list\_of\_args \rangle}$. The arguments to the
instruction are of the form
$\langle$var$\rangle$:$\langle$type$\rangle$, where \emph{var} is the
variable name and \emph{type} is one of the following: r, n, bool, id,
prop and offset for register, constant integer, constant Boolean,
identifier, property name and jump offset value, respectively.

In WebKit, bytecode is organized into code blocks. Each code block is
a sequence of bytecodes with line numbers and corresponds to the
instructions for a function or an \texttt{eval} statement. A code
block is generated when a function is created or an \texttt{eval} is
executed. In our instrumentation, we perform control flow analysis on
a code block when it is created and in our formal model we abstractly
represent a code block as a CFG, written $\zeta$. Formally, a CFG is a
directed graph, whose nodes are bytecodes and whose edges represent
possible control flows. There are no cross-function edges. A CFG also
records the IPD of each node. IPDs are computed using an algorithm by
Lengauer and Tarjan~\cite{Lengauer} when the CFG is created. If the
CFG contains uncaught exceptions, we also create a SEN. For a CFG
$\zeta$ and a non-branching node $\iota \in \zeta$, $\Succ(\zeta,
\iota)$ denotes $\iota$'s unique successor. For a conditional
branching node $\iota$, $\Left(\zeta, \iota)$ and
$\Right(\zeta,\iota)$ denote successors when the condition is true and
false, respectively.

The bytecode interpreter is a standard stack machine, with support for
JS features like scope chains and prototype chains. The state
of the machine (with our instrumentation) is a quadruple
$\langle\iota, \theta, \sigma, \rho\rangle$, where $\iota$ represents
the current node that is being executed, $\theta$ represents the heap,
$\sigma$ represents the call-stack and $\rho$ is the $pc$-stack.

We assume an abstract, countable set ${\cal A} = \{a,b,\ldots\}$ of
heap locations, which are references to objects. The heap $\theta$ is
a partial map from locations to objects. An object $O$ may be:
\begin{compactitem}
\item An ordinary JS object $N = (\{p_i \mapsto v_i\}_{i=0}^n,
  \_\_proto\_\_ \mapsto a^{\ell_p}, \ell_s)$, containing
  properties named $p_0,\ldots,p_n$ that map to labeled values
  $v_0,\ldots,v_n$, a prototype field that points to a parent at heap
  location $a$, and two labels $\ell_p$ and $\ell_s$. $\ell_p$ records
  the $pc$ where the object was created.
  $\ell_s$ is the so-called structure label, which is an upper bound
  on all $pc$s that have influenced which fields exist in the
  object.\footnote{The $\_\_proto\_\_$ field is the parent of the
    object; it is not the same as the prototype field of a function
    object, which is an ordinary property. Also, in our actual model,
    fields $p_i$ map to more general property descriptors that also
    contain attributes along with labeled values. We elide attributes
    here to keep the presentation simple.}
\item A function object $F = (N, \zeta, \Sigma)$, where $N$ is an
  ordinary object, $\zeta$ is a CFG, which corresponds to the the
  function stored in the object, and $\Sigma$ is the scope chain
  (closing context) of the function.
\end{compactitem}

A labeled value $v = r^\ell$ is a value $r$ paired with a security
label $\ell$. A value $r$ in our model may be a heap location $a$ or a
JS primitive value $n$, which includes integers, Booleans, regular
expressions, arrays, strings and the special JS values
\texttt{undefined} and \texttt{null}.


The call-stack $\sigma$ contains one call-frame for each incomplete
function call. A call-frame $\mu$ contains an array of registers for
local variables, a CFG $\zeta$ for the function represented by the
call-frame, the return address (a node in the CFG of the previous
frame), and a pointer to a scope-chain that allows access to variables
in outer scopes. Additionally, each call-frame has an exception table
which maps each potentially exception-throwing bytecode in the
function to the exception handler within the function that surrounds
the bytecode; when no such exception handler exists, it points to the
SEN of the function (we conservatively assume that any unknown code
may throw an exception, so bytecodes \texttt{call} and \texttt{eval}
are exception-throwing for this purpose).
$|\sigma|$ denotes the size of the call-stack and $!\sigma$ its top
frame. Each register contains a labeled value.


A scope chain, $\Sigma$, is a sequence of scope chain nodes (SCNs),
denoted $S$, paired with labels. In WebKit, a scope chain node $S$ may
either be an object or a variable environment $V$, which is an array
of labeled values. Thus, $\Sigma \;::=\; (S_1, \ell_1) : \ldots :
(S_n, \ell_n)$ and $S \;::=\; O ~|~ V$ and $V \;::=\; v_1 : \ldots :
v_n$.

Each entry of the $pc$-stack $\rho$ is a triple $(\ell, \iota, p)$,
where $\ell$ is a security label, $\iota$ is a node in a CFG, and $p$
is a pointer to some call-frame on the call stack $\sigma$. (For
simplicity, we ignore a fourth Boolean field described in
Section~\ref{sec:challenges} in this presentation.) When we enter a new
control context, we push the new $pc$ $\ell$ together with the IPD
$\iota$ of the entry point of the control context and a pointer $p$ to
current call-frame. The pair $(\iota, p)$ uniquely identifies where
the control of the context ends; $p$ is necessary to distinguish the
same branch point in different recursive calls of the
function~\cite{just11PLASTIC}. In our semantics, we use the
meta-function $\isIPD$ to pop the stack. It takes the current
instruction, the current $pc$-stack and the call stack $\sigma$, and
returns a new $pc$-stack.
\begin{equation*}
 \isIPD(\iota, \rho, \sigma) := 
 \begin{cases}
   \rho.\mathit{pop()} & \mbox{if}~ !\rho = (\_, \iota, !\sigma)
\\ \rho & \mbox{otherwise}
 \end{cases}
\end{equation*} 

As explained in Section~\ref{sec:challenges}, as an optimization, we
push a new node $(\ell, \iota, \sigma)$ onto $\rho$ only when $(\iota,
\sigma)$ (the IPD) differs from the corresponding pair on the top of
the stack and, to handle exceptions correctly, we also require that
$\iota$ not be the SEN. Otherwise, we just join $\ell$ with the label
on the top of the stack. This is formalized in the function
$\rho.push(\ell, \iota, \sigma)$, whose obvious definition we elide.

If $x$ is a pair of any syntactic entity and a security label, we
write $\Upsilon(x)$ for the entity and $\Gamma(x)$ for the label. In
particular, for $v = r^\ell$, $\Upsilon(v) = r$ and $\Gamma(v) =
\ell$.



\subsection{Semantics and IFC with Intra-procedural CFGs}
\label{sec:real}

We now present the semantics, which faithfully models our
implementation using intra-procedural CFGs with SENs.  The semantics
is defined as a set of state transition rules that define the
judgment: $\langle \iota, \theta, \sigma, \rho\rangle \leadsto
\langle\iota', \theta', \sigma', \rho'\rangle$.
\Figure{fig:semantics:ideal} shows rules for selected bytecodes. For
reasons of space we omit rules for other bytecodes and formal
descriptions of some meta-function like $opCall$ that are used in the
rules. $C \Rightarrow A \diamond B$ is shorthand for a meta-level
(if ($C$) then $A$ else $B$).


\begin{figure}[tb]
\input{./Selected-Semantics.tex}
\caption{Semantics, selected rules}
\vspace*{-5mm}\label{fig:semantics:ideal}
\end{figure}

\texttt{prim} reads the values from two registers $\texttt{src1}$ and
$\texttt{src}2$, performs a binary operation generically denoted by
$\oplus$ on the values and writes the result into the register
$\texttt{dst}$. $\texttt{dst}$ is assigned the join of the labels in
$\texttt{src1}$, $\texttt{src2}$ and the head of the $pc$-stack
($!\rho$). To implement deferred NSU (Section~\ref{sec:challenges}),
the existing label in \texttt{dst} is compared with the current
$pc$. If the label is lower than the $pc$, then the label of
$\texttt{dst}$ is joined with $\star$. Note that the premise $\rho' =
\isIPD(\iota', \rho, \sigma)$ pops an entry from the $pc$-stack if its
IPD matches the new program node $\iota'$. This premise occurs in all
semantic rules.

\texttt{jfalse} is a conditional jump. It skips \texttt{offset} number
of successive nodes in the CFG if the register \texttt{cond} contains
\texttt{false}, else it falls-through to the next node. Formally, the
node it branches to is either $\Right(\zeta, \iota)$ or
$\Left(\zeta,\iota)$, where $\zeta$ is the CFG in $!\sigma$. In
accordance with deferred NSU, the operation is performed only if
\texttt{cond} is not labeled $\star$. \texttt{jfalse} also starts a
new control context, so a new node is pushed on the top of the
$pc$-stack with a label that is the join of $\Gamma(\texttt{cond})$
and the current label on the top of the stack (unless the IPD of the
branch point is already on top of the stack or it is the SEN, in which
case we join the new label with the previous). Traversed from bottom
to top, the $pc$-stack always has monotonically non-decreasing labels.

\texttt{put-by-id} updates the property \texttt{prop} in the object
pointed to by register \texttt{base}. As explained in
Section~\ref{sec:challenges}, we allow this only if the value to be
written is not labeled $\star$. The flag \texttt{direct} states
whether or not to traverse the prototype chain in finding the
property; it is set by the compiler as an optimization. If the flag is
\texttt{true}, then the chain is not traversed (meta-function
\emph{putDirect} handles this case). If \texttt{direct} is
\texttt{false}, then the chain is traversed (meta-function
\emph{putIndirect}). Importantly, when the chain is traversed, the
resulting value is labeled with the join of
prototype labels $\ell_p$ and structure labels $\ell_s$ of all traversed
objects. This is standard and necessary to prevent implicit leaks
through the $\_\_proto\_\_$ pointers and structure changes to objects.


\texttt{push-scope}, which corresponds to the start of the JS
construct \texttt{with(obj)}, pushes the object pointed to by the
register \texttt{scope} into the scope chain. Because pushing an
object into the scope chain can implicitly leak information from the
program context later, we also label all nodes in the scope-chain with
the $pc$'s at which they were added to the chain. Further, deferred
NSU applies to the scope chain pointer in the call-frame as it does to
all other registers.



\texttt{call} invokes a function of the target object stored in the
register \texttt{func}. Due to deferred NSU, the call proceeds only if
$\Gamma(\texttt{func})$ is not $\star$. The call creates a new
call-frame and initializes arguments, the scope chain pointer
(initialized with the function object's $\Sigma$ field), CFG and the
return node in the new frame. The CFG in the call-frame is copied from
the function object pointed to by \texttt{func}. All this is
formalized in the meta-function $opCall$, whose details we omit
here. Call is a branch instruction and it pushes a new label on the
$pc$-stack which is the join of the current $pc$,
$\Gamma(\texttt{func})$ and the structure label $\ell_f$ of the
function object (unless the IPD of the current node is the SEN or
already on the top of the $pc$-stack, in which case we join the new
$pc$-label with the previous).
\texttt{call} also initializes the new registers' labels to the new
$pc$.  A separate bytecode, not shown here and executed first in the
called function, sets register values to
\texttt{undefined}. \texttt{eval} is similar to call but the code to
be executed is also compiled.

\texttt{ret} exits a function. It returns
control to the caller, as formalized in the meta-function $opRet$. The
return value is written to an interpreter variable ($\gamma$).

\texttt{throw} throws an exception, passing the value in register
\texttt{ex} as argument to the exception handler. Our $pc$-stack push
semantics ensure that the exception handler, if any, is present in
the call-frame pointed to by the \emph{top} of the $pc$-stack. The
meta-function \emph{throwException} pops the call-stack up to this
call-frame and transfers control to the exception handler, by looking
it up in the exception table of the call-frame. The exception value in
the register \texttt{ex} is transferred to the handler through an
interpreter variable.

The semantics of other bytecodes have been described in
Section~\ref{sec:app-sem}. 


\paragraph{Correctness of IFC}
We prove that our IFC analysis guarantees termination-insensitive
non-interfer\-ence~\cite{volpano}. Intuitively, this means that if a
program is run twice from two states that are observationally
equivalent for the adversary and both executions terminate, then the
two final states are also equivalent for the adversary. To state the
theorem formally, we formalize equivalence for various data structures
in our model. The only nonstandard data structure we use is the CFG,
but graph equality suffices for it. A well-known complication is that
low heap locations allocated in the two runs need not be identical. We
adopt the standard solution of parametrizing our definitions of
equivalence with a partial bijection $\beta$ between heap
locations. The idea is that two heap locations are related in the
partial bijection if they were created by corresponding allocations in
the two runs. We then define a rather standard relation $\langle
\iota_1, \theta_1, \sigma_1, \rho_1\rangle \sim^{\beta}_{\ell}
\langle\iota_2, \theta_2, \sigma_2, \rho_2\rangle$, which means that
the states on the left and right are equivalent to an observer at
level $\ell$, up to the bijection $\beta$ on heap locations. The
details have been presented in Section~\ref{sec:app-ni}.

\begin{myThm}[Termination-insensitive non-interference]
Suppose:\linebreak[5] (1) $\langle \iota_1, \theta_1, \sigma_1, \rho_1\rangle
\sim^\beta_\ell \langle\iota_2, \theta_2, \sigma_2, \rho_2\rangle$,
(2) $\langle \iota_1, \theta_1,\sigma_1,\rho_1\rangle \leadsto^{*}
\langle \texttt{end} ,\theta_1', [],[]\rangle$, and (3) $\langle
\iota_2, \theta_2,\sigma_2,\rho_2\rangle \leadsto^{*} \langle
\texttt{end},\theta_2', [],[]\rangle$.  Then, $\exists \beta'
\supseteq \beta$ such that $\theta_1' \sim^{\beta'}_\ell \theta_2'$.
\end{myThm}

\section{Implementation}
\label{sec:impl}

We instrumented WebKit's JS engine (JavaScriptCore) to
implement the IFC semantics of the previous section. Before a function
starts executing, we generate its CFG and calculate IPDs of its nodes
by static analysis of its bytecode.
We modify the source-to-bytecode compiler to emit a slightly
different, but functionally equivalent bytecode sequence for
\texttt{finally} blocks; this is needed for accurate computation of
IPDs. For evaluation purposes, we label each source script with the
script's domain of origin; each seen domain is dynamically allocated a
bit in our bit-set label. In general, our instrumentation terminates a
script that violates IFC. However, for the purpose of evaluating
overhead of our instrumentation, we ignore IFC violations in all
experiments described here.

We also implement and evaluate a variant of \emph{sparse
  labeling}~\cite{plas09} which optimizes the common case of
computations that mostly use local variables (registers in the
bytecode). Until a function reads a value from the heap with a label
different from the $pc$, we propagate taints only on heap-writes, but
not on in-register computations. Until that point, all registers are
assumed to be implicitly tainted with $pc$. This simple optimization
reduces the overhead incurred by taint tracking significantly in
microbenchmarks. For both the basic and optimized version, our
instrumentation adds approximately 4,500 lines of code to WebKit.



Our baseline for evaluation is the uninstrumented interpreter with JIT
disabled. For comparison, we also include measurements with JIT
enabled. Our experiments are based on WebKit build $\#$r122160 running
in Safari 6.0.  The machine has a 3.2GHz Quad-core Intel Xeon
processor with 8GB RAM and runs Mac OS X version 10.7.4.

\begin{figure}[tb]
\centering
\includegraphics[height=6cm]{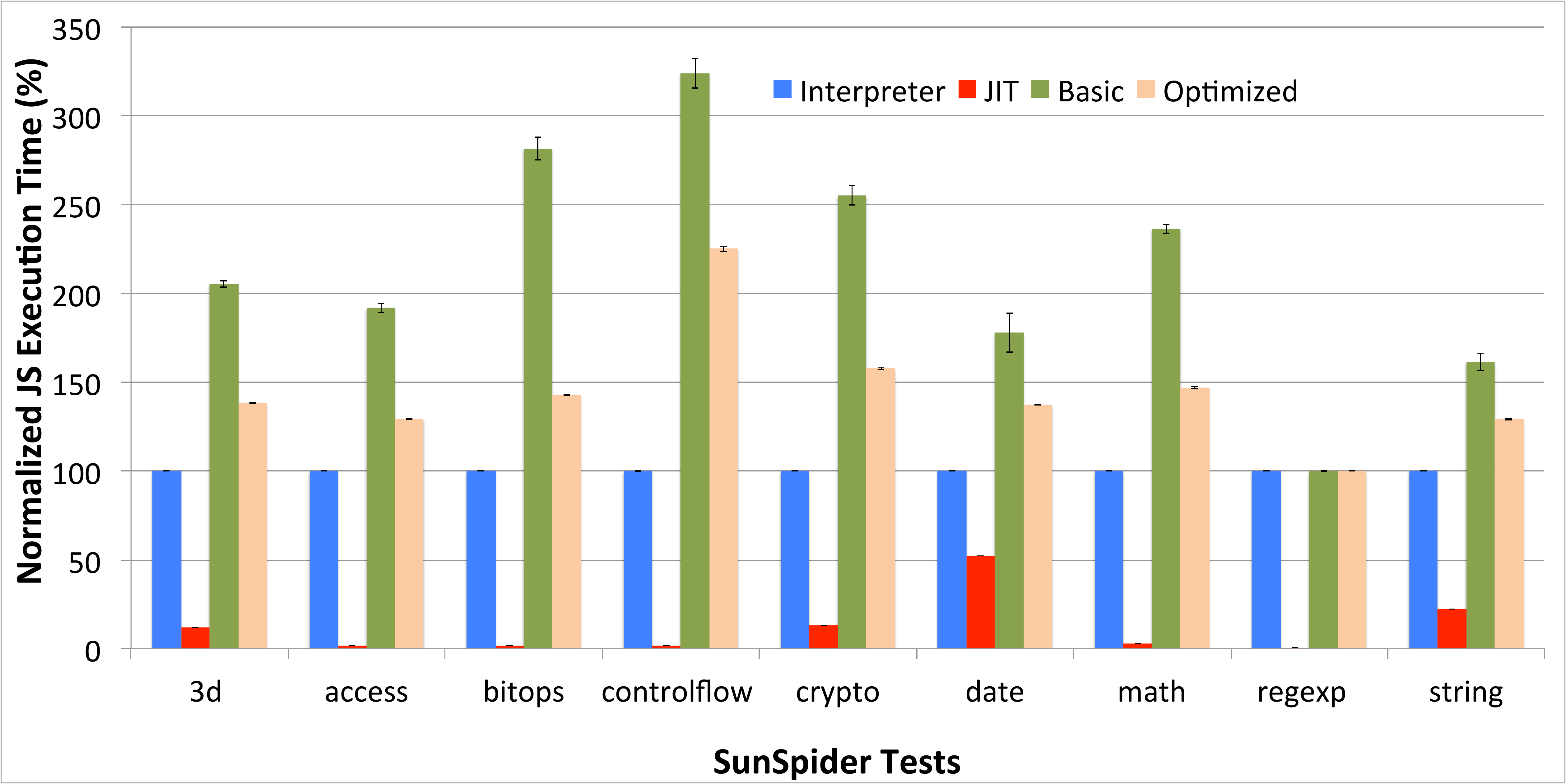}
\vspace*{-4mm}
\caption{Overheads of basic and optimized IFC in SunSpider benchmarks}
\vspace*{-4mm}\label{fig:ss-normalized}
\end{figure}

\paragraph{Microbenchmark}
We executed the standard SunSpider 1.0.1 JS benchmark suite on
the uninstrumented interpreter with JIT disabled and JIT enabled, and
on the basic and the optimized IFC instrumentations with JIT
disabled. Results are shown in Figure~\ref{fig:ss-normalized}. The
x-axis ranges over SunSpider tests and the y-axis shows the average
execution time, normalized to our baseline (uninstrumented interpreter
with JIT disabled) and averaged across 100 runs. Error bars are
standard deviations. Although the overheads of IFC vary from test to
test, the average overheads over our baseline are 121\% and 45\% for
basic IFC and optimized IFC, respectively. The test \emph{regexp} has
almost zero overhead because it spends most time in native code, which
we have not yet instrumented. We also note that, as expected, the
JIT-enabled configuration performs extremely well on the SunSpider
benchmarks.



\paragraph{Macrobenchmarks}
We measured the execution time of the intial JS on 9 popular English
language Websites.  We load each Website in Safari and measure the
total time taken to \emph{execute} the JS code without user
interaction. This excludes time for network communication and internal
browser events and establishes a very conservative baseline. The
results, normalized to our baseline, are shown in
\Figure{fig:macro}. Our overheads are all less than 42\% (with an
average of around 29\% in both instrumentations).  Interestingly, we
observe that our optimization is less effective on real websites
indicating that real JS accesses the heap more often than the
SunSpider tests. When compared to the amount of time it takes to fetch
a page over the network and to render it, these overheads are
negligible. Enabling JIT worsens performance compared to our baseline
indicating that, for the code executed here, JIT is not useful.

We also experimented with JSBench~\cite{JSbench}, a sophisticated
benchmark derived from JS code in the wild. The average
overhead on all JSBench tests (a total 23 iterations) is approximately
38\% for both instrumentations. The average time for running the
benchmark tests on the uninstrumented interpreter with JIT disabled
was about 636.11ms with a standard deviation of about 0.30\% of the
mean. The average time for running the same benchmark tests on the
instrumented interpreter and the optimized version was about 874.31ms
and 880.85ms respectively, with a standard deviation of about 4.09\% 
and 5.04\% of the mean in the two cases. 




\begin{figure}[tb]
\centering
\includegraphics[height=7cm]{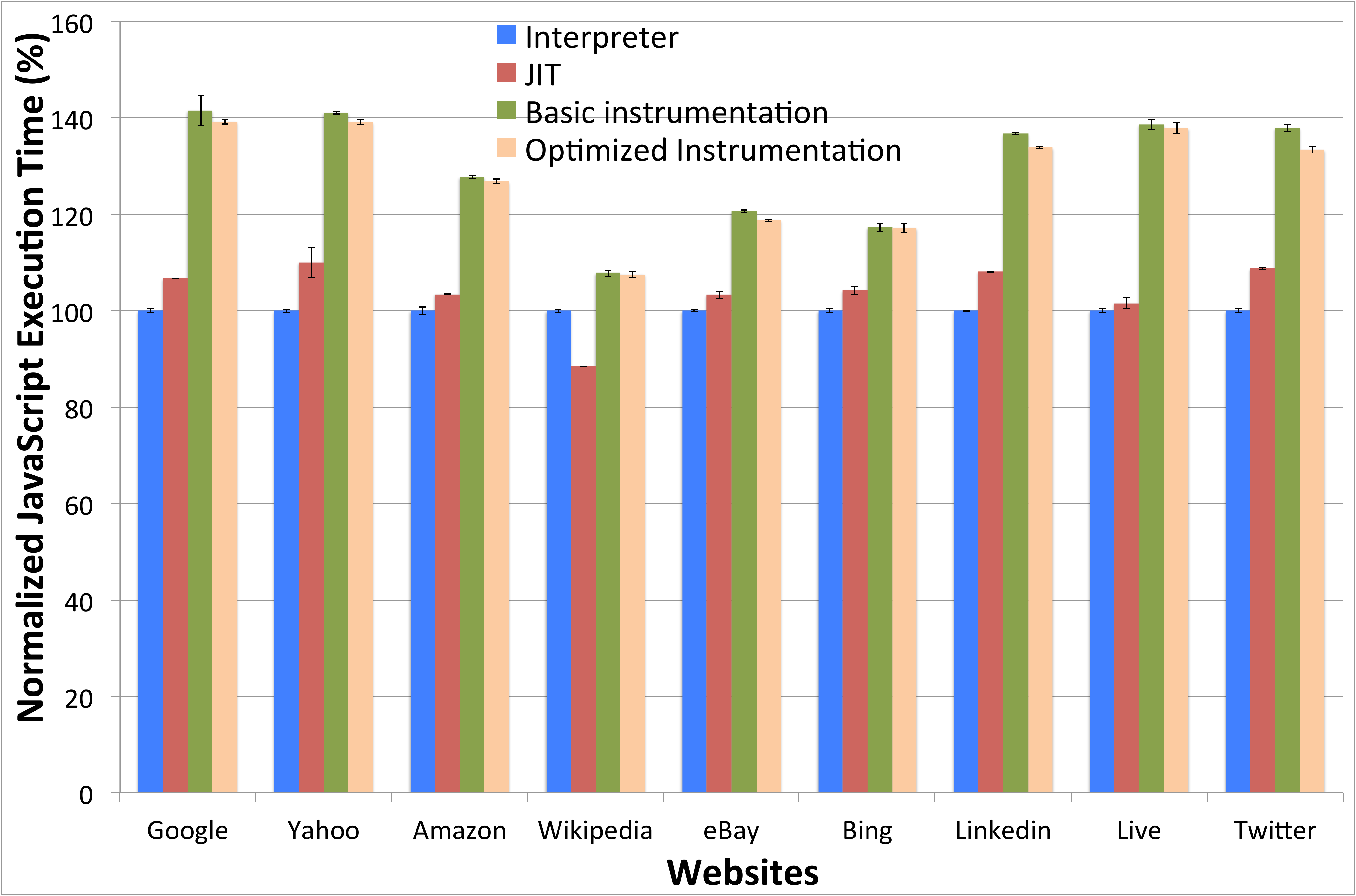}
\vspace*{-3mm}
\caption{Overheads of basic and optimized IFC in real websites}
\vspace*{-4mm}\label{fig:macro}
\end{figure}


\section{Conclusion and Future Work}
\label{sec:conc}

We have explored dynamic information flow control for JS
bytecode in WebKit, a production JS engine. We
formally model the bytecode, its semantics, our instrumentation and
prove the latter correct. We identify challenges, largely arising from
pervasive use of unstructured control flow in bytecode, and resolve
them using very limited static analysis. Our evaluation indicates only
moderate overheads in practice.

In ongoing work, we are instrumenting the DOM and other native
JS methods. We also plan to generalize our model and
non-interference theorem to take into account the reactive nature of
Web browsers. Going beyond non-interference, the design and
implementation of a policy language for representing allowed
information flows looks necessary.

\paragraph{Acknowledgments} This work was funded in part by the
Deutsche Forschungsgemeinschaft (DFG) grant ``Information Flow Control
for Browser Clients'' under the priority program ``Reliably Secure
Software Systems'' (RS3) and the German Federal Ministry of
Education and Research (BMBF) within the Centre for IT-Security,
Privacy and Accountability (CISPA) at Saarland University.





\section{Appendix}\label{sec:appendix}
\setcounter{theorem}{0}
\setcounter{definition}{0}
\setcounter{myLemma}{0}
\setcounter{mydef}{0}
\setcounter{myaxiom}{0}

\newcommand{\eq}{\sim^\beta}
\newtheorem{myCor}[myThm]{Corollary}

\subsection{Data Structures}\label{sec:app-ds}
\begin{figure}
\input{./ds.tex}
\caption{Data Structures}
\label{fig:ds}
\end{figure}

The formal model described in Section 5 was typechecked in C++. The
various data structures used for defining the functions
used in the semantics of the language are given in
Figure~\ref{fig:ds}. The source-code of the
JavaScript program is represented as a structure containing the source
and a Boolean flag indicating the \emph{strict mode} is set or
not. The instruction as indicated before is a structure consisting of
the \emph{opcode} and the list of \emph{operands}. The \emph{opcode}
is a string indicating the operation and the \emph{operand} is a union
of \emph{registerIndex, immediateValue, identifier, Boolean, funcIndex} and \emph{offset}.
The \emph{immediateValue} denotes the directly supplied value to an
opcode, \emph{registerIndex} is the index of the register containing the value
to be operated upon, \emph{identifier} represents the string name
directly used by the opcode, \emph{Boolean} is a often a flag
indicating the truth value of some parameter and \emph{offset}
represents the offset where the control jumps to.
Similarly, \emph{functionIndex} indicates the index
 of the function object being invoked.

The function's source code is represented in the form of a control flow
graph (CFG). Formally, it is defined as a struct with a list of CFG
nodes, each of which contain the instructions that are to
be performed and the edges point to the next instruction in the
program. Multiple outgoing edges indicate a branching
instruction. It also contains variables indicating the number of
variables used by the function code and a reference to the
\emph{globalObject}. 

The labels are interpreted as a structure consisting of long integer
\emph{label}. The \emph{label} represents the value of the label,
which are interpreted as bit vectors. A special label \emph{star}
which represents partially leaked data, is used for deferred
no-sensitive-upgrade check.  The program counter ($pc$) is 
represented as a stack of $pc$-nodes, each of which contains the context
label and the $\IPD$ of the operation that pushed the node, the
callframe of the current node and the handler flag indicating the
presence of an exception handler.

Different types of values are used as operands for performing the
operations. They include Boolean, integer, string, double and
objects or special values like NaN or \emph{undefined}. These values
are associated with a label each and are wrapped by the JSValue
class. All the values used in the data structures have the type
JSValue. The objects consist of properties, a prototype chain pointer
with an associated label and a structure label for the object. The
properties are represented as a structure of the \emph{propertyName}
and its \emph{descriptor}. The descriptor of the property contains the
\emph{value}, some Boolean flags and a property label. 
The heap is a collection of objects with an associated memory
address. It is essentially a map from location to object. 

There are subtypes of JSObject that define the function object and the
global object. The function object contains a pointer to the
associated CFG and the scope chain. It also contains a field defining
the type of function it represents, namely, host or user-defined. 

The call-stack is made up of various call-frame nodes, each of which
contains a set of registers, the associated CFG, the return address of
the function, a pointer to the scope chain, and an exception
table. The registers store values and objects and are used as operands
for performing the operations. The exception table contains the
details about the handlers associated with different instructions in
the CFG of the call-frame. The scope chain is a list of nodes
containing objects or activation objects along with a label indicating
the context in which the object in that node was added. The activation
object is a structure containing a pointer to the call-frame node for
which the activation object was created.

The next section defines the different procedures used in the
semantics of the language. The statement \textbf{stop} implies that the
program execution hangs.

\subsection{Algorithms}\label{sec:app-algo}
The different meta-functions used in the semantics presented in Section~\ref{sec:app-sem}
are described below: 

\scriptsize
\begin{algorithmic}
\Procedure {\emph{isInstanceOf}}{JSLabel context, JSValue obj,
  JSValue protoVal}
\State {oProto := obj.prototype.\_\_proto\_\_}
\While {oProto}
 \If {oProto = protoVal}
   \State {ret := JSValue::construct(true)}
   \State {ret.label := context}
   \State \Return {ret}
 \EndIf
   \State {oProto := oProto.prototype.\_\_proto\_\_}
   \State {context := context.Join(oProto.prototype.l)}
\EndWhile
   \State {ret := JSValue::construct(false)}
   \State {ret.label := context}
   \State \Return {ret}
\EndProcedure
\\ \\
\Procedure {\emph{opRet}}{CallFrameStack* callStack, int ret}
\State {JSValue retValue = callStack.cFN.rf[ret].value}
\If {hostCallFrameFlag} 
   \State {callStack.pop()}
   \State \Return {nil, callStack, retValue}
\EndIf
\State {callStack.pop()}
\State \Return {callStack.cFN.returnAddress, callStack, retValue}
\EndProcedure
 \\ \\

\Procedure {\emph{opCall}}{CallFrameStack* callStack, CFGNode* ip,
		int func, int argCount} 
	\State {JSValue funcValue := callStack.cFN.rf[func].value}
	\State {JSFunctionObject fObj}
	\State {CallFrameNode *sigmaTop := new CallFrameNode()}
        \State {CallFrameNode *prevTop := callStack.top()}
	\State {callStack.push(sigmaTop)}
	\State  {CallType callType := getCallData(funcValue, \&fObj)}
	\If {callType = CallTypeJS} 
		\State  {ScopeChainNode* sc := fObj.scopeChain}
		\State  {callStack.cFN.cfg := fObj.funcCFG}
		\State  {callStack.cFN.returnAddress := ip.Succ}
		\State  {callStack.cFN.sc := sc}
		\State  {args.noOfArgs := argCount}
		\State  {callStack.cFN.argCount := argCount}
		\For {i $\gets$ 0, argCount}
                          \State  {callStack.cFN.rf[sigmaTop.baseRegister() +
                            i].value := }
                          \State  {~~~prevTop.rf[prevTop.headRegister()-i].value}
                \EndFor

		\State  {ip := callStack.cFN.cfg.cfgNode}
        \ElsIf {callType = CallTypeHost} 
                \State  {\textbf{stop}}
		\Comment  {Not modeled}
	\EndIf
	\State  {retState.ip := ip}
	\State  {retState.sigma := callStack}
	\State  \Return {retState}
\EndProcedure
\\ \\
\Procedure {\emph{opCallEval}}{JSLabel contextLabel, CallFrameStack* callStack,
		CFGNode* ip, int func, int argCount} 
	\State {JSValue funcValue := callStack.cFN.rf[func].value}
	\State {JSFunctionObject* fObj}
	\State {JSObject* variableObject}
	\State {Argument* arguments}
	\If {isHostEval(funcValue)}
		\State{ScopeChainNode* sc := fObj.scopeChain}
		\State{callStack.cFN.returnAddress := ip + 1}
		\State{callStack.cFN.sc := sc}
		\State{callStack.cFN.argCount := argCount}

		\State{SourceCode progSrc := funcValue.getSource()}
		\State{Compiler::preparse(progSrc)}

		\State{CFG* evalCodeBlock := Compiler::compile(progSrc)}

		\State{unsigned numVars := evalCodeBlock.numVariables()}
		\State{unsigned numFuncs := evalCodeBlock.numFuncDecls()}

		\If {numVars $||$ numFuncs}
			\If {evalCodeBlock.strictMode}
				\State{JSActivation* variableObject := new JSActivation()}
				\State{variableObject.create(callStack)}
				\State{SChainObject* scObj}
				\State{scObj.actObj := variableObject}
				\State{sc.push(scObj, variableObject, contextLabel)}
			\Else
				\For {(ScopeChainNode* n := sc;; n := n.next)}
					\If {n.isVariableObject() \&\& !n.isLexicalObject()}
						\State {variableObject := n.getObject()}
						\State {\textbf{break}}
					\EndIf
				\EndFor
			\EndIf

			\For {i $\gets$ 0, numVars}
				\State{Identifier iden := evalCodeBlock.variable(i)}
                                \If {!variableObject.hasProperty(iden)}
					\State{variableObject.insertVariable(iden)}
                                \EndIf
			\EndFor

			\For {i $\gets$ 0, numFuncs}
				\State{JSFunctionObject* fObj := evalCodeBlock.funcDecl(i)}
				\State{variableObject.insertFunction(fObj)}
			\EndFor
		\EndIf

		\State{callStack.cFN.cfg := evalCodeBlock}

		\State{ip := evalCodeBlock.cfgNode}
		\State{retState.ip := ip}
		\State{retState.sigma := callStack}
                \State \Return {retState}
	\Else
		\State \Return {opCall(contextLabel, callStack, ip, func, argCount)}
	\EndIf
\EndProcedure
\\ \\

\Procedure {\emph{createArguments}}{Heap* h, CallFrameStack* callStack}
       \State{JSObject* jsArgument := JSArgument::create(h, callStack)}
	\State{h.o[++(h.location)] := *jsArgument}
	\State{retState.theta := h}
	\State{retState.val := JSValue::construct(jsArgument)}
	\State \Return {retState}
\EndProcedure
\\ \\
\Procedure {\emph{newFunc}}{CallFrameStack* callStack, Heap* heap, int
  funcIndex, JSLabel context}
	\State{CFG* cBlock := callStack.cFN.cfg}
        \State{SourceCode fcCode := cBlock.getFunctionSrc(funcIndex)}
	\State{CFG* fcBlock := Compiler::compile(fcCode,callStack.cFN)}
	\State{JSFunctionObject* fObj := JSFunctionObject::create(
          fcBlock,callStack.cFN.sc)}
        \State{fObj.structLabel := context}
	\State{heap.o[++(heap.location)] := *fObj}
	\State{retState.theta := heap}
	\State{retState.val := JSValue::construct(fObj)}
	\State \Return {retState}
\EndProcedure
\\ \\
\Procedure {\emph{createActivation}}{CallFrameStack* callStack, JSLabel contextLabel}
	\State{JSActivation* jsActivation := new JSActivation()}
	\State{jsActivation.create(callStack)}
	\State{jsActivation.structLabel := contextLabel}
	\State{SChainObject* scObj}
	\State{scObj.actObj := jsActivation}

	\State{JSValue vActivation := JSValue::jsValuefromActivation (jsActivation)}
	\If {callStack.cFN.scopeLabel $\geq$ contextLabel}
		\State{callStack.cFN.sc.push(scObj, VariableObject, contextLabel)}
		\State{callStack.cFN.scopeLabel := contextLabel}
	\Else
		\State {\textbf{stop}}
        \EndIf
	\State \Return {retState}
\EndProcedure
\\ \\
\Procedure {\emph{createThis}}{JSLabel contextLabel, CallFrameStack* callStack, Heap* h}
	\State{JSFunctionObject* callee := callStack.cFN.callee}
	\State{PropertySlot p(callee)}
	\State{String str := "prototype"}
	\State{JSValue proto := p.getValue(str)}
	\State{JSObject* obj := new JSObject()}
	\State{obj.structLabel := contextLabel}
	\State{obj.prototype.\_\_proto\_\_ := proto.toObject()}
	\State{obj.prototype.l := proto.toObject().structLabel.join( contextLabel)}
	\State{h.o[++(h.location)] := *obj}
	\State{retState.theta := h}
	\State{retState.val := JSValue::construct(obj)}
	\State \Return {retState}
\EndProcedure
\\ \\
\Procedure {\emph{newObject}}{Heap* h, JSLabel contextLabel}
	\State{JSObject* obj := emptyObject()}
        \State{obj.structLabel := contextLabel}
	\State{obj.prototype.\_\_proto\_\_ := ObjectPrototype::create()}
	\State{obj.prototype.l := contextLabel}
	\State{h.o[++(h.location)] := *obj}
	\State{retState.theta := h}
	\State{retState.val := JSValue::construct(obj)}
	\State \Return {retState}
\EndProcedure
\\ \\
\Procedure {\emph{getPropertyById}}{JSValue v, String p, int dst}
	\State{JSObject* O := v.toObject()}
	\State{JSLabel label := O.structLabel}
	\State{JSValue ret := jsUndefined()}

	\If {O.isUndefined()}
		\State{ret.label := label}
		\State \Return {ret}
	\EndIf
	\While {O $\neq$ \textbf{null}}
		\If {O.containsProperty(p)}
                        \If {p.isGetter()}
                                 \State{JSValue v = p.getValue()}
                                 \State{JSFunctionObject* funcObj =
                                   (JSFunctionObject*) v.toObject()}
                                 \State{CallFrameNode *sigmaTop = new CallFrameNode()}
                                 \State{callStack.push(sigmaTop)}
                                 \State{ScopeChainNode* sc = fObj.scopeChain}
                                 \State{CFG* newCodeBlock = fObj.funcCFG}
                                 \State{callStack.cFN.cfg = *newCodeBlock}
                                 \State{callStack.cFN.returnAddress = ip + 1}
                                 \State{callStack.cFN.sc = sc}
                                 \State{callStack.cFN.getter = true}
                                 \State{callStack.cFN.dReg = dst}
                                 \State{ip = newCodeBlock.cfgNode}
                                 \State{interpreter.iota = ip}
                                 \State{interpreter.sigma = callStack}
                        \Else
                                 \State{ret := getProperty(p).getValue()}
                                 \State{ret.label := label}
                        \EndIf
                        \State \Return {ret}
		\Else
			\State{O := O.prototype.\_\_proto\_\_}
		\EndIf
		\State{label := label.join(label)}
	\EndWhile
\EndProcedure 
\\ \\
\Procedure {\emph{putDirect}}{JSLabel contextLabel, CallFrameStack* callStack, Heap* h,
		int base, String property, int propVal}
	\State{JSValue baseValue := callStack.cFN.rf[base]. value}
	\State{JSValue propValue := callStack.cFN.rf[propVal]. value}
	\State{JSObject* obj := baseValue.toObject()}
	\State{PropertyDescriptor dataPD := PropertyDescriptor::createPD(true, true, true)}
	\State{dataPD.value := propValue}
	\State{obj.setProperty(property, dataPD)}
	\State{obj.structLabel := obj.structLabel.join(contextLabel)}
	\State{h.o[++(h.location)] := *obj}
	\State \Return {h}
\EndProcedure
\\ \\
\Procedure {\emph{putIndirect}}{JSLabel contextLabel, CallFrameStack* callStack, Heap* h,
		int base, String property, int val}
	\State{JSValue baseValue := callStack.cFN.rf[base].value}
	\State{JSValue propValue := callStack.cFN.rf[val].value}
	\State{JSObject* obj := baseValue.toObject()}

	\State{bool isStrict := callStack.cFN.cfg.isStrictMode()}
	\State{contextLabel := obj.structLabel.join(contextLabel)}

	\If {obj.containsPropertyInItself(property)
			\&\& obj. getProperty(property).isDataProperty() \&\& !isStrict
			\&\& obj.isWritable()}
		\State{obj.getProperty(property).setValue(propValue)}
		\State{h.o[++(h.location)] := *obj}
		\State \Return {h}
	\EndIf
	\State \Return {putDirect(contextLabel, callStack, h, base, property, val)}
\EndProcedure
\\ \\
\Procedure {\emph{delById}}{JSLabel contextLabel, CallFrameStack* callStack, Heap* h,
		int base, Identifier property}
	\State{JSValue baseValue := callStack.cFN.rf[base].value}
	\State{JSObject* obj := baseValue.toObject()}
	\State{int loc := h.findObject(obj)}
	\State{Property prop := obj.getProperty(property)}
	\State{PropertyDescriptor pd := prop.getPropertyDescriptor()}

	\If {obj.getPropertyValue(prop).label $\geq$ contextLabel}
		\If {!obj.containsPropertyInItself(property)}
			\State {retState.theta := h}
			\State {retState.val := JSValue::construct(true)}
			\State \Return {retState}
		\EndIf
		\If {obj.containsPropertyInItself(property) \&\&
                  prop. isConfigurable()}
			\If{!(callStack.cFN.cfg.isStrictMode())}
			\State{pd.value := JSValue::constructUndefined()}
			\State{obj.setProperty(property, pd)}
			\State{h.o[loc] := *obj}
			\State{retState.theta := h}
			\State{retState.val := JSValue::construct(true)}
			\State \Return {retState}
                        \EndIf
		\EndIf
		\State{retState.theta := h}
		\State{retState.val := JSValue::construct(false)}
		\State \Return {retState}
	\Else
		\State {\textbf{stop}}
	\EndIf
\EndProcedure
\\ \\
\Procedure {\emph{putGetterSetter}}{JSLabel contextLabel, CallFrameStack* callStack, Heap* h,
		int base, Identifier property, JSValue getterValue, JSValue setterValue}
	\State{JSValue baseValue := callStack.cFN.rf[base].value}
	\State{JSObject* obj := baseValue.toObject()}
	\State{int loc := h.findObject(obj)}
	\State{JSFunctionObject *getterObj, *setterObj}
	\State{JSFunctionObject *getterFuncObj := \textbf{null}, *setterFuncObj := \textbf{null}}
	\If {!getterValue.isUndefined()}
		\State{getterFuncObj := getterValue.toFunctionObject (callStack.cFN.cfg, callStack.cFN.sc)}
	\EndIf
	\If {!setterValue.isUndefined()}
		\State{setterFuncObj := setterValue.toFunctionObject (callStack.cFN.cfg, callStack.cFN.sc)}
	\EndIf

	\If {getterFuncObj $\neq$ \textbf{null}}
		\State{obj.setGetter(property, getterObj)}
        \EndIf
	\If {setterFuncObj $\neq$ \textbf{null}}
		\State{obj.setSetter(setterObj)}
        \EndIf
	\State{PropertyDescriptor accessor := PropertyDescriptor ::createPD(false, false,
			false, true)}
	\State{JSValue v := JSValue::constructUndefined()}
	\State{v.label := contextLabel}
	\State{accessor.value := v}
	\State{obj.setProperty(property, accessor)}
	\State{obj.structLabel := contextLabel}
	\State{h.o[loc] := *obj}
	\State \Return {h}
\EndProcedure
\\ \\
\Procedure {\emph{getPropNames}}{CallFrameStack* callStack, Instruction* ip, int base,
		int i, int size, int breakOffset}
	\State{JSValue baseVal := callStack.cFN.rf[base].value}
	\State{JSObject* obj := baseVal.toObject()}
	\State{PropertyIterator* propItr := obj.getProperties()}

	\If {baseVal.isUndefined() $||$ baseVal.isNull()}
		\State{retState.v1 := jsUndefined()}
		\State{retState.v2 := jsUndefined()}
		\State{retState.v3 := jsUndefined()}
		\State{retState.ip := ip + breakOffset}
		\State \Return {retState}
	\EndIf

	\State{retState.v1 := JSValue::construct(propItr)}
	\State{retState.v2 := JSValue::construct(0)}
	\State{retState.v3 := JSValue::construct(propItr.size())}
	\State{retState.ip := ip + 1}
	\State \Return {retState}
\EndProcedure
\\ \\
\Procedure {\emph{getNextPropName}}{CallFrameStack* cStack,
		Instruction* ip, JSValue base, int i, int size, int iter, int offset,
		int dst}
	\State{JSObject* obj := base.toObject()}
	\State{PropertyIterator* propItr := cStack.cFN.rf[iter].value. toPropertyIterator()}
	\State{int b := rFile[i].value.toInteger()}
	\State{int e := rFile[size].value.toInteger()}
	\While {b $<$ e}
		\State{String key := propItr.get(b)}
		\State{retState.value1 := JSValue::construct(b + 1)}
		\If {!(key.isNull())}
			\State{retState.value2 := JSValue::construct(key)}
			\State{ip := ip + offset}
			\State{\textbf{break}}
		\EndIf
		\State{b++}
	\EndWhile
        \State \Return {retState}
\EndProcedure
\\ \\
\Procedure {\emph{resolveInSc}}{JSLabel contextLabel, ScopeChainNode* scopeHead,
		String property}
	\State{JSValue v}
	\State{JSLabel l}
	\State{ScopeChainNode* scn := scopeHead}

	\While {scn $\neq$ NULL}
		\State{PropertySlot pSlot := scn.getPropertySlot()}
		\If {pSlot.contains(property)}
			\State{v := pSlot.getValue(property)}
			\State{v.label := contextLabel}
			\State \Return {v}
		\EndIf
		\State{scn := scn.next}
		\If {scn.scObjType = VariableObject}
			\State{contextLabel =
                          contextLabel.join(scn.Object. actObj.structLabel)}
		\ElsIf {scn.scObjType = LexicalObject}
			\State{contextLabel =
                          contextLabel.join(scn.Object. obj.structLabel)}
                \EndIf
		\State{contextLabel := contextLabel.join(scn. scopeNextLabel)}
	\EndWhile

	\State{v := jsUndefined()}
	\State{v.label := contextLabel}
	\State \Return {v}
\EndProcedure
\\ \\
\Procedure {\emph{resolveInScWithSkip}}{JSLabel contextLabel, ScopeChainNode* scopeHead,
		String property, int skip}
	\State{JSValue v}
	\State{JSLabel l}
	\State{ScopeChainNode* scn := scopeHead}

	\While {skip$-$$-$}
		\State{scn := scn.next}
		\If {scn.scObjType = VariableObject}
			\State{contextLabel := contextLabel.join(
                          scn.Object. actObj.structLabel)}
		\ElsIf {scn.scObjType = LexicalObject}
			\State{contextLabel :=
                          contextLabel.join(scn.Object. obj.structLabel)}
		\EndIf
                \State{contextLabel := contextLabel.join(scn. scopeNextLabel)}
	\EndWhile

	\While {scn $\neq$ \textbf{null}}
		\State{PropertySlot pSlot := scn.getPropertySlot()}
		\If {pSlot.contains(property)}
			\State{v := pSlot.getValue(property)}
			\State{v.label := contextLabel}
			\State \Return {v}
		\EndIf
		\State{scn := scn.next}
		\If {scn.scObjType = VariableObject}
			\State{contextLabel :=
                          contextLabel.join(scn.Object. actObj.structLabel)}
		\ElsIf {scn.scObjType = LexicalObject}
			\State{contextLabel :=
                          contextLabel.join(scn.Object. obj.structLabel)}
		\EndIf
		\State{contextLabel := contextLabel.join(scn. scopeNextLabel)}
	\EndWhile 

	\State{v := jsUndefined()}
	\State{v.label := contextLabel}
	\State \Return {v}
\EndProcedure
\\ \\
\Procedure {\emph{resolveGlobal}}{JSLabel contextLabel, 
		CallFrameStack* cStack, String property}
	\State{JSValue v}
	\State{struct CFG* cBlock := cStack.cFN.cfg}
	\State{JSGlobalObject* globalObject := cBlock. getGlobalObject()}
	\State{PropertySlot pSlot(globalObject)}
	\If {pSlot.contains(property)}
		\State{v := pSlot.getValue(property)}
		\State{v.label := contextLabel}
		\State \Return {v}
	\EndIf	
        \State{v := jsUndefined()}
	\State{v.label := contextLabel}
	\State \Return {v}
\EndProcedure
\\ \\
\Procedure {\emph{resolveBase}}{JSLabel contextLabel, 
		CallFrameStack* cStack, ScopeChainNode* scopeHead, String property,
		bool strict}
	\State{JSValue v}

	\State{ScopeChainNode* scn := scopeHead}
	\State{CFG *cBlock := cStack.cFN.cfg}
	\State{JSGlobalObject *gObject := cBlock.globalObject}

	\While {scn $\neq$ \textbf{null}}
		\State{JSObject* obj := scn.get()}
		\State{contextLabel := obj.structLabel.join(contextLabel)}
		\State{PropertySlot pSlot(obj)}
		\If {scn.next = \textbf{null} \&\& strict \&\&
                  !pSlot.contains (property)}
			\State{v := emptyJSValue()}
			\State{v.label := contextLabel}
			\State \Return {v}
		\EndIf
		\If {pSlot.contains(property)}
			\State{v := JSValueContainingObject(obj)}
			\State{v.label := contextLabel}
			\State \Return {v}
		\EndIf
		\State{scn := scn.next}
		\If {scn $\neq$ \textbf{null}}
			\State{contextLabel :=
                          contextLabel.join(scn. scopeNextLabel)}
                \EndIf
	 \EndWhile

	\State{v := JSValue::construct(gObject)}
	\State{v.label := contextLabel}
	\State \Return {v}
\EndProcedure
\\ \\ 
\Procedure {\emph{resolveBaseAndProperty}}{JSLabel contextLabel, 
		CallFrameStack cStack, int bRegister, int pRegister, String property}
	\State{JSValue v}
	\State{ScopeChainNode* scn := cStack.cFN.sc}

	\While {scn $\neq$ \textbf{null}}
		\State{JSObject* obj := scn.get()}
		\State{contextLabel := obj.structLabel.join(contextLabel)}
		\State{PropertySlot pSlot(obj)}
		\If {pSlot.contains(property)}
			\State{v := pSlot.getValue(property)}
			\State{v.label := contextLabel}
			\State{ret.val1 := v}
			\State{v := JSValueContainingObject(obj)}
			\State{v.label := contextLabel}
			\State{ret.val2 := v}
			\State \Return {ret}
		\EndIf
		\State{scn := scn.next}
		\If {scn $\neq$ \textbf{null}}
			\State{contextLabel :=
                          contextLabel.join(scn. scopeNextLabel)}
                \EndIf
	\EndWhile
\EndProcedure
\\ \\
\Procedure {\emph{getScopedVar}}{JSLabel contextLabel, CallFrameStack* callStack, Heap* h,
		int index, int skip}
	\State{JSValue v}
	\State{ScopeChainNode* scn := callStack.cFN.sc}

	\While {skip$-$$-$}
		\If {scn.scObjType = VariableObject}
			\State{contextLabel :=
                          contextLabel.join(scn.Object.actObj.
                          structLabel)}
		\ElsIf{scn.scObjType = LexicalObject}
			\State{contextLabel :=
                          contextLabel.join(scn.Object.obj.
                          structLabel)}
                \EndIf
		\State{contextLabel :=
                  contextLabel.join(scn. scopeLabel)}
		\State{scn := scn.next}
	\EndWhile

	\State{v := scn.registerAt(index)}
	\If {scn.scObjType = VariableObject}
		\State{v.label :=
                  contextLabel.join(scn.Object.actObj. structLabel)}
	\ElsIf {scn.scObjType = LexicalObject}
		\State{v.label := contextLabel.join(scn.Object.obj. structLabel)}
        \EndIf
	\State \Return {v}
\EndProcedure
\\ \\
\Procedure {\emph{putScopedVar}}{JSLabel contextLabel, CallFrameStack* callStack,
		Heap* h, int index, int skip, int value}
	\State {CallFrameStack* cStack}
	\State {ScopeChainNode* scn := callStack.cFN.sc}
	\State {JSValue val := callStack.cFN.rf[value].value}
	\While {skip$-$$-$}
		\If {scn.scObjType = VariableObject}
			\State{contextLabel :=
                          contextLabel.join(scn.Object. actObj.structLabel)}
		\ElsIf {scn.scObjType = LexicalObject}
			\State{contextLabel :=
                          contextLabel.join(scn.Object. obj.structLabel)}
		\EndIf
		\State{contextLabel :=
                  contextLabel.join(scn. scopeLabel)}
		\State{scn := scn.next}
	\EndWhile
	\State{cStack := scn.setRegisterAt(contextLabel, index, val)}
	\State \Return {cStack}
\EndProcedure
\\ \\
\Procedure {\emph{pushScope}}{JSLabel contextLabel, CallFrameStack* callStack,
		Heap* h, int scope}
	\State{ScopeChainNode* sc := callStack.cFN.sc}
	\State {JSValue v := callStack.cFN.rf[scope].value}
	\State {JSObject* o := v.toObject()}
	\State {SChainObject* scObj}
        \If {sc.scopeLabel $\geq$ contextLabel}
	\State {scObj.obj := o}
	\State {sc.push(scObj, LexicalObject, contextLabel)}
	\State {callStack.cFN.sc := sc}
        \ElsIf {sc.scopeLabel = star}
	\State {scObj.obj := o}
	\State {sc.push(scObj, LexicalObject, star)}
	\State {callStack.cFN.sc := sc}
        \EndIf
	\State \Return {callStack}
\EndProcedure
\\ \\
\Procedure {\emph{popScope}}{JSLabel contextLabel, CallFrameStack* callStack,
		Heap* h}
	\State {ScopeChainNode* sc := callStack.cFN.sc}
	\State {JSLabel l := sc.scopeLabel}
	\If {l $\geq$ contextLabel}
		\State {sc.pop()}
		\State {callStack.cFN.sc := sc}
	\Else
		\State {\textbf{stop}}
        \EndIf
	\State \Return {callStack}
\EndProcedure
\\ \\
\Procedure {\emph{jmpScope}}{JSLabel contextLabel, CallFrameStack* callStack,
		Heap* h, int count}
	\State {ScopeChainNode* sc := callStack.cFN.sc}
        \While {count$-$$-$ $>$ 0}
	\State {JSLabel l := sc.scopeLabel}
	\If {l $\geq$ contextLabel}
		\State {sc.pop()}
		\State {callStack.cFN.sc := sc}
	\Else
		\State {\textbf{stop}}
        \EndIf
        \EndWhile
	\State \Return {callStack}
\EndProcedure
\\ \\
\Procedure {\emph{throwException}}{CallFrameStack* callStack, CFGNode* iota}
 	\State {CFGNode* handler}
 	\While {callStack.cFN.hasHandler()==false}
 		\State{callStack.pop()}
        \EndWhile
 	\While {callStack.cFN.sc.length() - callStack.cFN.getHandlerScopeDepth()}
 		\State{callStack.cFN.sc.pop()}
        \EndWhile
 	\State{handler := callStack.cFN.getHandler(iota)}
 	\State{interpreter.iota := handler}
 	\State{interpreter.sigma := callStack}
\EndProcedure
\end{algorithmic}
\normalsize

\subsection{Semantics}\label{sec:app-sem}
\scriptsize \begin{equation*}
\inference[prim:]
{\iota = \mbox{``prim dst:r src1:r src2:r''} \qquad \\
\mL := \Gamma(!\sigma(\src1)) \sqcup \Gamma(!\sigma(\src2)) \sqcup
\Gamma(!\rho)\qquad\mV := \Upsilon(!\sigma(\src1)) \oplus \Upsilon(!\sigma(\src2)) \\
(\Gamma(!\sigma(\dst)) \geq \Gamma(!\rho)) \Rightarrow (\mL :=
\mL)  \diamond (\mL := \star) \\
\sigma' := \sigma\Big[ \,^{\Upsilon(!\sigma(\dst)) := \mV}_{\Gamma(!\sigma(\dst)) :=\mL}
\Big]\qquad\iota' := \Succ(!\sigma'.\CFG, \iota)\qquad ~\rho' := \isIPD(\iota',\rho, \sigma')\\}
{\iota,  \theta, \sigma, \rho ~\leadsto~ \iota',  \theta,
\sigma', \rho'}
\end{equation*} \normalsize
\emph{prim} reads the values from two registers ($\src1$ and
  $\src2$), performs the binary operation generically denoted by
  $\oplus$, and writes the result into the $\dst$
  register. The label assigned to the value in $\dst$ register is the
  join of the label of value in $\src1$, $\src2$ and the head of the
  pc-stack ($!\rho$). In order to avoid implicit leak of information,
  the label of the existing value in $\dst$ is compared with the
  current context label. If the label is lower than the context label,
  the label of the value in $\dst$ is set to $\star$. 

\scriptsize \begin{align*}
\inference[mov:]
{\iota = \mbox{``mov dst:r src:r''}\qquad  \\
 \mL := \Gamma(!\sigma(\src)) \sqcup \Gamma(!\rho)\qquad~\mV=\Upsilon(!\sigma(\src)) \\
(\Gamma(!\sigma(\dst)) \geq \Gamma(!\rho)) \Rightarrow (\mL :=
\mL)  \diamond (\mL := \star) \\
 \sigma' :=\sigma\Big[ \,^{\Upsilon(!\sigma(\dst)) := \mV}_{\Gamma(!\sigma(\dst)) :=\mL}\Big]\qquad~
\iota' := \Succ(!\sigma'.\CFG,\iota)\qquad ~\rho' := \isIPD(\iota',\rho, \sigma')}
{\iota,  \theta, \sigma, \rho ~\leadsto~ \iota',  \theta,
\sigma', \rho'}
\end{align*} \normalsize 
\emph{mov} copies the value from the $\src$ register to the $\dst$
  register. The label assigned to the value in $\dst$ register is the
  join of the label of value in $\src$ and the head of the
  pc-stack ($!\rho$). In order to avoid implicit leak of information,
  the label of the existing value in $\dst$ is compared with the
  current context label. If the label is lower than the context label,
  the label of the value in $\dst$ is joined with $\star$.

\scriptsize \begin{align*}
\inference[jfalse:]
{\iota = \mbox{``jfalse cond:r target:offset''}\qquad \\
\Gamma(!\sigma(\cond)) \neq \star\qquad~
\rho'' := \rho.\push(\Gamma(!\sigma(\cond)) \sqcup \Gamma(!\rho),
\IPD(\iota), \cf(\iota),\false)\qquad\\
\Upsilon(!\sigma(\cond)) = \false \Rightarrow \iota' := \Left(!\sigma.\CFG,\iota) \diamond \iota'
:= \Right(!\sigma.\CFG,\iota)\\~\rho' := \isIPD(\iota',\rho'', \sigma)}
{\iota,  \theta, \sigma, \rho ~\leadsto~ \iota',  \theta, \sigma, \rho'
}
\end{align*} \normalsize 
\emph{jfalse} is a branching instruction. Based on the
  value in the $\cond$ register, it decides which branch to take. The
  operation is performed only if the value in $\cond$ is not labelled
  with a $\star$. If it contains a $\star$, we terminate the execution to
  prevent possible leak of information.
  The push function defined in the rule does the following: 
  A node is pushed on the top of the pc-stack containing the IPD of
  the branching instruction and the label of the value in $\cond$
  joined with the context, to define the context of this branch. If
  the IPD of the instruction is SEN or the same as the top of the
  pc-stack, then we just join the label on top of the pc-stack with
  the context label determined by the $\cond$ register. 

\scriptsize \begin{align*}
\inference[loop-if-less:]
{\iota = \mbox{``loop-if-less src1:r src2:r target:offset''}\qquad\\
\Gamma(!\sigma(\src1)) \neq \star\qquad~\Gamma(!\sigma(\src2)) \neq \star\qquad~
\mL := \Gamma(!\sigma(\src1)) \sqcup
\Gamma(!\sigma(\src2)) \sqcup \Gamma(!\rho)\qquad\\
\Upsilon(!\sigma(\src1)) <~\Upsilon(!\sigma(\src2)) \Rightarrow \iota' :=
\Left(!\sigma.\CFG, \iota) \diamond \iota' := \Right(!\sigma.\CFG, \iota)\qquad\\
\rho'' := \rho.\push(\mL,
\IPD(\iota), \cf(\iota),\false)\qquad~\rho' := \isIPD(\iota',\rho'', \sigma)
}
{\iota,  \theta, \sigma, \rho ~\leadsto~ \iota',  \theta,
\sigma, \rho'}
\end{align*} \normalsize 
\emph{loop-if-less} is another branching instruction. If the value of
$\src1$ is less than $\src2$, then it jumps to the \emph{target}, else continues
with the next instruction. The operation is performed only if the
values  in $\src1$ and $\src2$ are not labelled with a $\star$. If any
one of them contains a $\star$, we abort the execution to
prevent possible leak of information.   
The push function defined in the rule does the following: 
A node is pushed on the top of the pc-stack containing the IPD of the
branching instruction and the join of the label of the values in $\src1$ and $\src2$
joined with the context, to define the context of this branch. If
  the IPD of the instruction is SEN or the same as the top of the
  pc-stack, then we just join the label on top of the pc-stack with
  the context label determined above.

\scriptsize \begin{align*}
\inference[typeof:]
{\iota = \mbox{``typeof dst:r src:r''} \\
\mL:=(\Gamma(\src) \sqcup \Gamma(!\rho))\qquad\mV:=\mathit{determineType}(!\sigma(\src)) \\
(\Gamma(!\sigma(\dst)) \geq \Gamma(!\rho)) \Rightarrow (\mL :=
\mL)  \diamond (\mL := \star) \\
\sigma' :=
\sigma\Big[ \,^{\Upsilon(!\sigma(\dst)) := \mV}_{\Gamma(!\sigma(\dst)) :=\mL}
\Big]\qquad\iota' := \Succ(!\sigma'.\CFG,\iota)\qquad\rho' := \isIPD(\iota',\rho, \sigma')}
{\iota,  \theta, \sigma, \rho ~\leadsto~ \iota',  \theta,
\sigma', \rho'}
\end{align*} \normalsize 
\emph{typeof} determines the type string for $\src$ according to
ECMAScript rules, and puts the result in register $\dst$. We do a
deferred NSU check on $\dst$ before writing the
result in it. The \emph{determineType} function returns the data type
of the value passed as the parameter.

\scriptsize \begin{align*}
\inference[instanceof:]
{\iota = \mbox{``instanceof dst:r value:r cProt:r''} \\
v:=\mathit{isInstanceOf}(\Gamma(!\rho),!\sigma(\mathit{value),
  cProt})\qquad\mL = \Gamma(v)\qquad \mV = \Upsilon(v),\\
(\Gamma(!\sigma(\dst)) \geq \Gamma(!\rho)) \Rightarrow
(\mL := \mL)  \diamond (\mL := \star), \\
\sigma' :=
\sigma\Big[ \,^{\Upsilon(!\sigma(\dst)) := \mV}_{\Gamma(!\sigma(\dst)) :=\mL}
\Big]\qquad\iota' := \Succ(!\sigma'.\CFG,\iota) \qquad\rho' := \isIPD(\iota',\rho, \sigma')}
{\iota,  \theta, \sigma, \rho ~\leadsto~ \iota',  \theta,
\sigma', \rho'}
\end{align*} \normalsize 
\emph{instanceof} tests whether the \emph{cProt} is in the prototype chain of
the object in register \emph{value} and puts the Boolean result in the
$\dst$ register after deferred NSU check. 

\scriptsize \begin{align*}
\inference[enter:]
{\iota = \mbox{``enter''}\qquad 
\iota' := \Succ(!\sigma.\CFG, \iota)\qquad\rho' := \isIPD(\iota',\rho, \sigma)}
{\iota,  \theta, \sigma, \rho ~\leadsto~ \iota',  \theta,
\sigma, \rho'}
\end{align*} \normalsize 
\emph{enter} marks the beginning of a code block. 

\scriptsize \begin{align*}
\inference[ret:]
{\iota = \mbox{``ret res:r''}\qquad
(\iota',\sigma',\gamma) := \mathit{opRet(\sigma, res)}\qquad\rho' := \isIPD(\iota',\rho, \sigma')}
{\iota,  \theta, \sigma, \rho ~\leadsto~ \iota',  \theta,
\sigma', \rho'}
\end{align*} \normalsize 
\emph{ret} is the last instruction to be executed in a
  function. It pops the call-frame and returns the control to the
  callee's call-frame. The return value of the function is written to a
  local variable in the interpreter ($\gamma$), which can be read by
  the next instruction being executed.

\scriptsize \begin{align*}
\inference[end:]
{\iota = \mbox{``end res:r''}\qquad
\mathit{opEnd(\sigma, res)}}
{\iota,  \theta, \sigma, \rho ~\leadsto~ -}
\end{align*} \normalsize 
\emph{end} marks the end of a program. \emph{opEnd} passes the value present in
\emph{res} register to the caller (the native function that invoked
the interpreter). 

\scriptsize \begin{align*}
\inference[call:]
{\iota = \mbox{``op-call func:r args:n''} \\
\Gamma(\mathit{func}) \neq \star\qquad
(\iota', \sigma', \mH, \ell_f):= \mathit{opCall}(\sigma, \iota, \mathit{func},
\mathit{args})\qquad\\
\rho'' := \rho.\push(\ell_f \sqcup \Gamma(!\sigma(\mathit{func}))
\sqcup \Gamma(!\rho),\IPD(\iota), \cf(\iota), \mH)\qquad
\rho' := \isIPD(\iota',\rho'', \sigma')}
{\iota,  \theta, \sigma, \rho ~\leadsto~ \iota', \theta, \sigma',
\rho'}
\end{align*} \normalsize 
\emph{call}, initially, checks the function object's label for
  $\star$ and if the label contains a $\star$, the program execution
  is aborted. The reason for termination is the
  possible leak of information as explained above. If not,
  call creates a new call-frame, copies the arguments,
  initializes the registers, scope-chain pointer, codeblock and the
  return address. 
  The registers are initialized to \emph{undefined}
  and assigned a label obtained by joining the label of the context in
  which the function was created and the label of the function object
  itself. We treat \emph{call} as a branching instruction and hence, push
  a new node on the top of the pc-stack with the label determined
  above along with its IPD and call-frame. The
  field $\mH$ in the push function is determined by looking up the
  exception table. If it contains an associated exception handler, it
  sets the field to $\true$ else it is set to $\false$. 
  If the IPD is the SEN then we just join the label
  on the top of the stack with the currently calculated label. It then points
  the instruction pointer to the first instruction of the new code block.

\scriptsize \begin{align*}
\inference[call-put-result:]
{\iota = \mbox{``call-put-result res:r''} \\
\mL := \Gamma(\gamma) \sqcup \Gamma(!\rho)\qquad\mV := \Upsilon(\gamma) \\
(\Gamma(!\sigma(\mathit{res})) \geq \Gamma(!\rho)) \Rightarrow \mL := \mL
\diamond \mL := \star \\
\sigma' :=
\sigma\Big[ \,^{\Upsilon(!\sigma(\mathit{res})) := \mV}_{\Gamma(!\sigma(\mathit{res})) :=\mL}
\Big]\qquad\iota' := \Succ(!\sigma'.\CFG,\iota)\qquad\rho' := \isIPD(\iota',\rho, \sigma') }
{\iota,  \theta, \sigma, \rho ~\leadsto~ \iota',  \theta,
\sigma', \rho'}
\end{align*} \normalsize 
\emph{call-put-result} copies the return value $\gamma$ to the \emph{res}
  register. The label assigned to the value in \emph{res} register is the
  join of the label of the return value and the head of the
  pc-stack. In order to avoid implicit leak of information,
  \emph{deferred no-sensitive-upgrade} is performed.

\scriptsize \begin{align*}
\inference[call-eval:]
{\iota = \mbox{``call-eval func:r args:n''} \\
\Gamma(!\sigma(\mathit{func})) \neq \star\qquad
(\iota', \sigma', \mH, \ell_f) := \mathit{opCallEval(\Gamma(!\rho), \sigma, \iota,
  func, args)}\\
\rho'' := \rho.\push(\ell_f \sqcup \Gamma(!\sigma(\mathit{func})) \sqcup
\Gamma(!\rho),\IPD(\iota), \cf(\iota), \mH)
\qquad\rho' := \isIPD(\iota',\rho'', \sigma')}
{\iota,  \theta, \sigma, \rho ~\leadsto~ \iota', \theta,
\sigma', \rho'}
\end{align*} \normalsize 
\emph{call-eval} calls a function with the string passed as an
argument converted to a code block. If \emph{func} register contains
the original global eval function, then it is performed in local
scope, else it is similar to \emph{call}.

\scriptsize \begin{align*}
\inference[create-arguments:]
{\iota = \mbox{``create-arguments dst:r''} \\
 (\theta',v) := \mathit{createArguments(\theta, \sigma)}\qquad 
\mL := \Gamma(!\rho)\qquad\mV := \Upsilon(v)\\
(\Gamma(!\sigma(\dst)) \geq \Gamma(!\rho)) \Rightarrow \mL := \mL
\diamond \mL := \star \\
\sigma' :=
\sigma\Big[ \,^{\Upsilon(!\sigma(\dst)) := \mV}_{\Gamma(!\sigma(\dst)) :=\mL}
\Big]\qquad\iota' := \Succ(!\sigma'.\CFG,\iota)\qquad\rho' := \isIPD(\iota',\rho, \sigma')}
{\iota,  \theta, \sigma, \rho ~\leadsto~ \iota',  \theta',
\sigma', \rho'}
\end{align*} \normalsize 
\emph{create-arguments} creates the arguments object and places its
pointer in the local $\dst$ register after the deferred NSU check. The
label of the arguments object is set to the context. 

\scriptsize \begin{align*}
\inference[new-func:]
{\iota = \mbox{``new-func dst:r funcIndex:f''} \\
(\theta',v) := \mathit{newFunc(\sigma, \theta, funcIndex, \Gamma(!\rho))}\qquad\mL :=
\Gamma(v) \sqcup \Gamma(!\rho)\qquad
\mV := \Upsilon(v) \\
(\Gamma(!\sigma(\dst)) \geq \Gamma(!\rho)) \Rightarrow \mL := \mL
\diamond \mL := \star \\
\sigma' :=
\sigma\Big[ \,^{\Upsilon(!\sigma(\dst)) := \mV}_{\Gamma(!\sigma(\dst)) :=\mL}
\Big]\qquad\iota' := \Succ(!\sigma'.\CFG,\iota)\qquad\rho' := \isIPD(\iota',\rho, \sigma') }
{\iota,  \theta, \sigma, \rho ~\leadsto~ \iota',  \theta',
\sigma', \rho'}
\end{align*} \normalsize 
\emph{new-func} constructs a new function instance from function at
\emph{funcIndex} and the current scope chain and puts the result in
$\dst$ after deferred NSU check. 


\scriptsize \begin{align*}
\inference[create-activation:]
{\iota = \mbox{``create-activation dst:r''}\\
(\sigma',v) := \mathit{createActivation(\sigma, \Gamma(!\rho))}\qquad\mL :=
\Gamma(v) \sqcup \Gamma(!\rho)\qquad
\mV := \Upsilon(v) \\
(\Gamma(!\sigma(\dst)) \geq \Gamma(!\rho)) \Rightarrow \mL := \mL
\diamond \mL := \star \\
\sigma' :=
\sigma\Big[ \,^{\Upsilon(!\sigma''(\dst)) := \mV}_{\Gamma(!\sigma''(\dst)) :=\mL}
\Big]\qquad\iota' := \Succ(!\sigma'.\CFG,\iota)\qquad\rho' := \isIPD(\iota',\rho, \sigma')}
{\iota,  \theta, \sigma, \rho ~\leadsto~ \iota',  \theta,
\sigma', \rho'}
\end{align*} \normalsize 
\emph{create-activation} creates the activation object for the current
call-frame if it has not been already created and writes it to the
$\dst$ after the deferred NSU check and pushes the object in the
scope-chain. If the label of the head of the existing 
scope-chain is less than the context, then the label of the pushed
node is set to $\star$, else it is set to the context.

\scriptsize \begin{align*}
\inference[construct:]
{\iota = \mbox{``construct func:r args:n''} \\
\Gamma(!\sigma(\mathit{func})) \neq \star\qquad
(\iota', \sigma', \mH, \ell_f) := \mathit{opCall(\sigma, \iota, func, args)}\\
\rho'' := \rho.\push(\ell_f \sqcup \Gamma(!\sigma'(\mathit{func}))
\sqcup \Gamma(!\rho),\IPD(\iota), \cf(\iota), \mH)\qquad
\rho' := \isIPD(\iota',\rho'', \sigma')}
{\iota,  \theta, \sigma, \rho ~\leadsto~ \iota', \theta,
\sigma', \rho'}
\end{align*} \normalsize 
\emph{construct} invokes register \emph{func} as a constructor and is
similar to \emph{call}. For JavaScript functions, the \emph{this}
object being passed (the first argument in the list of arguments) is a
new object. For host constructors, no
\emph{this} is passed.

\scriptsize \begin{align*}
\inference[create-this:]
{\iota = \mbox{``create-this dst:r''} \\
(\theta',v) := \mathit{createThis(\Gamma(!\rho), \sigma, \theta)}\qquad\mL :=
\Gamma(v) \sqcup \Gamma(!\rho)\qquad
\mV := \Upsilon(v) \\
(\Gamma(!\sigma(\dst)) \geq \Gamma(!\rho)) \Rightarrow \mL := \mL
\diamond \mL := \star \\
\sigma' :=
\sigma\Big[ \,^{\Upsilon(!\sigma(\dst)) := \mV}_{\Gamma(!\sigma(\dst)) :=\mL}
\Big]\qquad\iota' := \Succ(!\sigma'.\CFG,\iota)\qquad\rho' := \isIPD(\iota',\rho, \sigma')}
{\iota,  \theta, \sigma, \rho ~\leadsto~ \iota',  \theta',
\sigma', \rho'}
\end{align*} \normalsize 
\emph{create-this} creates and allocates an object as \emph{this} used for
construction later in the function. The object is labelled the context
and placed in $\dst$ after deferred NSU check. The prototype chain
pointer is also labelled with the context label.

\scriptsize \begin{align*}
\inference[new-object:]
{\iota = \mbox{``new-object dst:r''}\\
(\theta',v) :=  \mathit{newObject}(\theta, \Gamma(!\rho))\qquad 
\mL := \Gamma(v) \sqcup \Gamma(!\rho)\qquad\mV := \Upsilon(v) \\
(\Gamma(!\sigma(\dst)) \geq \Gamma(!\rho)) \Rightarrow (\mL :=
\mL)  \diamond (\mL := \star) \\
\sigma' :=
\sigma\Big[ \,^{\Upsilon(!\sigma(\dst)) := \mV}_{\Gamma(!\sigma(\dst)) :=\mL}
\Big]\qquad\iota' := \Succ(!\sigma'.\CFG,\iota)\qquad\rho' := \isIPD(\iota',\rho, \sigma')}
{\iota,  \theta, \sigma, \rho ~\leadsto~ \iota',  \theta',
\sigma', \rho'}
\end{align*} \normalsize 
\emph{new-object} constructs a new empty object instance and puts it
in $\dst$ after deferred NSU check. The object is labelled with the
context label and the prototype chain pointer is also labelled with
the context. 

\scriptsize \begin{align*}
\inference[get-by-id:]
{\iota = \mbox{``get-by-id dst:r base:r prop:id vdst:r''}\\
v := \mathit{getPropertyById(!\sigma(base), prop, vdst)}\qquad 
\mL := \Gamma(v) \sqcup \Gamma(!\rho)\qquad \mV := \Upsilon(v) \\
(\Gamma(!\sigma(\dst)) \geq \Gamma(!\rho)) \Rightarrow (\mL :=
\mL)  \diamond (\mL := \star) \\
\sigma' :=
\sigma\Big[ \,^{\Upsilon(!\sigma(\dst)) := \mV}_{\Gamma(!\sigma(\dst)) :=\mL}
\Big]\qquad \iota' := \Succ(!\sigma'.\CFG,\iota)\qquad\rho' := \isIPD(\iota',\rho, \sigma')
}
{\iota,  \theta, \sigma, \rho ~\leadsto~ \iota',  \theta,
\sigma', \rho'}
\end{align*} \normalsize 
\emph{get-by-id} gets the property named by the identifier
\emph{prop} from the object in the \emph{base} register and puts it into the
$\dst$ register after the deferred NSU check. If the object does not 
contain the property, it looks
up the prototype chain to determine if any of the proto objects
contain the property. When traversing the prototype chain, the context
is joined with the structure label of all the objects and the
prototype chain pointer labels until the property is found or
the end of the chain. It then joins the property label to the
context. If the property is not found, it returns
\emph{undefined}. The joined label of the context is the label of the
property put in the $\dst$ register.

If the property is an accessor property, it calls the getter
function, sets the getter flag in the call-frame and updates the
destination register field with the register where the value is to be
inserted. It then transfers the control to the first instruction in
the getter function.

\scriptsize \begin{align*}
\inference[put-by-id:]
{\iota = \mbox{``put-by-id base:r prop:id value:r direct:b''} \\
\Gamma(!\sigma(\val)) \neq \star \\
(\mathit{direct} = \true \Rightarrow \theta' := \mathit{putDirect}(\Gamma(!\rho) ,\sigma,
\theta, \mathit{base, prop, value)}~\diamond \\
\theta' := \mathit{putIndirect}(\Gamma(!\rho), \sigma, \theta, \mathit{base, prop, value}))\\
\iota' := \Succ(!\sigma.\CFG, \iota)\qquad\rho' := \isIPD(\iota',\rho, \sigma)}
{\iota,  \theta, \sigma, \rho ~\leadsto~ \iota',  \theta',
\sigma, \rho'}
\end{align*} \normalsize 
\emph{put-by-id} writes into the heap the property of an object. We check for $\star$
  in the label of \emph{value} register. If it contains a $\star$, the
  program aborts as this could potentially result in an implicit information flow. If
  not, it writes the property into the object. The basic
  functionality is to search for the property in the object and its
  prototype chain, and change it. If the property is not found, a new
  property for the current object with the property label as the
  context is created. Based on whether the property is in the object
  itself (or needs to be created
  in the object itself) or in the prototype chain of the object, it
  calls \emph{putDirect} and \emph{putIndirect}, respectively.

\scriptsize \begin{align*}
\inference[del-by-id:]
{\iota = \mbox{``del-by-id dst:r base:r prop:id''}\\
\Gamma(!\sigma(\base)) \neq \star\qquad 
(\theta',v) := \mathit{delById}(\Gamma(!\rho), \sigma, \theta, \mathit{base,
  prop)}\qquad
\mL := \Gamma(v) \sqcup \Gamma(!\rho)\\
\mV := \Upsilon(v)\qquad
(\Gamma(!\sigma(\dst)) \geq \Gamma(!\rho)) \Rightarrow (\mL :=
\mL)  \diamond (\mL := \star) \\
\sigma' :=
\sigma\Big[ \,^{\Upsilon(!\sigma(\dst)) := \mV}_{\Gamma(!\sigma(\dst)) :=\mL}
\Big]\qquad\iota' := \Succ(!\sigma'.\CFG,\iota)\qquad\rho' := \isIPD(\iota',\rho, \sigma')}
{\iota,  \theta, \sigma, \rho ~\leadsto~ \iota',  \theta',
\sigma', \rho'}
\end{align*} \normalsize 
\emph{del-by-id} deletes the property specified by \emph{prop} in
the object contained in \emph{base}. If the structure label of the
object is less than  the context, the deletion does not happen. If the
property is found, the property is deleted and Boolean value
\emph{true} is written to $\dst$, else it writes \emph{false} to
$\dst$. The label of the Boolean value is the structure label of the
object  joined with the property label. 

\scriptsize \begin{align*}
\inference[getter-setter:]
{\iota = \mbox{``put-getter-setter base:r prop:id getter:r
    setter:r''},  \\
\star \notin \Gamma(!\sigma(getter)),
~~\star \notin \Gamma(!\sigma(setter)),\\
\theta':= \mathit{putGetterSetter(\Gamma(!\rho), \sigma, base, prop, !\sigma(getter), !\sigma(setter))}, \\
\iota' := \Succ(!\sigma.\CFG, \iota), ~\rho' := \isIPD(\iota',\rho)}
{\iota,  \theta, \sigma, \rho ~\leadsto~ \iota',  \theta', \sigma,
\rho'}
\end{align*} \normalsize 
\emph{getter-setter} puts the accessor descriptor to the object in
register \emph{base}. It initially checks if the structure label of
the object is greater or equal to the context. The property for  which
the accessor properties are added is given in the register
\emph{prop}. The property label
of the accessor functions is set to the
context. \emph{putGetterSetter} calls \emph{putIndirect} internally
and sets the getter/setter property of the object with the specified value.

\scriptsize \begin{align*}
\inference[get-pnames:]
{\iota = \mbox{``get-pnames dst:r base:r i:r size:r breakTarget:offset''} \\
\Gamma\mathit{(!\sigma(base))} \neq \star\qquad
(v_1,v_2, v_3, \iota') := \mathit{getPropNames(\sigma,\iota, base, i, size,
breakTarget)}\\
\mL_n := \Gamma(!\sigma(\mathit{base})) \sqcup \Gamma(v_n) \sqcup \Gamma(!\rho)\qquad\mV_n := \Upsilon(v_n),n
:= 1, 2, 3 \\
(\Gamma(!\sigma(\dst)) \geq \Gamma(!\rho)) \Rightarrow (\mL_1 :=
\mL_1)  \diamond (\mL_1 := \star) \\
(\Gamma(!\sigma(i)) \geq \Gamma(!\rho)) \Rightarrow (\mL_2 :=
\mL_2)  \diamond (\mL_2 := \star) \\
(\Gamma(!\sigma(\mathit{size})) \geq \Gamma(!\rho)) \Rightarrow (\mL_3 :=
\mL_3)  \diamond (\mL_3 := \star) \\
\sigma'' := \sigma\Big[ \,^{\Upsilon(!\sigma(\dst)) := \mV_1}_{\Gamma(!\sigma(\dst)) :=\mL_1}\Big]\qquad
\sigma''' := \sigma''\Big[ \,^{\Upsilon(!\sigma''(i)) := \mV_2}_{\Gamma(!\sigma''(i)) :=\mL_2} \Big]\qquad
\sigma' := \sigma'''\Big[ \,^{\Upsilon(!\sigma'''(\mathit{size})) :=
  \mV_3}_{\Gamma(!\sigma'''(\mathit{size})) :=\mL_3}\Big]\\
v_n = \mathit{undefined} \Rightarrow 
(\mL =  \Gamma(!\sigma(\mathit{base})))~~~\diamond \\
(\mL = \Gamma\mathit{(!\sigma(base)) \sqcup
  \Gamma(\theta(!\sigma(base))) \sqcup (\forall p \in
Prop(\theta(!\sigma(base))). \Gamma(p))}) \\
\rho'' = \rho.\push(\mL, \IPD(\iota), \cf(\iota), \false)\qquad\rho' := \isIPD(\iota',\rho'', \sigma')}
{\iota,  \theta, \sigma, \rho ~\leadsto~ \iota',  \theta,
\sigma', \rho'}
\end{align*} \normalsize 
\emph{get-pnames} creates a property name list for object in register
\emph{base} and puts it in $\dst$, initializing $i$ and \emph{size}
for iteration through the list, after the deferred NSU check. If
\emph{base} is \emph{undefined} or \emph{null}, it jumps to
\emph{breakTarget}. It is a branching instruction and pushes the label
with join of all the property labels and the structure label of the
object along with the IPD on the pc-stack. If
  the IPD of the instruction is SEN or the same as the top of the
  pc-stack, then we just join the label on top of the pc-stack with
  the context label determined above.

\scriptsize \begin{align*}
\inference[next-pname:]
{\iota = \mbox{``next-pname dst:r base:r i:n size:n iter:n target:offset''}\\
(v_1,v_2, \iota') := \mathit{getNextPropNames(\sigma, \iota, base, i, size, iter,
target)} \\
\mL_n := \Gamma(v_n) \sqcup \Gamma(!\rho)\qquad\mV_n := \Upsilon(v_n),~n
:= 1, 2 \\ 
(\Gamma(!\sigma(\dst)) \geq \Gamma(!\rho)) \Rightarrow (\mL_1 :=
\mL_1)  \diamond (\mL_1 := \star) \\
(\Gamma(!\sigma(i)) \geq \Gamma(!\rho)) \Rightarrow (\mL_2 :=
\mL_2)  \diamond (\mL_2 := \star) \\
\sigma'' :=\sigma\Big[ \,^{\Upsilon(!\sigma(\dst)) := \mV_1}_{\Gamma(!\sigma(\dst)) :=\mL_1}\Big]\qquad
\sigma' :=\sigma''\Big[ \,^{\Upsilon(!\sigma''(i)) := \mV_2}_{\Gamma(!\sigma''(i)) :=\mL_2}\Big]\qquad\rho' := \isIPD(\iota',\rho, \sigma')}
{\iota,  \theta, \sigma, \rho ~\leadsto~ \iota',  \theta, \sigma',
\rho'}
\end{align*} \normalsize 
\emph{next-pname} copies the next name from the property name list
created by \emph{get-pnames} in \emph{iter} to $\dst$ after deferred
NSU check, and jumps to
\emph{target}. If there are no names left, it continues with the next
instruction. Although, it behaves as a branching instruction, the
context pertaining to this opcode is already pushed in
\emph{get-pnames}. Also, the IPD corresponding to this instruction, is
the same as the one determined by \emph{get-pnames}. Thus, we do not
push on the pc-stack in this instruction. 

\scriptsize \begin{align*}
\inference[resolve:]
{\iota = \mbox{``resolve dst:r prop:id''}\\
v := \mathit{resolveInSc(\Gamma(!\rho), !\sigma.\SC, prop)}\\
\mL := \Gamma(v) \sqcup \Gamma(!\rho)\qquad\mV := \Upsilon(v)\qquad
(\Gamma(!\sigma(\dst)) \geq \Gamma(!\rho)) \Rightarrow (\mL :=
\mL)  \diamond (\mL := \star) \\
\sigma' :=\sigma\Big[ \,^{\Upsilon(!\sigma(\dst)) := \mV}_{\Gamma(!\sigma(\dst)) :=\mL}\Big]\qquad
\iota' := \Succ(!\sigma'.\CFG,\iota)\qquad \rho' := \isIPD(\iota',\rho, \sigma')}
{\iota,  \theta, \sigma, \rho ~\leadsto~ \iota',  \theta,
\sigma', \rho'}
\end{align*} \normalsize 
\emph{resolve} searches for the property in the scope
  chain and writes it into $\dst$ register, if found. The label of
  the property written in $\dst$ is a join of the context label, all the
  nodes (structure label of the object contained in it) traversed in
  the scope chain and the label associated with the pointers in the
  chain until the node (object) where the property is found. If the
  initial label of the value contained in $\dst$ was lower than the
  context label, then the label of the value in $\dst$ is joined with
  $\star$. In case the property is not found, the instruction throws an
  exception (similar to throw, as described later).

\scriptsize \begin{align*}
\inference[resolve-skip:]
{\iota = \mbox{``resolve-skip dst:r prop:id skip:n''}\\
v := \mathit{resolveInScWithSkip(\Gamma(!\rho), !\sigma.\SC, prop, skip)}\\
\mL := \Gamma(v) \sqcup \Gamma(!\rho)\qquad\mV := \Upsilon(v)\qquad
(\Gamma(!\sigma(\dst)) \geq \Gamma(!\rho)) \Rightarrow (\mL :=
\mL)  \diamond (\mL := \star) \\
\sigma' :=
\sigma\Big[ \,^{\Upsilon(!\sigma(\dst)) := \mV}_{\Gamma(!\sigma(\dst)) :=\mL}
\Big]\qquad
\iota' := \Succ(!\sigma'.\CFG,\iota)\qquad\rho' := \isIPD(\iota',\rho, \sigma')}
{\iota,  \theta, \sigma, \rho ~\leadsto~ \iota',  \theta, \sigma',
\rho'}
\end{align*} \normalsize 
\emph{resolve-skip} looks up the property named by \emph{prop} in
the scope chain similar to \emph{resolve}, but it skips the top
\emph{skip} levels and writes the result to register $\dst$. If the
property is not found, it also raises an exception and behaves
similarly to \emph{resolve}.

\scriptsize \begin{align*}
\inference[resolve-global:]
{\iota = \mbox{``resolve-global dst:r prop:id''}\\
v := \mathit{resolveGlobal(\Gamma(!\rho), \sigma, prop)}\\
\mL := \Gamma(v) \sqcup \Gamma(!\rho)\qquad\mV := \Upsilon(v)\qquad
(\Gamma(!\sigma(\dst)) \geq \Gamma(!\rho)) \Rightarrow (\mL :=
\mL)  \diamond (\mL := \star) \\
\sigma' :=\sigma\Big[ \,^{\Upsilon(!\sigma(\dst)) := \mV}_{\Gamma(!\sigma(\dst)) :=\mL}\Big]\qquad
\iota' := \Succ(!\sigma'.\CFG,\iota)\qquad\rho' := \isIPD(\iota',\rho, \sigma')}
{\iota,  \theta, \sigma, \rho ~\leadsto~ \iota',  \theta,
\sigma', \rho'}
\end{align*} \normalsize 
\emph{resolve-global} looks up the property named by \emph{prop} in
the global object. If the structure of the global object matches the
one passed here, it looks into the global object. Else, it falls back
to perform a full \emph{resolve}.

\scriptsize \begin{align*}
\inference[resolve-base:]
{\iota = \mbox{``resolve-base dst:r prop:id isStrict:bool''}\\
v := \mathit{resolveBase(\Gamma(!\rho), \sigma, !\sigma.\SC, prop, isStrict)}\\
\mL := \Gamma(v) \sqcup \Gamma(!\rho)\qquad\mV := \Upsilon(v)\qquad
(\Gamma(!\sigma(\dst)) \geq \Gamma(!\rho)) \Rightarrow (\mL :=
\mL)  \diamond (\mL := \star) \\
\sigma' :=\sigma\Big[ \,^{\Upsilon(!\sigma(\dst)) := \mV}_{\Gamma(!\sigma(\dst)) :=\mL}\Big]\qquad
\iota' := \Succ(!\sigma'.\CFG,\iota)\qquad\rho' := \isIPD(\iota',\rho, \sigma')}
{\iota,  \theta, \sigma, \rho ~\leadsto~ \iota',  \theta,
\sigma', \rho'}
\end{align*} \normalsize 
\emph{resolve-base} looks up the property named by \emph{prop} in
the scope chain similar to \emph{resolve} but writes the object to register $\dst$. If the
property is not found and \emph{isStrict} is \emph{false}, the global
object is stored in $\dst$.

\scriptsize \begin{align*}
\inference[resolve-with-base:]
{\iota = \mbox{``resolve-with-base bDst:r pDst:r prop:id''}\\
\mathit{(bdst, pdst) := resolveBaseAndProperty(\Gamma(!\rho), \sigma, baseDst, propDst,
  prop)}\\
\mL1 := \Gamma(!\sigma(\mathit{bdst)) \sqcup \Gamma(!\rho)\qquad\mV1 :=\Upsilon(!\sigma(bdst))}\\
\mL2 := \Gamma(!\sigma(\mathit{pdst)) \sqcup \Gamma(!\rho)\qquad\mV2 := \Upsilon(!\sigma(pdst))} \\
(\Gamma(!\sigma(\mathit{bDst})) \geq \Gamma(!\rho)) \Rightarrow (\mL1 :=
\mL1)  \diamond (\mL1 := \star) \\
(\Gamma(!\sigma(\mathit{pDst})) \geq \Gamma(!\rho)) \Rightarrow (\mL2 :=
\mL2)  \diamond (\mL2 := \star) \\
\sigma'' :=
\sigma\Big[ \,^{\Upsilon(!\sigma(\mathit{bDst})) := \mV1}_{\Gamma(!\sigma(\mathit{bDst})) :=\mL1}
\Big]\qquad
\sigma' :=
\sigma''\Big[ \,^{\Upsilon(!\sigma''(\mathit{pDst})) := \mV2}_{\Gamma(!\sigma''(\mathit{pDst})) :=\mL2}
\Big]\\
\iota' := \Succ(!\sigma'.\CFG,\iota)\qquad\rho' := \isIPD(\iota',\rho, \sigma')}
{\iota,  \theta, \sigma, \rho ~\leadsto~ \iota',  \theta, \sigma',
\rho'}
\end{align*} \normalsize 
\emph{resolve-with-base} looks up the property named by \emph{prop} in
the scope chain similar to \emph{resolve-base} and writes the object
to register \emph{bDst}. It also, writes the property to \emph{pDst}. If the
property is not found it raises an exception like \emph{resolve}.

\scriptsize \begin{align*}
\inference[get-scoped-var:]
{\iota = \mbox{``get-scoped-var dst:r index:n skip:n''} \\
v := \mathit{getScopedVar(\Gamma(!\rho), \sigma, \theta, index, skip)}\qquad
\mL := \Gamma(v) \sqcup \Gamma(!\rho)\qquad\mV := \Upsilon(v)\\
(\Gamma(!\sigma(\dst)) \geq \Gamma(!\rho)) \Rightarrow (\mL :=
\mL)  \diamond (\mL := \star) \\
\sigma' :=
\sigma\Big[ \,^{\Upsilon(!\sigma(\dst)) := \mV}_{\Gamma(!\sigma(\dst)) :=\mL}
\Big]\qquad
\iota' := \Succ(!\sigma'.\CFG,\iota)\qquad\rho' := \isIPD(\iota',\rho, \sigma')}
{\iota,  \theta, \sigma, \rho ~\leadsto~ \iota',  \theta,
\sigma', \rho'}
\end{align*} \normalsize 
\emph{get-scoped-var} loads the contents of the \emph{index} local
from the scope chain skipping \emph{skip} nodes and places it in
$\dst$, after deferred NSU. The label of the value in $\dst$ includes
the join of the current context along with all the structure label of
objects in the skipped nodes. 

\scriptsize \begin{align*}
\inference[put-scoped-var:]
{\iota = \mbox{``put-scoped-var index:n skip:n value:r''}\\
\Gamma(!\sigma(\val)) \neq \star\qquad
 \sigma' := \mathit{putScopedVar(\Gamma(!\rho), \sigma, \theta, index, skip, value)}\\ 
\iota' := \Succ(!\sigma'.\CFG,\iota)\qquad\rho' := \isIPD(\iota',\rho, \sigma')}
{\iota,  \theta, \sigma, \rho ~\leadsto~ \iota',  \theta,
\sigma', \rho'}
\end{align*} \normalsize 
\emph{put-scoped-var} puts the contents of the \emph{value} in the \emph{index} local
in the scope chain skipping \emph{skip} nodes. The label of the value includes
the join of the current context along with the structure label of all
the objects in the skipped nodes.

\scriptsize \begin{align*}
\inference[push-scope:]
{\iota = \mbox{``push-scope scope:r''}\\
\sigma' := \mathit{pushScope(\Gamma(!\rho),\sigma, scope)}\qquad
\iota' := \Succ(!\sigma'.\CFG,\iota)\qquad\rho' := \isIPD(\iota',\rho, \sigma')}
{\iota,  \theta, \sigma, \rho ~\leadsto~ \iota',  \theta,
\sigma', \rho'}
\end{align*} \normalsize 
\emph{push-scope} converts \emph{scope} to object and pushes it onto
the top of the current scope chain. The contents of the register
\emph{scope} are replaced by the created object. The scope chain
pointer label is set to the context. 

\scriptsize \begin{align*}
\inference[pop-scope:]
{\iota = \mbox{``pop-scope''}\\
\sigma' := \mathit{popScope}(\Gamma(!\rho), \sigma)\qquad
\iota' := \Succ(!\sigma'.\CFG,\iota)\qquad\rho' := \isIPD(\iota',\rho, \sigma')}
{\iota,  \theta, \sigma, \rho ~\leadsto~ \iota',  \theta,
\sigma', \rho'}
\end{align*} \normalsize 
\emph{pop-scope} removes the top item from the current scope chain if
the scope chain pointer label is greater than or equal to the
context. 

\scriptsize \begin{align*}
\inference[jmp-scope:]
{\iota = \mbox{``jmp-scope count:n target:n''} \\
\sigma' := \mathit{jmpScope(\Gamma(!\rho), \sigma, count)}\qquad 
\iota' := \Succ(!\sigma'.\CFG,\iota)\qquad\rho' := \isIPD(\iota',\rho, \sigma')}
{\iota,  \theta, \sigma, \rho ~\leadsto~ \iota',  \theta,
\sigma', \rho'}
\end{align*} \normalsize 
\emph{jmp-scope} removes the top \emph{count} items from the current scope chain if
the scope chain pointer label is greater than or equal to the
context. It then jumps to offset specified by \emph{target}.

\scriptsize \begin{align*}
\inference[throw:]
{\iota = \mbox{``throw ex:r''}\qquad
\eV := \Upsilon(!\sigma(\mathit{ex}))\\
(\sigma', \iota') := \mathit{throwException(\sigma, \iota)} \qquad \rho' := \isIPD(\iota',\rho, \sigma')}
{\iota,  \theta, \sigma, \rho ~\leadsto~ \iota',  \theta,
\sigma', \rho'}
\end{align*} \normalsize 
\emph{throw} throws an exception and points to the
  exception handler as the next instruction to be executed, if
  any. The exception handler might be in the same function or in an
  earlier  function. If it is not present, the program terminates. If
  it has an exception handler, it has an edge to the synthetic exit
  node. Apart from this,
  \emph{throwException} pops the call-frames from the call-stack until
  it reaches the call-frame containing the exception handler. It writes
  the exception value to a local interpreter variable ($\eV$), which
  is then read by catch.

\scriptsize \begin{align*}
\inference[catch:]
{\iota = \mbox{``catch ex:r''}\\
\mL := \Gamma(\mathit{\eV}) \sqcup \Gamma(!\rho)\qquad
(\Gamma(!\sigma(\mathit{ex})) \geq \Gamma(!\rho)) \Rightarrow (\mL :=
\mL)  \diamond (\mL := \star) \\
\sigma' :=
\sigma\Big[ \,^{\Upsilon(!\sigma(\mathit{ex})) := \Upsilon(\eV))}_{\Gamma(!\sigma(\mathit{ex})) :=\mL}
\Big]\qquad
\eV := \Empty\\
\iota' := \Succ(!\sigma'.\CFG,\iota)\qquad\rho' := \isIPD(\iota',\rho, \sigma')}
{\iota,  \theta, \sigma, \rho ~\leadsto~ \iota',  \theta,
\sigma', \rho'}
\end{align*} \normalsize 
\emph{catch} catches the exception thrown by an
instruction whose handler corresponds to the \emph{catch} block. It
reads the exception value from $\eV$ and writes into the register
\emph{ex}. If the label of the register is less than the context, a
$\star$ is joined with the label. It then makes the $\eV$ empty and
proceeds to execute the first instruction in the \emph{catch} block.

\subsection{Proofs and Results}\label{sec:app-ni}

The fields in a frame of the pc-stack are denoted by the following
symbols: $!\rho.\ipd$ represents the IPD field in the top frame of the
pc-stack, $\Gamma(!\rho)$ returns the label field in the top frame of
the pc-stack, and $!\rho.\mC$ returns the call-frame field in the top frame
of the pc-stack.

In the definitions and proofs that follow, we assume that the level of
attacker is $L$ in the three-element lattice presented earlier, i.e.,
in the equivalence relation $\eq_\ell$, $\ell = L$. The level of the
attacker is omitted for clarity purposes from definitions and proofs.

\begin{mydef}[Partial bijection]
\label{def:pb}
A partial bijection $\beta$ is a binary relation on heap locations
satisfying the following properties: (1) if $(a,b) \in \beta$ and
$(a,b') \in \beta$, then $b = b'$, and (2) if $(a,b) \in \beta$ and
$(a',b) \in \beta$, then $a = a'$.
\end{mydef}

Using partial bijections, we define equivalence of values, labeled
values and objects. 

\begin{mydef}[Value equivalence]
\label{def:val}
Two values $r_1$ and $r_2$ are equivalent up to $\beta$, written $r_1
\sim^\beta r_2$ if either (1) $r_1 = a$, $r_2 = b$ and $(a,b) \in
\beta$, or (2) $r_1 = r_1 = v$ where $v$ is some primitive value.
\end{mydef}

\begin{mydef}[Labeled value equivalence]
\label{def:lval}
Two labeled values $v_1 = r_1^{\ell_1}$ and $v_2 = r_2^{\ell_2}$ are
equivalent up to $\beta$, written $v_1 \sim^\beta v_2$ if one of the
following holds: (1) $\ell_1 = \star$ or $\ell_2 = \star$, or (2)
$\ell_1 = \ell_2 = H$, or (3) $\ell_1 = \ell_2 = L$ and $r_1
\sim^\beta r_2$.
\end{mydef}
The first clause of the above definition is standard for the
permissive-upgrade check. It equates a partially leaked value
to every other labeled value.

Objects are formally denoted as $N = (\{p_i \mapsto \{v_i, \mathit{flags}_i\}\}_{i=0}^n,$ $
\_\_proto\_\_  \mapsto a^{\ell_p}, \ell_s)$ . Here $p_i$s correspond to
the property name, $v_i$s are their respective values
and $\mathit{flags}_i$ represent the \emph{writable}, \emph{enumerable} and
\emph{configurable} flags as described in the
\emph{PropertyDescriptor} structure in the cpp model
above. As the current model does not allow modification of the $\mathit{flags}$,
they are always set to $\true$. Thus, we do not need to account for
the $\mathit{flags}_i$ in the equivalence definition below.
$\_\_proto\_\_$ represents a labelled pointer to the object's prototype.

\begin{mydef}[Object equivalence]
\label{def:obj}
For ordinary objects $N = (\{p_i \mapsto \{v_i,$ $\mathit{flags}_i\}\}_{i=0}^n,$ $ \_\_proto\_\_
\mapsto a^{\ell_p}, \ell_s)$ and $N' = (\{p_i' \mapsto \{v_i', \mathit{flags}_i'\}\}_{i=0}^m,
\_\_proto\_\_$ $\mapsto a'^{\ell_p'}, \ell_s')$, we say $N \sim^\beta
N'$ iff either $\ell_s = \ell_s' = H$ or the following hold: (1)
$\ell_s = \ell_s' = L$, (2) $[p_0,\ldots,p_n] = [p_0',\ldots,p_m']$
(in particular, $n = m$), (3) $\forall i.\, v_i \sim^\beta v_i'$, and
(4) $a^{\ell_p} \sim^\beta a^{\ell_p'}$.

For function objects $F = (N, f, \Sigma)$ and $F' = (N', f',
\Sigma')$, we say $F \sim^\beta F'$ iff either $N.\ell_s = N.\ell_s' =
H$ or $N \sim^\beta N'$, $f =^\beta f'$ and $\Sigma \sim^\beta
\Sigma'$. 
\end{mydef}
The equality $f =^\beta f'$ of nodes $f,f'$ in CFGs means that the
portions of the CFGs reachable from $f,f'$ are equal modulo renaming
of operands to bytecodes under $\beta$. Equivalence $\Sigma \sim^\beta
\Sigma'$ of scope chains is defined below. Because we do not allow
$\star$ to flow into heaps, we do not need corresponding clauses in
the definition of object equivalence.

\begin{mydef}[Heap equivalence]
\label{def:heap}
For two heaps $\theta_1, \theta_2$, we say that $\theta_1 \sim^\beta \theta_2$ iff $\forall (a,b) \in
\beta. \; \theta_1(a) \sim^\beta \theta_2(b)$.
\end{mydef}

Unlike objects, we allow $\star$ to permeate scope chains, so our
definition of scope chain equivalence must account for it. Scope
chains are denoted as $\Sigma$. A scope-chain node contains a label $\ell$
along with an object $S$ (either JSActivation or JSObject) represented as
$(S,\ell)$.
 
\begin{mydef}[Scope chain equivalence]
\label{def:sc}
For two scope chain nodes $S,S'$, we say that $S \sim^\beta S'$ if one
of the following holds: (1) $S = O$, $S' = O'$ and $O \sim^\beta O'$,
or (2) $S = v_0 : \ldots : v_n$, $S' = v_0': \ldots : v_n'$ and
$\forall i.\; v_i \sim^\beta v_i'$.

Equivalence of two scope chains $\Sigma, \Sigma'$ is defined
by the following rules. (1) $nil \sim^\beta nil$ (2) $(nil \sim^\beta
(S, \ell))$ if $\ell = H$ or $\ell = \star$ (3) $((S, \ell) \sim^\beta
nil)$ if $\ell = H$ or $\ell = \star$ and (4)
$((S, \ell) : \Sigma) \sim^\beta ((S',\ell'): \Sigma')$ if one of the
following holds: (a) $\ell = \star$ or $\ell' = \star$, (b) $\ell =
\ell' = H$, or (c) $\ell = \ell' = L$, $S \sim^\beta S'$ and $\Sigma
\sim^\beta \Sigma'$.
\end{mydef}

\begin{mydef}[Call-frame equivalence]
\label{def:cf}
For two call frames $\mu_1,\mu_2$, we say $\mu_1 \sim^\beta \mu_2$ iff
(1) $\#Registers(\mu_1) = \#Registers(\mu_2)$, (2) $\forall i.\;
\mu_1.Registers[i] \sim^\beta \mu_2.Registers[i]$, (3) $\mu_1.\CFG
=^\beta \mu_2.\CFG$, (4) $\mu_1.Scopechain \sim^\beta
\mu_2.Scopechain$, (5) $\mu_1.\iota_r = \mu_2.\iota_r$ (6)
$(\mu_1.\ell_c = \mu_2.\ell_c = H) \vee (\mu_1.\ell_c = \mu_2.\ell_c =
L \wedge \mu_1.f_{callee} \eq \mu_2.f_{callee})$ (7) $\mu_1.argcount = \mu_2.argcount
$ (8) $\mu_1.getter = \mu_2.getter$ and (9) $\mu_1.dReg =_{\beta} \mu_1.dReg$
\end{mydef}
Note that a register is simply a labeled value in our semantics, so
clause (2) above is well-defined.

\begin{mydef}[$pc$-stack equivalence]
\label{def:pc}
For two pc-stacks $\rho_1, \rho_2$, we say $\rho_1 \sim \rho_2$ iff
the corresponding nodes of $\rho_1$ and $\rho_2$ 
having label $L$ are equal, except for the call-frame $(\mC)$ field. 
\end{mydef}

In proofs that follow, two pc-stack nodes are equal if their
respective fields are equal, except for the call-frame $(\mC)$ field.

\begin{mydef}[Call-stack equivalence]
\label{def:cs}
Given $\rho_1 \sim \rho_2$, suppose:
\begin{enumerate}
\item $e_1$ is the lowest $H$-labelled node in $\rho_1$ 
\item $e_2$ is the lowest $H$-labelled node in $\rho_2$
\item $\mu_1$ is the node of $\sigma_1$ pointed to by $e_1$ 
\item $\mu_2$ is the node of $\sigma_2$ pointed to by $e_2$
\item $\sigma_1'$ is prefix of $\sigma_1$ up to and including $\mu_1$
  or \\ if $\Gamma(!\rho_1) = L$ or $\rho_1$ is empty, $\sigma_1' = \sigma_1$
\item $\sigma_2'$ is prefix of $\sigma_2$ up to and including $\mu_2$
  or \\ if $\Gamma(!\rho_2) = L$ or $\rho_2$ is empty, $\sigma_2' = \sigma_2$
\end{enumerate}
then $\sigma_1 \sim^{\beta}_{\rho_1, \rho_2} \sigma_2$, iff 
(1) $|\sigma_1'| = |\sigma_2'|$, and (2) $\forall i \leq
|\sigma_1'|.(\sigma_1'[i] \eq \sigma_2'[i])$.
\end{mydef}

\begin{mydef}[State equivalence]
\label{def:ceq}
Two states $s_1 = \langle
\iota_1, \theta_1, \sigma_1,\rho_1\rangle$ and $s_2 = \langle
\iota_2, \theta_2, \sigma_2,\rho_2\rangle$ are equivalent, written as
$s_1 \eq s_2$, iff $\iota_1 = \iota_2$, $\rho_1  \sim \rho_2$,
$\theta_1 \eq \theta_2$, and $\sigma_1 \eq_{\rho_1, \rho_2} \sigma_2$.
\end{mydef}

\begin{myLemma}[Confinement Lemma]
\label{lem:confinement}
If $~\langle\iota,  \theta, \sigma,\rho \rangle~\leadsto~\langle\iota', 
\theta', \sigma',\rho' \rangle$ and $\Gamma(!\rho) = H$, then
$\rho \sim \rho'$, $\sigma \eq_{\rho,\rho'} \sigma'$
and $\theta \eq \theta'$ where $\beta = \{(a,a)\,|\,a \in \theta \}$
\end{myLemma}
\begin{proof}
As $\Gamma(!\rho) = H$, the $L$ labelled nodes in the pc-stack will remain
unchanged. Branching instructions pushing a new node would have label
$H$ due to monotonicity of pc-stack. Even if $\iota'$ is the IPD corresponding to the
$!\rho.\ipd$, it would only pop the $H$ labelled node. Thus, the $L$ labelled
nodes will remain unchanged. Hence, $\rho \sim \rho'$.

We assume that the $!\rho$ is the first node labelled $H$ in the
context stack. 
For, other higher labelled nodes above the first node labelled
$H$ in the pc-stack, the call-frames corresponding to the nodes having $L$ label in
the pc-stack remain the same. Hence, $\sigma \eq_{\rho,\rho'}
\sigma'$. \\
By case analysis on the instruction type:
\begin{enumerate}
\item \emph{prim}: 
\begin{enumerate}
\item If $\Gamma(!\sigma(\dst)) \geq \Gamma(!\rho)$, then 
  $\Gamma(!\sigma(\dst)) = H$. \\
  By premise of prim,  $\Gamma(!\sigma'(\dst)) = H$. 
  By Definition~\ref{def:lval}, $!\sigma(\dst) \eq
  \,!\sigma'(\dst)$.
\item If $\Gamma(!\sigma(\dst)) < \Gamma(!\rho)$, then
  $\Gamma(!\sigma'(\dst))$ will contain a $\star$ and by
  Definition~\ref{def:lval}, $!\sigma(\dst) \eq \,!\sigma'(\dst)$. 
\end {enumerate}
 Only $\dst$ changes in the call-frame, so by
 Definition~\ref{def:cf}, $!\sigma \eq \,!\sigma'$. Also,
 other call-frames remain unchanged. By Definition~\ref{def:cs}, $\sigma \eq_{\rho,\rho'} \sigma'$. \\ 
 $\theta = \theta'$, thus, $\theta \eq \theta'$.
\item \emph{mov}: Similar to prim.
\item \emph{jfalse}: 
$\sigma = \sigma'$ and $\theta = \theta'$, so, $\sigma \eq_{\rho,\rho'} \sigma'$ and $\theta \eq \theta'$.
\item \emph{loop-if-less}:  Similar to jfalse. 
\item \emph{typeof}: Similar to prim.
\item \emph{instanceof}: Similar to prim.
\item \emph{enter}: 
 $\sigma = \sigma'$,  so $\sigma \eq_{\rho,\rho'} \sigma'$. $\theta =
  \theta'$, so $\theta \eq \theta'$.
\item \emph{ret}: 
  If $!\sigma.getter = false$ then only $!\sigma$ is popped, the
  call-frames until $(!\rho.\mC)$ are unchanged. When $!\sigma.getter
  = true$, then it sets $!\sigma'(!\sigma.dReg)$ with $!\sigma(res)$. Now,
  let $\sigma_1$ is the prefix of $\sigma$ such that $!\rho.\mC
  =_{\beta} !\sigma_1$. If $!\sigma' \notin \sigma_1$ then changes in
  $!\sigma'$ does not effect the callframe equivalence and if
  $!\sigma' \in \sigma_1$ then
  $\Gamma(!\sigma'(!\sigma.dReg))=\star$ (when
  $\Gamma(!\sigma_1(!\sigma.dReg)))=L$ or
  $\Gamma(!\sigma_1(!\sigma.dReg)))=\star$) and
  $\Gamma(!\sigma'(!\sigma.dReg)))=H$ (when
  $\Gamma(!\sigma_1(!\sigma.dReg)))=H$), each of the cases give
  $!\sigma_1(!\sigma.dReg) \eq !\sigma'(!\sigma.dReg)$ from
  Definition~\ref{def:lval}. So, $\sigma \eq_{\rho,\rho'} \sigma'$, by
  Definition~\ref{def:cs}.  $\theta = \theta'$, so $\theta \eq
  \theta'$.
\item \emph{end}: The confinement lemma does not apply. 
\item \emph{call}: If it pushes on top of pc-stack, $!\rho.\mC$ is the
  lowest $H$-labelled node in $\rho'$. If it joins the
  label with $!\rho$, the $L$ labelled nodes remain unchanged and the
  $!\rho.\mC =\,!\rho'.\mC$. All the call-frames until $!\rho.\mC$
  remain unchanged. So, by Definition~\ref{def:cs},  
  $\sigma \eq_{\rho,\rho'} \sigma'$. $\theta = \theta'$, so $\theta \eq \theta'$.
\item \emph{call-put-result}: Similar to prim. 
\item \emph{call-eval}: If it is a user-defined eval, it is similar to
  call. \\ In strict mode, it pushes a node on scope-chain with label
  $H$ if $\Gamma(!\sigma'.\Sigma) = H$, else labels it $\star$.
  In non-strict mode, it does not push a node on the scope-chain. 
  $!\sigma$ remains equivalent
  with corresponding call-frame in $\sigma'$ by Definition~\ref{def:cf}. 
  As other $L$ call-frames are unchanged, by Definition~\ref{def:cs}, $\sigma \eq_{\rho,\rho'}
  \sigma'$. \\ 
  $\theta = \theta'$, so $\theta \eq \theta'$.
\item \emph{create-arguments}: 
  Over the initial $\beta$, 
  by Definition~\ref{def:heap}, $\theta \eq \theta'$. If the
  argument object is created at $x$, then $\beta = (x,x) \cup \beta$
  after the step is taken. \\
 $\sigma \eq_{\rho,\rho'} \sigma'$ (Similar to prim).  
\item \emph{new-func}: 
  Over the initial $\beta$, 
  by Definition~\ref{def:heap}, $\theta \eq \theta'$. If the
  function object is created at $x$, then $\beta = (x,x) \cup \beta$
  after the step is taken. \\
 $\sigma \eq_{\rho,\rho'} \sigma'$ (Similar to prim).  
\item \emph{create-activation}: 
  Over the initial $\beta$, 
  by Definition~\ref{def:heap}, $\theta \eq \theta'$. If the
  argument objects is created at $x$, then $\beta = (x,x) \cup \beta$
  after the step is taken. \\
  It puts the object in $\dst$ with label $H$ or $\star$,
  depending on $\dst$ value's initial label. 
  Also, pushes a node containing the object in the scope chain with a
  $\star$, if $\Gamma(!\sigma.\Sigma) = L \vee \star$ or with label
  $H$, if $\Gamma(!\sigma.\Sigma) = H$ or $(!\sigma.\Sigma) = nil$.
  Thus, $!\sigma.\Sigma \eq !\sigma'.\Sigma $ by
  Definition~\ref{def:sc}.
  By Definition~\ref{def:cf}, $!\sigma \eq !\sigma'$. Other
  call-frames are unchanged, so $\sigma \eq_{\rho,\rho'} \sigma'$ by Definition~\ref{def:cs}. 
\item \emph{construct}: Similar to call.
\item \emph{create-this}: Similar to create-arguments.
\item \emph{new-object}: 
Over the initial $\beta$, 
  by Definition~\ref{def:heap}, $\theta \eq \theta'$. If the
  new object is created at $x$, then $\beta = (x,x) \cup \beta$
  after the step is taken. \\
 $\sigma \eq_{\rho,\rho'} \sigma'$ (Similar to prim).  
\item \emph{get-by-id}: Similar to mov when the property is a data
  property. If the property is an accessor property then getter is
  invoked and if the invocation of getter pushes an entry on top of pc-stack,
  $!\rho.\mC$ remains the lowest $H$-labelled node in $\rho'$. If it joins
  the label with $!\rho$, the $L$ labelled nodes remain unchanged and
  the $!\rho.\mC =\,!\rho'.\mC$. All the call-frames until $!\rho.\mC$
  remain unchanged. So, by Definition~\ref{def:cs}, $\sigma
  \eq_{\rho,\rho'} \sigma'$. $\theta = \theta'$, so $\theta \eq
  \theta'$.
\item \emph{put-by-id}: Sets the
  property of the object \emph{base} object to the \emph{value} with label $H$ 
  if the structure label of the object $\ell_s = H$. Thus, the object remains
  low-equivalent by Definition~\ref{def:obj}.
  Thus, $\theta \eq \theta'$ by Definition~\ref{def:heap}. \\
  Also, $\sigma = \sigma'$, so, $\sigma \eq_{\rho,\rho'} \sigma'$.
\item \emph{del-by-id }: Deletes the property if
  structure label of object, $\ell_s = H$. Thus, the object remains
  low-equivalent by Definition~\ref{def:obj}.
  By Definition~\ref{def:heap}, $\theta \eq \theta'$. \\
  $\sigma \eq_{\rho,\rho'} \sigma'$ (Similar to mov).
\item \emph{getter-setter}: Sets accessor property of the object
  \emph{base} object with $\Gamma(\mathit{getter})$ and
  $\Gamma(\mathit{setter})$ and label $H$ if the structure label of
  the object $\ell_s = H$. Thus, the object remains low-equivalent by
  Definition~\ref{def:obj}.  Thus, $\theta \eq \theta'$ by
  Definition~\ref{def:heap}.  Also, $\sigma = \sigma'$, so, $\sigma
  \eq_{\rho,\rho'} \sigma'$.
\item \emph{get-pnames}: Similar to mov and jfalse. 
\item \emph{next-pname}: Similar to mov. 
\item \emph{resolve}: If the property exists, it is similar to
  mov. If it does not, it is similar to throw.
\item \emph{resolve-skip}: Similar to resolve. 
\item \emph{resolve-global}: Similar to resolve. 
\item \emph{resolve-base}: Similar to resolve.
\item \emph{resolve-with-base}: Similar to resolve.
\item \emph{get-scoped-var}: Similar to mov.
\item \emph{put-scoped-var}: 
  Writes the value in the \emph{index}th register in \emph{skip}th node. 
  If $\Gamma(!\sigma(\mathit{index})) = H$, then,
  $\Gamma(!\sigma'\mathit{(index)}) = H$. 
  Else if $\Gamma(!\sigma(\mathit{index})) = L$, then,
  $\Gamma(!\sigma'\mathit{(index)}) = \star$. Other call-frames are unchanged.
  Thus, $\sigma \eq_{\rho,\rho'} \sigma'$ by Definition ~\ref{def:cf} and ~\ref{def:cs}. \\
  $\theta = \theta'$, so $\theta \eq \theta'$.
\item \emph{push-scope}: Pushes node on scope-chain with label $H$ if 
  $\Gamma(!\sigma.\Sigma) = H$ or $(!\sigma.\Sigma) = nil$. 
  Else, assigns a $\star$ as the label. Thus, $!\sigma.\Sigma
  \eq !\sigma'.\Sigma$. Registers remain unchanged. 
  By Definition~\ref{def:cf}, $!\sigma \eq !\sigma'$. Other
  call-frames are unchanged, so by Definition~\ref{def:cs}, 
  $\sigma \eq_{\rho,\rho'} \sigma'$. $\theta = \theta'$, so $\theta \eq \theta'$.
\item \emph{pop-scope}: Pops the node from the scope-chain if 
  $\Gamma(!\sigma.\Sigma) = H \vee \star$. Registers remain unchanged.
  By Definition~\ref{def:cf}, $!\sigma \eq !\sigma'$. Other
  call-frames are unchanged, so by Definition~\ref{def:cs}, $\sigma \eq_{\rho,\rho'} \sigma'$. 
 $\theta = \theta'$, so $\theta \eq \theta'$.
\item \emph{jmp-scope}: Similar to pop-scope.
\item \emph{throw}: 
  Pops the call-frames until the handler is reached, i.e., until
  $(!\rho.\mC)$. The property of IPD ensures that
  $!\sigma' =\, (!\rho.\mC)$. Either $!\rho'.\ipd =\, !\rho.\ipd$ or $\iota'
  = \,!\rho.\ipd$. Thus, $!\rho.\mC$ is $!\sigma'$.  
  This call-frame and the ones below
  remain unchanged. Thus, $\sigma \eq_{\rho,\rho'} \sigma'$ by Definition~\ref{def:cs}.\\
  $\theta = \theta'$, so $\theta \eq \theta'$.
\item \emph{catch}: Similar to mov.
\end{enumerate}
\end{proof}




\begin{myCor}
\label{cor:Evolution_callstacks}
If $~\langle\iota_0, \theta_0, \sigma_0,\rho_0 \rangle~\leadsto^{n}~\langle\iota_n, 
\theta_n, \sigma_n, \rho_n \rangle$ and $\forall (0 \leq i \leq n).\Gamma(!\rho_i) = H$, then 
$\rho_0 \sim \rho_n$, and $\sigma_0 \eq_{\rho_0,\rho_n} \sigma_n$
\end{myCor}
\begin{proof}
To prove: $\rho_0 \sim \rho_n$. \\
Proof by induction on n. \\
Basis: $\rho_0 \sim \rho_0$\\
IH : $\rho_0 \sim \rho_{n-1}$\\
From Definition~\ref{def:pc}, $L$ labelled nodes of $\rho_0$ and
$\rho_{n-1}$ are equal. From Lemma~\ref{lem:confinement}, $\rho_{n-1}
\sim \rho_n$ so, $L$ labelled nodes of $\rho_{n-1}$ and $\rho_n$ are
equal. Thus, $L$ labelled nodes of $\rho_0$ and $\rho_n$ are equal and
by Definition~\ref{def:pc}, $\rho_0 \sim \rho_n$.\\ \\
To prove: $\sigma_0 \eq_{\rho_0,\rho_n} \sigma_n$.\\
Basis: $\sigma_0 \eq_{\rho_0,\rho_0} \sigma_0$. \\
IH: $\sigma_0 \eq_{\rho_0,\rho_{n-1}} \sigma_{n-1}$.\\
From Lemma~\ref{lem:confinement}, $\sigma_{n-1} \eq_{\rho_{n-1},\rho_n}
\sigma_{n}$. As $\forall (0 \leq i \leq n).\Gamma(\rho_i) = H$, the
lowest $H$-labelled node is the same (pc-stack grows monotonically) in
$\rho_0, \rho_{n-1}, \rho_n$. Let the call-frames pointed to by lowest
$H$-labelled node be $\mC_0, \mC_{n-1}, \mC_n$ with call-stack size
until the call-frames $ k$ (from Definition~\ref{def:cs} size of
the prefix is same and by transitivity of equality it is the same
for all the three cases). \\
$\forall \mu_0 \in \sigma_0,
\mu_{n-1} \in \sigma_{n-1}, \mu_n \in \sigma_n $ until $\mC_0,
\mC_{n-1}, \mC_n$ respectively
with sizes $k$, 
 the following conditions hold:
\begin{enumerate}
 \item $\forall(1 \leq i \leq k).((\mu_0[i].\#Registers) =
   (\mu_{n-1}[i].\#Registers))$  and $\forall (1 \leq i \leq k).((\mu_{n-1}[i].\#Registers) =
   (\mu_n[i].\#Registers))$. \\
   Thus, $\forall (1 \leq i \leq k).((\mu_0[i].\#Registers) =
   (\mu_n[i].\#Registers))$.

\item As the number of registers is the same, given by $r$, \\
   $\forall(1 \leq i \leq k).\forall r((\mu_0[i].Registers[r]) \eq
   (\mu_{n-1}[i].Registers[r]))$  and \\ $\forall (1 \leq i \leq
   k).\forall r.((\mu_{n-1}[i].Registers[r]) \eq
   (\mu_n[i].Registers[r]))$. \\
 Let $v^{\ell_0}_0$, $v^{\ell_{n-1}}_{n-1}$ and $v^{\ell_n}_n$ represents the values in the
 registers for $\sigma_0$, $\sigma_{n-1}$ and $\sigma_n$
 respectively. Then from Definition~\ref{def:lval} 
 \begin{enumerate}
 \item $\ell_0 = \ell_{n-1} = H$: In this case $\ell_n = H$ and
   $v^{\ell_0}_0 \eq v^{\ell_n}_n$(from
   Lemma~\ref{lem:confinement} and Definition~\ref{def:lval}).
 \item $\ell_0 = \ell_{n-1} = L$ and $v_0 = v_{n-1}$: In this case either:
   \begin{enumerate}
   \item $\ell_n = \star$
   \item $\ell_n = L$ and $v_{n-1} = v_n$: In this case, the value remains unchanged.
   \end{enumerate}
  Thus, from Definition~\ref{def:lval} $v^{\ell_0}_0 \eq v^{\ell_n}_n$.
 \item $\ell_0 = \star \vee \ell_{n-1} = \star$: Now the following cases
   arise:
   \begin{enumerate}
   \item $\ell_0 = \star$: $v^{\ell_0}_0 \eq v^{\ell_n}_n$.
   \item $\ell_{n-1} = \star$: By Lemma~\ref{lem:confinement} $l_n =
     \star$. Thus, $v^{\ell_0}_0 \eq v^{\ell_n}_n$. 
   \end{enumerate}
 \end{enumerate}

\item $\forall(1 \leq i \leq k).((\mu_0[i].\CFG) =
   (\mu_{n-1}[i].\CFG))$  and $\forall (1 \leq i \leq k).((\mu_{n-1}[i].\CFG) =
   (\mu_n[i].\CFG))$. \\
   Thus, $\forall (1 \leq i \leq k).((\mu_0[i].\CFG) =
   (\mu_n[i].\CFG))$.

\item $\forall(1 \leq i \leq k).((\mu_0[i].\Sigma) \eq 
   (\mu_{n-1}[i].\Sigma))$  and $\forall (1 \leq i \leq
   k).((\mu_{n-1}[i].\Sigma) \eq 
   (\mu_n[i].\Sigma))$. \\
  From Definition~\ref{def:sc}:
  \begin{enumerate}
  \item If $nil_0$ and $nil_{n-1}$ be the two scope chains, then due to confinement
      (Lemma~\ref{lem:confinement}) $\mu_n[i].\Sigma = nil$
       or $\mu_n[i].\Sigma = (S,\ell_n)$, where $\ell_n=H$. In either case $\forall (1 \leq i \leq
   k).((\mu_{0}[i].\Sigma) \eq  (\mu_n[i].\Sigma))$ from  Definition~\ref{def:sc}.
  \item If $((S_0,\ell_0):\Sigma_0)$, $((S _{n-1},\ell _{n-1}):\Sigma
    _{n-1})$ and $((S_n,\ell_n):\Sigma_n)$ be the three 
    scope-chains, then for$((S_0,\ell_0):\Sigma_0)$ and $((S _{n-1},\ell _{n-1}):\Sigma
    _{n-1})$ one of the following holds:
    \begin{enumerate}
    \item $\ell_0 = \star \vee \ell_{n-1} = \star$: Due to confinement
      (Lemma~\ref{lem:confinement}) and Definition~\ref{def:sc}
      $\ell_n=\star$. 
    \item $\ell_0 = \ell_{n-1} = H$:  Due to confinement
      (Lemma~\ref{lem:confinement}) and Definition~\ref{def:sc}
      $l_n = H$ .
    \item $\ell_0 = \ell_{n-1} = L \wedge S_0 \eq S_{n-1} \wedge \Sigma_0 \eq \Sigma_{n-1}$:  Due to confinement
      (Lemma~\ref{lem:confinement}) either one should hold:
      \begin{enumerate}
      \item $\ell_n = \star$: By Definition~\ref{def:sc}.
      \item $\ell_n = L \wedge S_{n-1} \eq S_{n} \wedge \Sigma_{n-1}
        \eq \Sigma_{n}$: No additions to the scope chain.
      \end{enumerate}
    \end{enumerate}
    Thus, $\forall (1 \leq i \leq
   k).((\mu_{0}[i].\Sigma) \eq  (\mu_n[i].\Sigma))$ from  Definition~\ref{def:sc}.
  \end{enumerate}

\item $\forall(1 \leq i \leq k).((\mu_0[i].\iota_r) =
   (\mu_{n-1}[i].\iota_r))$  and $\forall (1 \leq i \leq k).((\mu_{n-1}[i].\iota_r) =
   (\mu_n[i].\iota_r))$. \\
   Thus, $\forall (1 \leq i \leq k).((\mu_0[i].\iota_r) =
   (\mu_n[i].\iota_r))$.

\item $\forall(1 \leq i \leq k).(((\mu_0[i].\ell_c) =
   (\mu_{n-1}[i].\ell_c) = H) \vee (((\mu_0[i].\ell_c) =
   (\mu_{n-1}[i].\ell_c) = L) \wedge ((\mu_0[i].f_c) =
   (\mu_{n-1}[i].f_c))))$  and \\
  $\forall (1 \leq i \leq k).(((\mu_{n-1}[i].\ell_c) =
   (\mu_n[i].\ell_c) = H) \vee (((\mu_{n-1}[i].\ell_c) =
   (\mu_n[i].\ell_c) = L) \wedge ((\mu_{n-1}[i].f_c) =
   (\mu_{n}[i].f_c))))$. \\Then either:
  \begin{itemize}
    \item
   $\forall (1 \leq i \leq k).((\mu_0[i].\ell_c) =
   (\mu_n[i].\ell_c) = H)$ or
    \item $\forall (1 \leq i \leq k).(((\mu_{0}[i].\ell_c) =
   (\mu_n[i].\ell_c) = L) \wedge ((\mu_{0}[i].f_c) =
   (\mu_{n}[i].f_c)))$.
  \end{itemize}

 \item $\forall(1 \leq i \leq k).((\mu_0[i].argcount) =
   (\mu_{n-1}[i].argcount))$  and $\forall (1 \leq i \leq k).$ $((\mu_{n-1}[i].argcount) =
   (\mu_n[i].argcount))$. \\
   Thus, $\forall (1 \leq i \leq k).((\mu_0[i].argcount) =
   (\mu_n[i].argcount))$.

 \item $\forall(1 \leq i \leq k).((\mu_0[i].getter) =
   (\mu_{n-1}[i].getter))$  and $\forall (1 \leq i \leq k).((\mu_{n-1}[i].getter) =
   (\mu_n[i].getter))$. \\
   Thus, $\forall (1 \leq i \leq k).((\mu_0[i].getter) =
   (\mu_n[i].getter))$.

 \item $\forall(1 \leq i \leq k).((\mu_0[i].dReg) =_{\beta}
   (\mu_{n-1}[i].dReg))$  and $\forall (1 \leq i \leq k).((\mu_{n-1}[i].dReg) =_{\beta}
   (\mu_n[i].dReg))$. \\
   Thus, $\forall (1 \leq i \leq k).((\mu_0[i].dReg) =_{\beta}
   (\mu_n[i].dReg))$.

\end{enumerate}
From Definition~\ref{def:cf} and Definition~\ref{def:cs}, $\sigma_0 \eq_{\rho_0,\rho_n} \sigma_n$.
\end{proof}

\begin{myCor}
\label{cor:Evolution_heap}
If $\langle\iota_0, \theta_0, \sigma_0,\rho_0
\rangle~\leadsto^{\star}~\langle\iota_n, \theta_n, \sigma_n,
\rho_n \rangle$ and 
$\forall (0 \leq i \leq n). \Gamma(!\rho_i) = H$, 
then $\theta_0 \eq \theta_n$
\end{myCor}
\begin{proof}
By induction on n. \\
Basis: $\theta_0 \eq \theta_0$ by Definition~\ref{def:heap}.\\
IH: $\theta_0 \eq \theta_{n-1}$.\\
From IH and Definition~\ref{def:heap}, $\forall(a,b) \in
\beta.(\theta_0(a) \eq \theta_{n-1}(b)$.
From Lemma~\ref{lem:confinement}, $\theta_{n-1} \eq \theta_n$. Thus,
$\forall(b,c) \in \beta.(\theta_{n-1}(b) \eq \theta_n(c)$\\
As $(a,b) \in \beta$ and $(b,c) \in \beta$, we have $(a,c) \in \beta$
because $\beta$ is an identity bijection. Thus, if $\forall(a,c) \in \beta.(\theta_{0}(a) \eq
\theta_n(c)$, then $\theta_0 \eq \theta_n$.
If $\theta_0(a)$ and $\theta_{n-1}(b)$ contain an ordinary object,
then for their respective structure labels $\ell_s$ and $\ell_s'$, either:
\begin{itemize}
\item $\ell_s = \ell_s' = H$: If $\ell_s' = H$,
then $\ell_s'' = H$ by Definition~\ref{def:obj}, 
where $\ell_s''$ is the structure label of the object in
$\theta_n(c)$. Thus, $\theta_0(a) \eq \theta_n(c)$.
\item $\ell_s = \ell_s' = L$: $[p_0,\ldots,p_n] = [p_0',\ldots,p_m']$
($n = m$), $\forall i.\, v_i \sim^\beta v_i'$, and
$a^{\ell_p} \sim^\beta a^{\ell_p'}$ for respective properties in
$\theta_0(a)$ and $\theta_{n-1}(b)$. \\
If $\ell_s' = L$, then $\ell_s'' = L$ and $[p_0',\ldots,p_m'] = [p_0'',\ldots,p_k'']$
($m = k$), $\forall i.\, v_i' \sim^\beta v_i''$, and
$a^{\ell_p'} \sim^\beta a^{\ell_p''}$ for respective properties in
$\theta_{n-1}(b)$ and $\theta_{n}(c)$. \\
$\ell_s = \ell_s'' = L$, $[p_0,\ldots,p_n] = [p_0'',\ldots,p_k'']$ ($n
= k$). If $\forall i.\, v_i \sim^\beta v_i'$ and $\forall i.\, v_i'
\sim^\beta v_i''$, then either $\ell_i = \ell_i' = \ell_i'' = H$ or 
$\ell_i = \ell_i' = \ell_i'' = L$ and $r_i = r_i' = r_i'' =
n$. Also, as $a^{\ell_p} \eq a^{\ell_p'}$ and  $a^{\ell_p'} \eq
a^{\ell_p''}$, we have  $a^{\ell_p} \eq a^{\ell_p''}$. Thus, 
by Definition~\ref{def:obj} $\theta_0(a) \eq \theta_n(c)$.
\end{itemize}
If $\theta_0(a)$ and $\theta_{n-1}(b)$ contain a function object,
then for their respective structure labels $\ell_s$ and $\ell_s'$, either:
\begin{itemize}
\item $\ell_s = \ell_s' = H$: If $\ell_s' = H$,
then $\ell_s'' = H$ by Definition~\ref{def:obj}, 
where $\ell_s''$ is the structure label of the function object in
$\theta_n(c)$. Thus, $\theta_0(a) \eq \theta_n(c)$.
\item $\ell_s = \ell_s' = L$: $\ell_s'' = L$ is the structure label of the function object in
$\theta_n(c)$. Thus, $N \eq N''$ from the above result for
objects. The CFGs $f =^\beta f' =^\beta f''$ and the scope chains
$\Sigma \eq \Sigma''$ by Corollary~\ref{cor:Evolution_callstacks}.
Thus, $\theta_0(a) \eq \theta_n(c)$.
\end{itemize}
Thus, $\theta_0 \eq \theta_n$.
\end{proof}

\begin{myLemma}[Supporting Lemma 1]
\label{lem:sup1}
Suppose \\
$\langle\iota, \theta_1, \sigma_1,\rho_1 \rangle~\leadsto~\langle\iota_1',
\theta_1', \sigma_1',\rho_1' \rangle$, \\
$\langle\iota,  \theta_2, \sigma_2,\rho_2 \rangle~\leadsto~\langle\iota_2',
\theta_2', \sigma_2',\rho_2' \rangle$, \\
$\rho_1 \sim \rho_2$, 
$\Gamma(!\rho_1) = \Gamma(!\rho_2) = L$, $\Gamma(!\rho_1') = \Gamma(!\rho_2')$
and $(\sigma_1 \eq_{\rho_1, \rho_2} \sigma_2) \wedge (\theta_1 \eq
\theta_2)$\\
then $\rho_1' \sim \rho_2'$, and $\exists \beta': ((\beta' \supseteq \beta) \wedge (\sigma_1'
\sim^{\beta'}_{\rho_1', \rho_2'} \sigma_2') \wedge (\theta_1'
\sim^{\beta'} \theta_2'))$. 
\end{myLemma}
\begin{proof}
Every instruction executes \emph{isIPD} at the end of the operation. 
If $\iota_i'$ is the IPD corresponding to the $!\rho_i.\ipd$, then it pops
the first node on the pc-stack. As $\rho_1 \sim \rho_2$ and
$\Gamma(!\rho_1)  = \Gamma(!\rho_2)$,
$\iota_i'$ would either pop in both the runs or in none. Thus,
$\rho_1' sim \rho_2'$. For instructions that push (branch), we explain
in respective instructions.\\
Proof by case analysis on the instruction type:
\begin{enumerate}
\item \emph{prim:} No new object is created, so $\beta' = \beta$. \\
  As $\sigma_1 \eq_{\rho_1,\rho_2} \sigma_2$, so $!\sigma_1 \eq\,
  !\sigma_2$ and $!\sigma_1(\src_i) \eq\, !\sigma_2(\src_i)$ for $i
  = 1,2$. Case analysis on the definition of $\eq$ for $\src_i$.
\begin{itemize}
\item If $(\Gamma(!\sigma_1(\src_1)) = \star~\vee$
  $\Gamma(!\sigma_1(\src_2)) = \star~\vee$
  $\Gamma(!\sigma_2(\src_1)) = \star~\vee$
  $\Gamma(!\sigma_2(\src_2)) = \star)$, then 
  $\Gamma(!\sigma_1(\dst)) = \star~\vee \Gamma(!\sigma_1(\dst)) 
  = \star$. Hence, $!\sigma_1(\dst) \eq\,!\sigma_2(\dst)$
  by Definition~\ref{def:lval}.
\item If $\Gamma(!\sigma_1(src_1)) = \Gamma(!\sigma_1(src_2)) = H$ and 
   $\Gamma(!\sigma_2(src_1)) = \Gamma(!\sigma_2(src_2)) = H$
   $\Gamma(!\sigma_1(\dst)) = \Gamma(!\sigma_2(\dst)) = H$. 
   So, $!\sigma_1(\dst) \eq\,!\sigma_2(\dst)$
   by Definition~\ref{def:lval}.
\item If $!\sigma_1(src_1)=\,!\sigma_2(src_1)~\wedge~\Gamma(!\sigma_1(src_2))
  = H~\wedge~\Gamma(!\sigma_2(src_2)) = H $, then 
   $\Gamma(!\sigma_1(\dst)) = \Gamma(!\sigma_2(\dst)) = H$.
  So, $!\sigma_1(\dst) \eq\,!\sigma_2(\dst)$
  by Definition~\ref{def:lval}. \\
  Symmetrical reasoning for
   $!\sigma_1(src_2)=\,!\sigma_2(src_2) ~\wedge~\Gamma(!\sigma_1(src_1))
   = H~\wedge~$ $\Gamma(!\sigma_2(src_1)) = H$.
\item $!\sigma_1(src_1)=\,!\sigma_2(src_1)~\wedge~
  !\sigma_1(src_2) =\,!\sigma_2(src_2) $: \\
  $!\sigma_1(\dst)=\,!\sigma_2(\dst)$. So, $!\sigma_1(\dst) \eq\,!\sigma_2(\dst)$
  by Definition~\ref{def:lval}. 
\end{itemize}
Only $\dst$ changes in the top call-frame of both the call-stacks. 
Thus, by Definition~\ref{def:cf}, $!\sigma'_1 \eq\,!\sigma'_2$. Other call-frames
in $\sigma_1'$ and $\sigma_2'$ are unchanged. By 
Definition~\ref{def:cs}, $\sigma'_1 \sim^{\beta'}_{\rho_1',\rho_2'} \sigma'_2$. \\
$\theta_1 = \theta'_1$ and $\theta_2 = \theta'_2$, so, $\theta'_1 \sim^{\beta'}
\theta'_2$.

\item \emph{mov:} Similar reasoning as prim with single source.

\item \emph{jfalse:} No new object is created, so $\beta' = \beta$.
\begin{itemize}
\item $!\sigma_1(cond)=\,!\sigma_2(cond) \wedge
  \Gamma(!\sigma_1(cond)) = \Gamma(!\sigma_2(cond)) = L$: 
  $L$ is the label to be pushed on $\rho$.
\item $\Gamma(!\sigma_1(cond)) = \Gamma(!\sigma_2(cond)) = H$: 
  $H$ is the label to be pushed on $\rho$.
\end{itemize}
The IPD of $\iota$ would be the same as we have same CFG in both the
cases. If the IPD is SEN, then we join the label of $!\rho_i$ with the
label obtained above, which is the same in both the runs. Thus,
$\Gamma(!\rho_1') = \Gamma(!\rho_2')$. 
Because $\rho_1 \sim \rho_2$, $\rho_1' \sim \rho_2'$. \\
If the IPD is not SEN, then it is some other node in the
same call-frame. Thus the $\ipd$ field is also the same. 
The $\mH$ field is $\false$ in both the
cases. Thus, the pushed node is the same in both the cases and hence,
$\rho_1' \sim \rho_2'$. As, $\Gamma(!\rho_1') = \Gamma(!\rho_2'')$,
either $\iota_1' = \iota_2' = \IPD(\iota)$ or $\iota_1'$ and
$\iota_2'$ may or may not be equal.\\
$\sigma'_1 = \sigma_1 \sim^{\beta'}_{\rho_1',\rho_2'} \sigma_2 = \sigma'_2$.
$\theta'_1 = \theta_1 \sim^{\beta'} \theta_2 = \theta'_2$.

\item \emph{loop-if-less:} Similar reasoning as jfalse.

\item \emph{type-of:} Similar to mov.

\item \emph{instance-of:} No new object is created, so $\beta' = \beta$. \\
  The label of the value in the $\dst$ is the label
  of the context joined with the label of all the prototype chain
  pointers traversed. As $!\sigma_1(\mathit{value}) \eq
  !\sigma_2(\mathit{value})$, where $\ell_s$ and $\ell_s'$ are the
  structure labels of objects pointed to by $!\sigma_1(\mathit{value})$
  and $!\sigma_2(\mathit{value})$ respectively, then by Definition~\ref{def:obj}:
\begin{itemize}
\item If $\ell_s = \ell_s' = H$, then $\Gamma(!\sigma_1'(dst)) = H$
   and  $\Gamma(!\sigma_2'(dst)) = H$. So, $!\sigma_1'(\dst) \eq\,
  !\sigma_2'(\dst)$ from Definition~\ref{def:lval}. 
\item If $\ell_s = \ell_s' = L$, then the objects have similar
  properties and prototype chains. If it is not an instance and none
  of traversed prototype chain and objects are $H$, then
  $\Gamma(!\sigma_1'(dst)) = \Gamma(!\sigma_2'(dst)) = L$ and
  $\false$. Else if it is present it has $\true$. So, $!\sigma_1'(\dst) \eq\,
  !\sigma_2'(\dst)$ from Definition~\ref{def:lval}.
  If any one of traversed prototype chain and objects are $H$, then
  $\Gamma(!\sigma_1'(dst)) = \Gamma(!\sigma_2'(dst)) = H$. 
  So, $!\sigma_1'(\dst) \eq\,
  !\sigma_2'(\dst)$ from Definition~\ref{def:lval}.
\end{itemize}
Only $\dst$ changes in the top call-frame of both the call-stacks. 
Thus, by Definition~\ref{def:cf}, $!\sigma'_1 \eq\,!\sigma'_2$. Other call-frames
are unchanged and by Definition~\ref{def:cs}, $\sigma'_1 \sim^{\beta'}_{\rho_1',\rho_2'} \sigma'_2$. \\
$\theta_1 = \theta'_1$ and $\theta_2 = \theta'_2$, so, $\theta'_1 \sim^{\beta'}
\theta'_2$.

\item \emph{enter:} No new object is created, so $\beta' = \beta$. \\
$\sigma'_1 = \sigma_1 \sim^{\beta'}_{\rho_1',\rho_2'} \sigma_2 = \sigma'_2$.
$\theta'_1 = \theta_1 \sim^{\beta'} \theta_2 = \theta'_2$.

\item \emph{ret:}  No new object is created, so $\beta' = \beta$. \\
  Since $\sigma_1 \eq \sigma_2$ so only two cases arise for the getter
  flag.
  \begin{itemize}
  \item $!\sigma_1.getter = !\sigma_2.getter = false$: $\sigma_1'$ is same as
  $\sigma_1$ with $!\sigma_1$ popped. Similarly, $\sigma_2'$ is same
  as $\sigma_2$ with $!\sigma_2$ popped. As other call-frames are
  unchanged by Definition~\ref{def:cs}, $\sigma_1'
  \sim^{\beta'}_{\rho_1',\rho_2'}
  \sigma_2'$.
\item $!\sigma_1.getter = !\sigma_2.getter = true$: only resgister
  which changes is the $\sigma_1'(!\sigma_1.dReg$ and
  $\sigma_2'(!\sigma_2.dReg)$. Now, if
  $\Gamma(!\sigma_1'(!\sigma_1.dReg)) =
  \Gamma(!\sigma_2'(!\sigma_2.dReg)) = H$ then
  $!\sigma_1'(!\sigma_1.dReg) \eq !\sigma_2'(!\sigma_1.dReg)$ from
  defintion~\ref{def:lval}. And if $\Gamma(!\sigma_1'(!\sigma_1.dReg))
  = \Gamma(!\sigma_2'(!\sigma_2.dReg)) = L$ then
  $!\sigma_1'(!\sigma_1.dReg)) = !\sigma_1(res)
  \sigma_2'(!\sigma_2.dReg) = !\sigma_2(res)$ and $!\sigma_1(res) \eq
  !\sigma_2(res)$.
  \end{itemize}
 $\theta'_1 = \theta_1 \sim^{\beta'} \theta_2 = \theta'_2$.

\item \emph{end:} No $\sigma_i'$ and $\theta_i'$.

\item \emph{call:} No new object is created, so $\beta' = \beta$. \\ 
  Pushes the same node on both $\rho$s (similar to jfalse). The only
difference is the $\mH$ field. As the CFGs are 
same, if it has an associated exception handler, we set the
$\mH$ field to $\true$ in both the runs. Else, it is $\false$.
Thus, $!\rho_1' =\, !\rho_2'$ is the node pushed on $\rho$ and hence,
$\rho_1' \sim \rho_2'$. \\
As $!\sigma_1(func) \eq\, !\sigma_2(func)$, if:
 \begin{itemize}
 \item $(\Gamma(!\rho_1')= H):$
   As call-frames until $!\sigma_1$ and $!\sigma_2$ remain unchanged,
   which correspond to the $\mC$ field in the lowest $H$-labelled
   node and $\sigma_1 \sim^{\beta}_{\rho_1,\rho_2} \sigma_2$, by
   Definition~\ref{def:cs}, $\sigma_1'
   \sim^{\beta'}_{\rho_1',\rho_2'}\sigma_2'$. 
 \item $(\Gamma(!\rho_1') = L):$ Registers created in the new call-frame
   contain \emph{undefined} with label $L$ and, as $\theta_1 \eq
   \theta_2$ so the function objects $N \eq N'$ 
   implying $!\sigma_1'.\CFG = !\sigma_2'.\CFG$ and
   $!\sigma_1'.\Sigma \eq \,!\sigma_2'.\Sigma$, also return addresses
   are the same and the callee is the same. So $!\sigma_1' \eq\,!\sigma_2'$. Other
   call-frames are unchanged so, $\sigma_1' \sim^{\beta'}_{\rho_1',\rho_2'} \sigma_2'$.
 \end{itemize}
$\theta'_1 = \theta_1 \sim^{\beta'} \theta_2 = \theta'_2$.

\item \emph{call-put-result:} Similar to move.

\item \emph{call-eval:}
 $\rho_1' \sim \rho_2'$: Similar to op-call.\\
 In strict mode, it pushes a node on the scope-chain with label
 $L$. The pushed nodes are low-equivalent. Thus, $!\sigma_1'.\Sigma
 \eq \,!\sigma_2'.\Sigma$ by Definition~\ref{def:sc}.
 In non-strict mode, it does not push anything and is similar to
 call. Thus,  $\sigma_1' \sim^{\beta'}_{\rho_1',\rho_2'} \sigma_2'$
 and $\theta'_1 = \theta_1 \sim^{\beta'} \theta_2 = \theta'_2$.

\item \emph{create-arguments:} 
  Let the argument object be created at $x$ and $y$ in $\theta_1$ and
  $\theta_2$, then $\beta' = \beta \cup (x,y)$.
  $\Gamma(\sigma_1'(\dst))=\Gamma(\sigma_1'(\dst))=L$ and as
  $!\sigma_1 \eq !\sigma_2$, the objects are low-equivalent.
  Thus, $!\sigma_1' \sim^{\beta'}\,!\sigma_2$ by
  Definition~\ref{def:cf} and $\sigma_1' \sim^{\beta'}_{\rho_1',\rho_2'}\sigma_2'$ by
  Definition~\ref{def:cs}. Also, $\theta_1' \sim^{\beta'}\theta_2'$
  by Definition~\ref{def:heap} as the objects are low-equivalent.

\item \emph{new-func:} 
  Let the function object be created at $x$ and $y$ in $\theta_1$ and
  $\theta_2$, then $\beta' = \beta \cup (x,y)$. Function objects are
  low-equivalent as $!\sigma_1.\Sigma \eq !\sigma_2.\Sigma$ and
  $!\sigma_1(\mathit{func}) \eq !\sigma_2(\mathit{func})$.
  $\sigma_1'(\dst) \eq \sigma_2'(\dst)$ by Definition~\ref{def:lval}.  
  Thus, $!\sigma_1' \sim^{\beta'}\,!\sigma_2$ by
  Definition~\ref{def:cf} and $\sigma_1' \sim^{\beta'}_{\rho_1',\rho_2'}\sigma_2'$ by
  Definition~\ref{def:cs}. Also, $\theta_1' \sim^{\beta'}\theta_2'$
  by Definition ~\ref{def:heap} as the objects are low-equivalent.


\item \emph{create-activation:} Similar to create-arguments.

\item \emph{construct:} Similar to call.

\item \emph{create-this:} Similar to create-this.

\item \emph{new-object:} Similar to create-arguments.

\item \emph{get-by-id:} No new object is created, so $\beta' = \beta$. \\
  As $!\sigma_1(\mathit{base}) \eq\, !\sigma_2(\mathit{base})$, either
  the objects have the same properties or are labelled $H$ because of
  Definition~\ref{def:obj}. In case of data property, either
  $\Gamma(\sigma_1'(\dst)) = \Gamma(\sigma_2'(\dst)) = H$ or
  $\Gamma(\sigma_1'(\dst)) = \Gamma(\sigma_2'(\dst)) = L$ and value of
  \emph{prop} is the same.  So, by Definition~\ref{def:lval}
  $!\sigma_1'(\dst) \eq\, !\sigma_2'(\dst)$.  
  
  In case of an accessor property, only $\dst$ changes in the top
  call-frame of both the call-stacks, and $\sigma_1'(dst)) \eq
  \sigma_2'(dst))$ since $\sigma_1 \eq \sigma_2$ and $\theta_1 \eq
  \theta_2$ Thus, by Definition~\ref{def:cf}, $!\sigma'_1
  \eq\,!\sigma'_2$. Other call-frames are unchanged and by
  Definition~\ref{def:cs}, $\sigma'_1 \sim^{\beta'}_{\rho_1',\rho_2'}
  \sigma'_2$. For $\rho_1' \sim \rho_2'$, reasoning is similar to
  \emph{call}.  $\theta_1 = \theta'_1$ and $\theta_2 = \theta'_2$, so,
  $\theta'_1 \sim^{\beta'} \theta'_2$.

\item \emph{put-by-id:} No new object is created, so $\beta' = \beta$. \\ 
   $\sigma_1' = \sigma_1 \sim^{\beta'}_{\rho_1',\rho_2'} \sigma_2 = \sigma_2'$. \\
   Because $\sigma_1 \sim^{\beta}_{\rho_1,\rho_2} \sigma_2$, 
   if \emph{value} is labelled $H$, then the properties created or
   modified will have label $H$, and structure labels of the
   respective objects will become $H$. 
   Else if \emph{value} is labelled $L$, then the properties created or
   modified will have same value and label $L$. Thus, the objects
   remain low-equivalent by Definition~\ref{def:obj} and hence, by
   Definition~\ref{def:heap}, $\theta'_1 \sim^{\beta'} \theta'_2$.

\item \emph{del-by-id:} No new object is created, so $\beta' = \beta$. \\
  If the deleted property is $H$ or if the structure label of the object is $H$ , 
  then $(\Gamma(\sigma_1'(\dst)) = \Gamma(\sigma_2'(\dst)) = H)$.
  Else if is labelled $L$, then $(\Gamma(\sigma_1'(\dst)) =
  \Gamma(\sigma_2'(\dst)) = L)$ and value is $\true$ or $\false$
  depending on whether the property is deleted or not. 
  $\sigma'_1 \sim^{\beta'}_{\rho_1',\rho_2'} \sigma'_2$ by Definition~\ref{def:lval},
  ~\ref{def:cf} and~\ref{def:cs}. If structure
  labels of the objects are $L$, they have same properties by
  Definition~\ref{def:obj}. If not, they have structure label as $H$. 
  Thus, objects remain low-equivalent by Definition~\ref{def:obj} and
  $\theta'_1 \sim^{\beta'} \theta'_2$ by Definition~\ref{def:heap}.

\item \emph{put-getter-setter:} Reasoning similar to put-by-id.

\item \emph{get-pnames:} No new object is created, so $\beta' =
  \beta$. \\
As $\sigma_1 \sim^{\beta}_{\rho_1,\rho_2} \sigma_2$, $!\sigma_1(base) \eq \,!\sigma_2(base)$ and
so are the objects (\emph{obj}1 and \emph{obj}2),
$\mathit{obj}1 \eq \mathit{obj}2$, as $\theta_1
\eq \theta_2$. Thus, the structure label of the object
 is either $H$ in both the runs or $L$
and have the same properties with values (Definition~\ref{def:obj}.
The IPD in both the cases is the same and so is the $\mC$ field. The
$mH$ field is set to $\false$. Thus, $\rho_1' \sim \rho_2'$.\\
For $\sigma_1' \sim^{\beta'}_{\rho_1',\rho_2'} \sigma_2'$, it is
similar to mov, but done for dst, i and size.\\
$\theta_1' = \theta_1 \sim^{\beta'} \theta_2 = \theta_2'$.

\item \emph{next-pname:} Similar to mov, but done for dst and base.

\item \emph{resolve:}
  No new object is created, so $\beta' = \beta$. \\
  If property is found in a $L$ object and the scope-chain node labels
  are also $L$, then the property value is the same as $!\sigma_1 \eq
  !\sigma_2$. If it is in $H$ object or any scope-chain node labels
  are $H$ or have a $\star$, then label of the property is $H$ or
  $\star$. Thus, $!\sigma_1'(\dst) \eq\, !\sigma_2'(\dst)$ by
  Definition~\ref{def:lval}. Thus, $\sigma_1' \sim^{\beta'}_{\rho_1',\rho_2'} \sigma_2'$.
  If property is not found in both runs, it is similar to throw. If
  property is not found in second run, then in the first run the property is
  in $H$ context. So, the exception thrown is also $H$. Until, the
  call-frame of $!\rho_2'.\ipd$, call-frames are unchanged, so
  $\sigma_1' \sim^{\beta'}_{\rho_1',\rho_2'} \sigma_2'$.
$\theta_1' = \theta_1 \sim^{\beta'} \theta_2 = \theta_2'$.

\item \emph{resolve-skip:} Similar to resolve.

\item \emph{resolve-global:} Similar to resolve.

\item \emph{resolve-base:} Similar to resolve.

\item \emph{resolve-with-base:} Similar to resolve.

\item \emph{get-scoped-var:} 
No new object is created, so $\beta' = \beta$. \\
 Reads \emph{index}th register in object in \emph{skip}th node in the
 scope-chain  and writes into $\dst$.
 As $(!\sigma_1.\Sigma \eq\,!\sigma_2.\Sigma)$, the value, if labelled $L$ is the same, 
 else is labelled $H$ or $\star$. By Definition~\ref{def:lval}, $!\sigma_1'(\dst)
 \eq\, !\sigma_2'(\dst)$ and by Definition~\ref{def:cf}, $!\sigma_1' \eq\,
 !\sigma_2'$. Thus, $\sigma_1'  \sim^{\beta'}_{\rho_1',\rho_2'} \sigma_2'$.
$\theta_1' = \theta_1 \sim^{\beta'} \theta_2 = \theta_2'$.

\item \emph{put-scoped-var:} 
No new object is created, so $\beta' = \beta$. \\
  Writes into the scope chain node the same value, if ``value'' is
  labelled $L$. If it is labelled $\star$ in any of the runs, scope
  chains remain equivalent. If value is $H$, it checks the label of
  register and puts the value with label $H$ or $\star$.  
  Thus, $(!\sigma_1'.\Sigma \eq\,!\sigma_2'.\Sigma)$ by
  Definition~\ref{def:sc} and $(!\sigma_1' \eq\,!\sigma_2')$ by Definition~\ref{def:cf}.
  By Definition~\ref{def:cs}, $\sigma_1' \sim^{\beta'}_{\rho_1',\rho_2'} \sigma_2'$.
$\theta_1' = \theta_1 \sim^{\beta'} \theta_2 = \theta_2'$.

\item \emph{push-scope:}
No new object is created, so $\beta' = \beta$. \\
Pushes in the scope-chain a node containing the object in
``scope'' with node label $L$. As $(!\sigma_1(scope) \eq\,!\sigma_2(scope))$ and
$(!\sigma_1.\Sigma \eq\,!\sigma_2.\Sigma)$, $(!\sigma_1'.\Sigma
\eq\,!\sigma_2'.\Sigma)$ by Definition~\ref{def:sc}.
 Registers and other call-frames remain the same, so, $\sigma_1'
\sim^{\beta'}_{\rho_1',\rho_2'} \sigma_2'$. 
$\theta_1' = \theta_1 \sim^{\beta'} \theta_2 = \theta_2'$.

\item \emph{pop-scope:}
No new object is created, so $\beta' = \beta$. \\
Pops a node from the scope chain if $\Gamma(!\sigma_1.\Sigma) =
\Gamma(!\sigma_1.\Sigma) \neq (\star \vee H)$. As 
$(!\sigma_1.\Sigma \eq\,!\sigma_2.\Sigma)$, so $(!\sigma_1'.\Sigma
\eq\,!\sigma_2'.\Sigma)$. 
Other registers remain the same, so, $\sigma_1' \sim^{\beta'}_{\rho_1',\rho_2'} \sigma_2'$.
$\theta_1' = \theta_1 \sim^{\beta'} \theta_2 = \theta_2'$.

\item \emph{jmp-scope:}
Similar to pop-scope.

\item \emph{throw:}
No new object is created, so $\beta' = \beta$. \\
  The property of IPD ensures that
  $!\sigma_1' = (!\rho_1.\mC)$ and $!\sigma_2' = (!\rho_2.\mC)$. 
  The $!\sigma_1'$ and $!\sigma_2'$ and the ones below them remain unchanged.
  Thus, $\sigma_1' \sim^{\beta'}_{\rho_1',\rho_2'} \sigma_2'$ by Definition~\ref{def:cs}. \\ 
  $\theta_1' = \theta_1 \sim^{\beta'} \theta_2 = \theta_2'$.

\item \emph{catch:} 
  Similar to mov.

\end{enumerate}
\end{proof}

\begin{myLemma}[Supporting Lemma 2] Suppose
\label{lem:sup2}
\begin{enumerate}
\item $\langle\iota_0, \theta_0', \sigma_0',\rho_0' \rangle~\leadsto~\langle\iota_1',
\theta_1', \sigma_1',\rho_1' \rangle~\leadsto^{n-1}~\langle\iota_n',
 \theta_n', \sigma_n',\rho_n' \rangle$,
\item $\langle\iota_0, \theta_0'', \sigma_0'', \rho_0'' \rangle~\leadsto~\langle\iota_1'',
\theta_1'', \sigma_1'',\rho_1'' \rangle~\leadsto^{m-1}~\langle\iota_m'',
\theta_m'', \sigma_m'',\rho_m'' \rangle $,
\item $(\rho_0'  \sim \rho_0'')$, $(\sigma_0' \eq_{\rho_0',\rho_0''} \sigma_0'')$, $(\theta_0' \eq \theta_0'')$, 
\item $ (\Gamma(!\rho_0')=\Gamma(!\rho_0'')=L)$, $ (\Gamma(!\rho_n')=\Gamma(!\rho_m'')=L)$,
\item $\forall(0 < i < n).(\Gamma(!\rho_i') = H)~\wedge~
\forall(0 < j < m).(\Gamma(!\rho_j'') = H)$,
\end{enumerate}
then $(\iota_n' = \iota_m'')$, $(\rho_n' \sim \rho_m'')$, $(\sigma_n'
\eq_{\rho_n', \rho_m''} \sigma_m'')$, and $(\theta_n' \eq \theta_m'')$.
\end{myLemma}
\begin{proof}
  Starting with the same instruction and high context in both the runs, we
  might get two different instructions, $\iota_1'$ and
  $\iota_1''$. This is only possible if $\iota$ was some branching
  instruction in the first place and this divergence happened in a
  high context. Now,
  \begin{enumerate}
 \item To prove $\iota_n' = \iota_m''$: \\From the property of the IPDs
    we know that if $\iota_0$ pushes a $H$ node on top of pc-stack
    which was originally $L$, $\IPD(\iota_0)$ pops
    that node. Since we start
    from the same instrucion $\iota_0$, $\iota_n' = \iota_m'' =
    \IPD(\iota)$, where $\Gamma(!\rho) = L$. 
  
  \item To prove $\rho_n' \sim \rho_m''$: 
    \begin{itemize}
      \item $n > 1$ and $m > 1$: $\Gamma(!\rho_1') =
        \Gamma(!\rho_1'')$, because $\iota_0$ pushes equal nodes and
        $\iota_1',\iota_1''$ are not the IPDs. As $\rho_0' \sim \rho_0''$ and
    $\Gamma(!\rho_1') = \Gamma(!\rho_1'')$, 
    from Lemma~\ref{lem:sup1} we get
    $\rho_1' \sim \rho_1''$ and $!\rho_1'.\ipd = !\rho_1''.\ipd =
    \IPD(\iota_0)$, if $\iota_1' \neq \IPD(\iota_0)$ and $\iota_1''
    \neq \IPD(\iota_0)$. As $\iota_n' = \iota_m'' = \IPD(\iota_0)$, it pops
    the $!\rho_1'$ and $!\rho_1''$, which correspond to $\rho_n'$ and $\rho_m''$ in
    the $n$th and $m$th step. (IPD is the point where we pop the final
    $H$ node on the pc-stack.)
    Because $\rho_1' \sim \rho_1''$ and from
    Corollary~\ref{cor:Evolution_callstacks}, $\rho_n' \sim \rho_m''$.
      \item $n = 1$ and $m > 1$: If $\Gamma(!\rho_1') \neq \Gamma(!\rho_1'')$ and $\iota_1' =
    \IPD(\iota_0)$, then $\Gamma(!\rho_1') = L$. 
     It pops the node pushed by $\iota_0$,
    i.e., $\Gamma(!\rho_n') = L$. In the other run as
    $\Gamma(!\rho_1'')=H$ and $\Gamma(!\rho_m'')
    = L$, by the property of IPD $\iota_m'' = \IPD(\iota_0)$, which
    would pop from the pc-stack $!\rho_1''$, the first frame labelled
    $H$ on the pc-stack. Thus, $\rho_n' \sim \rho_m''$.
       \item $n > 1$ and $m = 1$: Symmetric case of the above. 
      \end{itemize}
  \item To prove $\sigma_n' \eq_{\rho_n', \rho_m''} \sigma_m''$: 
    \begin{enumerate}
      \item $n > 1$ and $m > 1$: From Lemma~\ref{lem:sup1}
    we get $\sigma_1' \eq_{\rho_1', \rho_1''} \sigma_1''$. From
    Corollary~\ref{cor:Evolution_callstacks} we get $\sigma_1' \eq_{\rho_1', \rho_{n-1}'}
    \sigma_{n-1}'$. And from Lemma~\ref{lem:confinement} we have
    $\sigma_{n-1}' \eq_{\rho_{n-1}', \rho_n'} \sigma_n'$. 
    As, $\iota_n' = \iota_m'' = \IPD(\iota_0)$, we compare all
    call-frames of $\sigma_n'$ and $\sigma_m''$. As the IPD of an
    instruction can lie only in the same call-frame, comparison for all
  call-frames in $\sigma_0'$ and $\sigma_0''$ suffice. \\
  $\sigma_0' \eq_{\rho_0', \rho_0''} \sigma_0'' \Rightarrow \forall i.((\mu_i
  \in \sigma_0' \wedge \nu_i \in \sigma_0''),  (\mu_i \eq \nu_i)
  \wedge \forall( r \in \mu_i,\,\nu_i). (\mu_i(r)=v_1 \wedge \nu_i(r) =
  v_2,\, v_1 \eq v_2))$. \\
  Let $v_1$ and $v_2$ be represented by
  $v_n$ and $v_m$ in $\sigma_n'$ and $\sigma_m''$, respectively. The
  call-frames in $\sigma_n'$ and $\sigma_m''$ are represented by
  $\mu_n$ and $\nu_m$, respectively.
  \begin{itemize}
    \item We do case analysis on the different cases of
      Definitions~\ref{def:lval} for $v_1$ and
      $v_2$, to show $v_n \eq v_m$.
      As $(\forall(1 \leq i < n).(\Gamma(!\rho_i') = H)~\wedge~
      \forall(1 \leq j < m).(\Gamma(!\rho_j'') = H))$: 
      \begin{itemize}
        \item If $v_1 = v_2 \wedge \Gamma(v_1) = \Gamma(v_2) = L$, then
          either $v_n = v_m$ or $((\star = \Gamma(v_n)) \vee (\star =
          \Gamma(v_m)))$. By Definition~\ref{def:lval}(1), $v_n \eq v_m$.
        \item If $\Gamma(v_1) = \Gamma(v_2) = H$, then $\Gamma(v_n) =
          \Gamma(v_m) = H$. By Definition~\ref{def:lval}(2), $v_n \eq v_m$.
        \item If $\star = \Gamma(v_1)$, then $\star = \Gamma(v_n)$ and
          if $\star = \Gamma(v_2)$, then $\star = \Gamma(v_m)$. By
          Definition~\ref{def:lval}(1), $v_n \eq v_m$.
      \end{itemize}

    \item Lets $S_1$ and $S_2$ be the scopechains in $\sigma_0'$ and
      $\sigma_0''$. And $S_n$ and $S_m$ represent the scopechains in
      $\sigma_n'$ and $\sigma_m''$, i.e., $S_n$ and $S_m$ are the
      respective scope-chains in the $n$th and $m$th
     step of the two runs and $\ell_n$ and $\ell_m$ are their
     node labels. For scope chain pointers the
      following cases arise: 
      \begin{enumerate}
      \item $S_1 = S_2 = nil$: In this case $S_n$ and $S_m$ either remain nil or
        its head will have a $H$ label, because of the rules of the
        instructions that modify the scope-chain.
      \item $S_1 = (s_1,\ell_1):\Sigma_1 $ and $ S_2 =
        (s_2,\ell_2):\Sigma_2$:
        \begin{enumerate}
            \item $\ell_1 = \star \vee \ell_{2} = \star$: In this case $\ell_n$
              and $\ell_m$ will be $\star$ too.
            \item $\ell_1 = \ell_{2} = H$:  In this case $\ell_n$
              and $\ell_m$ will be H too.
            \item $\ell_1 = \ell_2 = L \wedge \Sigma_1 \eq \Sigma_2$: In
            this case $\ell_n = \star$ and $\ell_m = \star$  or scopechains
            remain unchanged.
        \end{enumerate}
      \end{enumerate}
  \end{itemize}
      \item $n = 1$ and $m > 1$: \\
        In case of \emph{jfalse} and \emph{loop-if-less}, $\sigma_0' =
        \sigma_1'$ and $\sigma_0'' = \sigma_1''$. And in case of
        \emph{get-pnames}, if $n =1$ and $m \neq 1$, $\Upsilon(!\sigma_0'(base))
        = \mathit{undefined}$ and $\Upsilon(!\sigma_0''(base))
        \neq \mathit{undefined}$. Because $\sigma_0' \eq \sigma_0''$,
        $!\sigma_0'(base) \eq !\sigma_0''(base)$. Hence,
        $\Gamma(!\sigma_0'(base)) = \Gamma(!\sigma_0''(base)) =
        H$. Thus, $\Gamma(!\sigma_1'(\dst)) = \Gamma(!\sigma_1'(i))
        =\Gamma(!\sigma_1'(size)) = H$ and similarly, $\Gamma(!\sigma_1''(\dst)) = \Gamma(!\sigma_1''(i))
        =\Gamma(!\sigma_1''(size)) = H$. Other registers remain
        unchanged and so do the other call-frames. Thus,  $\sigma_1'
        \eq_{\rho_1', \rho_1''} \sigma_1''$. From the case $(a)$
        above, we know that if $\sigma_1' \eq_{\rho_1', \rho_1''}
        \sigma_1''$, then $\sigma_n' \eq_{\rho_n', \rho_m''} \sigma_m''$.
      \item $n > 1$ and $m = 1$: Symmetric case of the above. 
  \end{enumerate}
  \item To prove $\theta_n' \eq \theta_m''$: 
    \begin{enumerate}
       \item $n > 1$ and $m > 1$: 
   From Lemma~\ref{lem:sup1}
    we get $\theta_1' \eq \theta_1''$. From
    Corollary~\ref{cor:Evolution_heap} we get $\theta_1' \eq
    \theta_{n-1}'$. And from Lemma~\ref{lem:confinement} we have
    $\theta_{n-1}' \eq \theta_n'$. 
   Assume $O_1$ is an object at $x$ in $\theta_1'$ and $O_2$ is an
   object at $y$ in $\theta_1''$, such that $(x,y) \in \beta$ and
   $O_n$ and $O_m$ are the respective objects in the $n$th and $m$th
   step of the two runs. 
    We do case analysis on the different cases of
      Definitions~\ref{def:obj} for $O_1$ and
      $O_2$, to show $O_n \eq O_m$.
      \begin{itemize}
        \item If $\Gamma(O_1) = \Gamma(O_2) = H$, then $\Gamma(O_n) =
          \Gamma(O_m) = H$. By Definition~\ref{def:obj}, $O_n \eq O_m$.
        \item If $O_1 = O_2 \wedge \Gamma(O_1) = \Gamma(O_2) = L$, then
          $O_n = O_m  \wedge \Gamma(O_n) = \Gamma(O_m) = L$. 
          By Definition~\ref{def:obj}, $O_n \eq O_m$.
      \end{itemize}
     Similarly, for function objects, the structure labels would remain $H$ if
     they were originally $H$ or will remain $L$ with the same CFGs
     and scope-chains. 
   \item $n = 1$ and $m > 1$: In case of \emph{jfalse},
     \emph{loop-if-less}, \emph{get-pnames}, $\theta_0' =
        \theta_1'$ and $\theta_0'' = \theta_1''$. Thus, $\theta_1' \eq
        \theta_1''$. From the case $(a)$ above, we know that if  $\theta_1' \eq
        \theta_1''$, then  $\theta_n' \eq \theta_m''$.
   \item $n > 1$ and $m = 1$: Symmetric case of the above. 
   \end{enumerate}

  \end{enumerate}
  
\end{proof}


\begin{mydef}[Trace]
\label{def:trace}
A trace is defined as a sequence of configurations or states resulting from a
program evaluation, i.e., for a program evaluation $\mathcal{P} = s_1 \leadsto s_2
\leadsto \ldots \leadsto s_n$ where $s_i = \langle
\iota_i, \theta_i, \sigma_i,\rho_i\rangle$, the corresponding trace is given
as $\mathcal{T}(\mathcal{P}) := s_1 :: s_2 :: \ldots :: s_n$.
\end{mydef}


\begin{mydef}[Epoch-trace]
\label{def:etrace}
An epoch-trace ($\mE$) over a trace $\mathcal{T} = s_1 :: s_2 :: \ldots :: s_n$ where $s_i = \langle
\iota_i, \theta_i, \sigma_i,\rho_i\rangle$  is defined inductively as:
\begin{align*}
  \mE(nil) &:=~nil \\
  \mE(s_i :: \mathcal{T}) &:=
  \begin{cases}
   s_i :: \mE(\mathcal{T})~ & \mathit{if}~~~\Gamma(!\rho_i) = L,\\
   \mE(\mathcal{T}) & \mathit{else~if} ~~\Gamma(!\rho_i) = H.
  \end{cases}
\end{align*}
\end{mydef}



\begin{theorem}[Termination-Insensitive Non-interference]
\label{thm:ni} 
\begin{tabbing}
Suppose $\mathcal{P}$ and  $\mathcal{P}'$ are two program
evaluations. \\
Then for their respective epoch-traces given by:\\
\hspace{5 mm} $\mE(\mathcal{T}(\mathcal{P})) = s_1 :: s_2 :: \ldots ::
s_n$,\\
\hspace{5 mm} $\mE(\mathcal{T}(\mathcal{P}')) = s_1' :: s_2' :: \ldots ::
s_m'$,\\
if $s_1 \eq s_1'$ and $n \leq m$,\\
then\\
\hspace{5 mm} $\exists \beta_n \supseteq \beta:~ s_n \sim^{\beta_n} s_n'$
\end{tabbing}
\end{theorem}
\begin{proof}
Proof proceeds by induction on n.\\
Basis: $s_1 \eq s_1'$, by assumption. \\
IH: $s_k \sim^{\beta_k} s_k'$ where $\beta_k \supseteq \beta$. \\
To prove: $\exists \beta_{k+1} \supseteq \beta:~ s_{k+1}
\sim^{\beta_{k+1}} s_{k+1}'$. \\
Let $s_k \leadsto_i s_{k+1}$ and $s_k \leadsto_{i'} s'_{k+1}$, then:
\begin{itemize}
\item $i = i' = 1$: From Lemma~\ref{lem:sup1}, $s_{k+1}
  \sim^{\beta_{k+1}} s_{k+1}'$ where $\beta_{k+1} \supseteq \beta$.
\item $i > 1$ or $i' > 1$: From Lemma~\ref{lem:sup2}, $s_{k+1}
  \sim^{\beta_{k+1}} s_{k+1}'$ where $\beta_{k+1} = \beta$.
\end{itemize}
\end{proof}

\begin{myCor}
\label{cor:ni}
Suppose:
\begin{enumerate}
\item $\langle \iota_1, \theta_1, \sigma_1, \rho_1\rangle \sim^\beta
\langle\iota_2, \theta_2, \sigma_2, \rho_2\rangle$
\item $\langle \iota_1, \theta_1,\sigma_1,\rho_1\rangle
  \leadsto^{*} \langle \texttt{end} ,\theta_1', [],[]\rangle$
\item $\langle \iota_2, \theta_2,\sigma_2,\rho_2\rangle \leadsto^{*}
  \langle \texttt{end},\theta_2', [],[]\rangle$
\end{enumerate}
Then, $\exists \beta' \supseteq \beta$ such that $\theta_1'
\sim^{\beta'} \theta_2'$.
\end{myCor}
\begin{proof}
$\sigma_1, \sigma_2$ and $\rho_1, \rho_2$ are empty at the end of $*$
steps. From the semantics, we know that in $L$ context both runs would
push and pop the same number of nodes. Thus, both take same number of
steps in $L$ context. Let $k$ be the number of states in $L$
context. Then in Theorem~\ref{thm:ni}, $n = m = k$. Thus, $s_k
\sim^{\beta_k} s_k'$, where $s_k = \langle \mathit{end}, \theta_1',
[], []\rangle$ and $s_k' = \langle \mathit{end}, \theta_2',
[], []\rangle$. By Definition~\ref{def:ceq}, $\theta_1' \sim^{\beta'}
\theta_2'$, where $\beta' = \beta_k$.
\end{proof}


\begin{thebibliography}{10}

\bibitem{richards11ECOOP}
Richards, G., Hammer, C., Burg, B., Vitek, J.:
\newblock The eval that men do -- a large-scale study of the use of eval in
  {J}ava{S}cript applications.
\newblock In Mezzini, M., ed.: ECOOP '11. Volume 6813 of LNCS. (2011)  52--78

\bibitem{jang10CCS}
Jang, D., Jhala, R., Lerner, S., Shacham, H.:
\newblock An empirical study of privacy-violating information flows in
  {J}ava{S}cript web applications.
\newblock In: Proc. 17th ACM Conference on Computer and Communications
  Security. (2010)  270--283

\bibitem{oopsla13}
Richards, G., Hammer, C., Zappa~Nardelli, F., Jagannathan, S., Vitek, J.:
\newblock Flexible access control for javascript.
\newblock In: Proc. 2013 ACM SIGPLAN International Conference on Object
  Oriented Programming Systems Languages \& Applications. OOPSLA '13 (2013)
  305--322

\bibitem{csf12}
Hedin, D., Sabelfeld, A.:
\newblock Information-flow security for a core of {J}ava{S}cript.
\newblock In: Proc. 25th IEEE Computer Security Foundations Symposium. (2012)
  3--18

\bibitem{bello13PhD}
Hedin, D., Birgisson, A., Bello, L., Sabelfeld, A.:
\newblock J{S}{F}low: Tracking information flow in {J}ava{S}cript and its
  {A}{P}{I}s.
\newblock In: Proc. 29th ACM Symposium on Applied Computing. (2014)

\bibitem{SME}
Devriese, D., Piessens, F.:
\newblock Noninterference through secure multi-execution.
\newblock In: Proc. 2010 IEEE Symposium on Security and Privacy. (2010)
  109--124

\bibitem{flowfox}
De~Groef, W., Devriese, D., Nikiforakis, N., Piessens, F.:
\newblock Flowfox: a web browser with flexible and precise information flow
  control.
\newblock In: Proc. 2012 ACM Conference on Computer and Communications
  Security. (2012)  748--759

\bibitem{goguen}
Goguen, J.A., Meseguer, J.:
\newblock Security policies and security models.
\newblock In: Proc. 1982 IEEE Symposium on Security and Privacy. (1982)  11--20

\bibitem{declass1}
Myers, A.C., Liskov, B.:
\newblock A decentralized model for information flow control.
\newblock In: Proc. 16th ACM Symposium on Operating Systems Principles. (1997)
  129--142

\bibitem{declass2}
Zdancewic, S., Myers, A.C.:
\newblock Robust declassification.
\newblock In: Proc. 14th IEEE Computer Security Foundations Workshop. (2001)
  15--23

\bibitem{volpano}
Volpano, D., Irvine, C., Smith, G.:
\newblock A sound type system for secure flow analysis.
\newblock J. Comput. Secur. \textbf{4}(2-3) (January 1996)  167--187

\bibitem{just11PLASTIC}
Just, S., Cleary, A., Shirley, B., Hammer, C.:
\newblock Information flow analysis for {J}ava{S}cript.
\newblock In: Proc. 1st ACM SIGPLAN International Workshop on Programming
  Language and Systems Technologies for Internet Clients. (2011)  9--18

\bibitem{plas10}
Austin, T.H., Flanagan, C.:
\newblock Permissive dynamic information flow analysis.
\newblock In: Proc. 5th ACM SIGPLAN Workshop on Programming Languages and
  Analysis for Security. (2010)  3:1--3:12

\bibitem{rni}
Bohannon, A., Pierce, B.C., Sj\"{o}berg, V., Weirich, S., Zdancewic, S.:
\newblock Reactive noninterference.
\newblock In: Proc. 16th ACM Conference on Computer and Communications
  Security. (2009)  79--90

\bibitem{taly}
Maffeis, S., Mitchell, J.C., Taly, A.:
\newblock An operational semantics for {J}ava{S}cript.
\newblock In: Proc. 6th Asian Symposium on Programming Languages and Systems.
  APLAS '08 (2008)  307--325

\bibitem{Guha}
Guha, A., Saftoiu, C., Krishnamurthi, S.:
\newblock The essence of {J}ava{S}cript.
\newblock In: Proc. 24th European Conference on Object-Oriented Programming.
  (2010)  126--150

\bibitem{S5}
Politz, J.G., Carroll, M.J., Lerner, B.S., Pombrio, J., Krishnamurthi, S.:
\newblock A tested semantics for getters, setters, and eval in {J}ava{S}cript.
\newblock In: Proceedings of the 8th Dynamic Languages Symposium. (2012)  1--16

\bibitem{popl14}
Bodin, M., Chargueraud, A., Filaretti, D., Gardner, P., Maffeis, S.,
  Naudziuniene, D., Schmitt, A., Smith, G.:
\newblock A trusted mechanised javascript specification.
\newblock In: Proc. 41st ACM SIGPLAN-SIGACT Symposium on Principles of
  Programming Languages. (2014)

\bibitem{guarnieri11ISSTA}
Guarnieri, S., Pistoia, M., Tripp, O., Dolby, J., Teilhet, S., Berg, R.:
\newblock Saving the world wide web from vulnerable javascript.
\newblock In: Proc. 2011 International Symposium on Software Testing and
  Analysis. ISSTA '11 (2011)  177--187

\bibitem{stagedIFC}
Chugh, R., Meister, J.A., Jhala, R., Lerner, S.:
\newblock Staged information flow for {J}ava{S}cript.
\newblock In: Proc. 2009 ACM SIGPLAN Conference on Programming Language Design
  and Implementation. (2009)  50--62

\bibitem{plas09}
Austin, T.H., Flanagan, C.:
\newblock Efficient purely-dynamic information flow analysis.
\newblock In: Proc. ACM SIGPLAN Fourth Workshop on Programming Languages and
  Analysis for Security. (2009)  113--124

\bibitem{zdancewic02PhD}
Zdancewic, S.A.:
\newblock Programming Languages for Information Security.
\newblock PhD thesis, Cornell University (August 2002)

\bibitem{esorics12}
Birgisson, A., Hedin, D., Sabelfeld, A.:
\newblock Boosting the permissiveness of dynamic information-flow tracking by
  testing.
\newblock In: Computer Security – ESORICS 2012. Volume 7459 of LNCS.
\newblock Springer Berlin Heidelberg (2012)  55--72

\bibitem{austin12POPL}
Austin, T.H., Flanagan, C.:
\newblock Multiple facets for dynamic information flow.
\newblock In: Proc. 39th annual ACM SIGPLAN-SIGACT Symposium on Principles of
  Programming Languages. (2012)  165--178

\bibitem{rnib}
Bielova, N., Devriese, D., Massacci, F., Piessens, F.:
\newblock Reactive non-interference for a browser model.
\newblock In: 5th International Conference on Network and System Security
  (NSS),. (2011)  97--104

\bibitem{FWF}
Bohannon, A., Pierce, B.C.:
\newblock Featherweight firefox: formalizing the core of a web browser.
\newblock In: Proc. 2010 USENIX conference on Web application development.
  WebApps'10 (2010)  11--22

\bibitem{denning}
Denning, D.E.:
\newblock A lattice model of secure information flow.
\newblock Commun. ACM \textbf{19}(5) (May 1976)  236--243

\bibitem{acsac09}
Dhawan, M., Ganapathy, V.:
\newblock Analyzing information flow in {J}ava{S}cript-based browser
  extensions.
\newblock In: Proc. 2009 Annual Computer Security Applications Conference.
  ACSAC '09 (2009)  382--391

\bibitem{denning82}
Robling~Denning, D.E.:
\newblock Cryptography and Data Security.
\newblock Addison-Wesley Longman Publishing Co., Inc., Boston, MA, USA (1982)

\bibitem{Xin}
Xin, B., Zhang, X.:
\newblock Efficient online detection of dynamic control dependence.
\newblock In: Proc. 2007 International Symposium on Software Testing and
  Analysis. (2007)  185--195

\bibitem{Masri}
Masri, W., Podgurski, A.:
\newblock Algorithms and tool support for dynamic information flow analysis.
\newblock Information {\&} Software Technology \textbf{51}(2) (2009)  385--404

\bibitem{Lengauer}
Lengauer, T., Tarjan, R.E.:
\newblock A fast algorithm for finding dominators in a flowgraph.
\newblock ACM Trans. Program. Lang. Syst. \textbf{1}(1) (January 1979)
  121--141

\bibitem{JSbench}
Richards, G., Gal, A., Eich, B., Vitek, J.:
\newblock Automated construction of javascript benchmarks.
\newblock In: Proceedings of the 2011 ACM International Conference on Object
  Oriented Programming Systems Languages and Applications. (2011)  677--694

\end{thebibliography}
\end{document}